%% file: taga.tex
\documentclass[conference]{IEEEtran}

\usepackage{graphicx}
\usepackage{tikz}
\usetikzlibrary{shapes}
\usepackage{environ}
\usetikzlibrary{decorations.pathreplacing}
\usepackage{todonotes}
\usepackage{amsmath,amsthm,amssymb}

\makeatletter
\newsavebox{\measure@tikzpicture}
\NewEnviron{scaletikzpicturetowidth}[1]{%
  \def\tikz@width{#1}%
  \def\tikzscale{1}\begin{lrbox}{\measure@tikzpicture}%
  \BODY
  \end{lrbox}%
  \pgfmathparse{#1/\wd\measure@tikzpicture}%
  \edef\tikzscale{\pgfmathresult}%
  \BODY
}
\makeatother

\ifCLASSINFOpdf
\else
\fi
\ifCLASSOPTIONcompsoc
  \usepackage[caption=false,font=normalsize,labelfont=sf,textfont=sf]{subfig}
\else
  \usepackage[caption=false,font=footnotesize]{subfig}
\fi
\begin{document}
\input{macros-thm}

\input{macros}
%
\title{Security and Resilience for TAGA: 
       a Touch and Go Assistant in the Aerospace Domain}

\author{\IEEEauthorblockN{Sibylle Fr{\"o}schle\IEEEauthorrefmark{1},
Martin Kubisch\IEEEauthorrefmark{3}, 
Marlon Gr{\"a}fing\IEEEauthorrefmark{5} 
\IEEEauthorblockA{\IEEEauthorrefmark{1}OFFIS Institute 
of Informatics
\& University of Oldenburg, Oldenburg,
Germany, 
Email: sibylle.froeschle@offis.de}
\IEEEauthorblockA{\IEEEauthorrefmark{3}Airbus, Munich, Germany,
Email: martin.kubisch@airbus.com}
\IEEEauthorblockA{\IEEEauthorrefmark{1}OFFIS Institute of Informatics,
Oldenburg, Germany, 
Email: marlon.graefing@offis.de}
}}

\maketitle

\begin{abstract}
\input{abstract}

\end{abstract}


%
\IEEEpeerreviewmaketitle


\input{intron}

\input{usecase/usecase}

\input{setup-2/setup}

\input{plain-2/plain}

\input{auth-2/auth}

\input{resil/resil}

\appendices
\input{related}

\input{app-protocols}

\input{app-fuel}

\input{plain-2/app-plain}

\input{auth-2/app-auth}

\input{resil/app-resil}

\removed{
\input{threats/threats}
\input{goals/goals}
\input{auth/auth}
}

\removed{
\appendices
\input{goals/app}
\input{plain/plain-app}
\input{auth/app}
\section{Relating to Section~\ref{s:plain}}
\input{auth/app}
}
\section*{Acknowledgment}
This work has been conducted within the ENABLE-S3 project that has received funding from  the ECSEL Joint Undertaking grant agreement no 692455. This joint undertaking receives support from the European Union's Horizon 2020 Research and Innovation Programme and Austria, Denmark, Germany, Finland, Czech Republic, Italy, Spain, Portugal, Poland, Ireland, Belgium, France, Netherlands, United Kingdom, Slovakia, and Norway.



\bibliographystyle{IEEEtran}
\bibliography{taga}
%


\end{document}

%% file: macros-thm.tex
\newtheorem{theorem}{Theorem}
\newtheorem{corollary}{Corollary}
\newtheorem{lemma}{Lemma}
\newtheorem{principle}{Principle}
\newtheorem{definition}{Definition}
\newtheorem{proposition}{Proposition}
\newtheorem{hypothesis}{Hypothesis}
\newtheorem{conjecture}{Conjecture}

\theoremstyle{remark}
\newtheorem{attack}{Attack}{\itshape}{\normalfont}
\newtheorem{measure}{Measure}{\bfseries}{\normalfont}
\newtheorem{assumption}{Assumption}{\itshape}{\normalfont}
\newtheorem{example}{Example}{\itshape}{\normalfont}
\newtheorem{fact}{Fact}{\itshape}{\normalfont}

\newtheorem{design}{Design}{\bfseries}{\normalfont}
\newtheorem{notation}{Notation}{\bfseries}{\normalfont}

%% file: macros.tex
\newcommand{\removed}[1]{}

\removed{
\spnewtheorem{term}{Terminology}{\bfseries}{\itshape}
\spnewtheorem{assumption}{Assumption}{\bfseries}{\itshape}
\spnewtheorem{principle}{Principle}{\bfseries}{\itshape}
\spnewtheorem{goal}{Goal}{\bfseries}{\itshape}
\spnewtheorem{attack}{Attack}{\itshape}{\normalfont}
\spnewtheorem{prattack}{Attack on Protocol}{\itshape}{\normalfont}
\spnewtheorem{hlattack}{Overall Attack}{\itshape}{\normalfont}
\spnewtheorem{measure}{Measure}{\itshape}{\normalfont}
\spnewtheorem{erule}{Rule}{\bfseries}{\normalfont}
}

\newcommand{\RGU}{\mathit{R_\mathit{GU}}}
\newcommand{\RAC}{\mathit{R_\mathit{AC}}}

\newcommand{\mcR}{\mathcal{R}}
\newcommand{\dom}{\mathit{dom}}
\newcommand{\nonDom}{\neg \dom}
\newcommand{\ACs}{\mathit{ACs}}

\newcommand{\LTKCcat}{LTKC-CAT}
\newcommand{\LTKCpcard}{LTKC-pcard}
\newcommand{\LTKCgcard}{LTKC-gcard}
\newcommand{\LTKCscard}{LTKC-pcard}
\newcommand{\LTKClcard}{LTKC-gcard}

\newcommand{\Op}{\mathit{Op}}
\newcommand{\statesG}{\mathit{AStates}}
\newcommand{\statesA}{\mathit{AStates}}
\newcommand{\cs}{\mathit{cs}}

\newcommand{\cid}{\mathit{cid}}
\newcommand{\DP}{\mathit{DP}}
\newcommand{\mcP}{\mathcal{P}}
\newcommand{\mcU}{\mathcal{U}}

\newcommand{\att}{\mathit{att}}
\newcommand{\okay}{\mathit{okay}}
\newcommand{\none}{\mathit{none}}
\newcommand{\undef}{\mathit{undef}}
\newcommand{\mmax}{\mathit{max}}
\newcommand{\mmin}{\mathit{min}}
\newcommand{\mstop}{\mathit{stop}}
\newcommand{\mwait}{\mathit{wait}}

\newcommand{\mmX}{\mathit{mism\_x}}
\newcommand{\mmLoc}{\mathit{mism\_l}}
\newcommand{\mmSer}{\mathit{mism\_s}}
\newcommand{\mmLocSer}{\mathit{mism\_ls}}

\newcommand{\OP}{\mathit{OP}}
\newcommand{\NOP}{N_\mathit{OP}}
\newcommand{\NGU}{N_\mathit{GU}}

\newcommand{\actor}{\mathit{actor}}
\newcommand{\Rx}{\mathit{\delta_x}}
\newcommand{\Cy}{\mathit{\gamma_y}}
\newcommand{\cst}{\mathit{st}}
\newcommand{\GC}{\mathit{GC}}
\newcommand{\IA}{\mathit{IA}}
\newcommand{\GI}{\mathit{GI}}
\newcommand{\GpA}{\mathit{G'A}}
\newcommand{\card}{\mathit{card}}
\newcommand{\phone}{\mathit{phone}}
\newcommand{\keyReveal}{\mathit{key\_reveal}}
\newcommand{\unaware}{\mathit{unaware}}
\newcommand{\fail}{\mathit{fail}}
\newcommand{\dataMismatch}{\mathit{data\_mismatch}}
\newcommand{\mismatch}{\mathit{mismatch}}
\newcommand{\mv}{\mathit{mv}}

\newcommand{\nfcl}[2]{#1{\leftrightarrow}#2}
\newcommand{\tapl}[2]{#1\,{\approx}\,#2}
\newcommand{\reader}{\mathit{reader}}
\newcommand{\Wx}{\tau_\mathit{x}}
\newcommand{\Gx}{\mathit{Gx}}
\newcommand{\Cx}{\mathit{Cx}}
\newcommand{\Wy}{\mathit{\tau_y}}

\newcommand{\partner}{\mathit{partner}}
\newcommand{\Gini}{\gamma_\mathit{ini}}
\newcommand{\Gmid}{\gamma_\mathit{mid}}
\newcommand{\Gfin}{\gamma_\mathit{fin}}
\newcommand{\Amid}{\alpha_\mathit{mid}}
\newcommand{\Aac}{\alpha_\mathit{ac}}
\newcommand{\Cac}{\sigma_\mathit{ac}}
\newcommand{\Wac}{\tau_\mathit{ac}}
\newcommand{\Gac}{\mathit{G}_\mathit{ac}}
\newcommand{\Winiac}{\tau^\mathit{ini}_\mathit{ac}}
\newcommand{\Giniac}{\gamma^\mathit{ini}_\mathit{ac}}
\newcommand{\Ciniac}{\sigma^\mathit{ini}_\mathit{ac}}
\newcommand{\Wini}{\tau_\mathit{ini}}
\newcommand{\Wfin}{\tau_\mathit{fin}}
\newcommand{\Wmid}{\tau_\mathit{mid}}
\newcommand{\Cini}{\sigma_\mathit{ini}}
\newcommand{\Cfin}{\sigma_\mathit{fin}}
\newcommand{\Cmid}{\sigma_\mathit{mid}}
\newcommand{\msgr}{\mathit{msg\mbox{-}r}}
\newcommand{\msgs}{\mathit{msg\mbox{-}s}}

\newcommand{\pos}{\mathit{pos}}
\newcommand{\CLS}{$_\mathit{\!CLS}$}
\newcommand{\bDH}{\mathit{DH}}
\newcommand{\GU}{\mathit{GU}}
\newcommand{\GO}{\mathit{GM}}
\newcommand{\GM}{\mathit{GM}}
\newcommand{\AO}{\mathit{AO}}
\newcommand{\SO}{\mathit{SO}}
\newcommand{\ST}{\mathit{ST}}
\newcommand{\SE}{\mathit{SE}}
\newcommand{\cert}{\mathit{cert}}
\newcommand{\wifi}{\mathit{wifi}}

\newcommand{\LTK}{\mathit{LTK}}
\newcommand{\ltk}{\mathit{ltk}}
\newcommand{\TKE}{{\bfseries TKE}}
\newcommand{\CRA}{{\bfseries CRA}}
\newcommand{\RCD}{{\bfseries RCD}}

\newcommand{\nread}{\mathit{read}}
\newcommand{\IMac}{\mathit{\scriptsize IMac}}
\newcommand{\ICrypto}{\mathit{\scriptsize ICrypt}}

\newcommand{\mcT}{\mathcal{T}}
\newcommand{\mcF}{\mathcal{F}}
\newcommand{\mcH}{\mathcal{H}}
\newcommand{\mcA}{\mathcal{A}}
\newcommand{\mcG}{\mathcal{G}}
\newcommand{\transcript}{\mathit{transcript}}
\newcommand{\KDE}{\mathit{KDE}}
\newcommand{\KDF}{\mathit{KDF}}
\newcommand{\type}{\mathit{type}}
\newcommand{\ID}{\mathit{ID}}
\newcommand{\SM}{\mathit{SM}}
\newcommand{\devID}{\mathit{devID}}
\newcommand{\devType}{\mathit{devType}}
\newcommand{\ssid}{\mathit{ssid}}
\newcommand{\freq}{\mathit{freq}}
\newcommand{\ccar}{\longrightarrow_C}
\newcommand{\cnfc}{\longrightarrow_N}
\newcommand{\cwlan}{\longrightarrow_W}
\newcommand{\AP}{\mathit{AP}}
\newcommand{\MU}{\mathit{MU}}

\newcommand{\key}{\mathit{key}}
\newcommand{\msg}{\mathit{msg}}
\newcommand{\params}{\mathit{par}}
\newcommand{\sign}{\mathit{sign}}
\newcommand{\tsa}{\mathit{tsa}}
\newcommand{\pcA}{\mathit{pcA}}
\newcommand{\pcG}{\mathit{pcG}}
\newcommand{\pcGS}{\mathit{pcGS}}
\newcommand{\loc}{\mathit{loc}}
\newcommand{\mac}{\mathit{mac}}
\newcommand{\sig}{\mathit{sig}}
\newcommand{\LTKC}{\mathit{LTKC}}

\newcommand{\AC}{\mathit{AC}}
\newcommand{\GUs}{\mathit{GUs}}
\newcommand{\GUdp}{\mathit{GU\mbox{-}disp}}
\newcommand{\GUph}{\mathit{GU\mbox{-}phyc}}
\newcommand{\schannel}{\mathit{sconn}}
\newcommand{\GUlo}{\mathit{GU\mbox{-}logc}}
\newcommand{\AClo}{\mathit{AC\mbox{-}logc}}

%% file: abstract.tex
There is currently a drive in the aerospace domain to introduce 
machine to machine communication over wireless networks to 
improve ground processes at airports such as refuelling and
air conditiong. To this end a session key has to be established 
between the aircraft and the respective ground unit such as a fuel 
truck or a pre-conditiong unit. This is to be provided by
a `touch and go assistant in the aerospace domain' (TAGA), which 
allows an operator to pair up a ground unit and an aircraft 
present at a parking slot with the help of a NFC system.    
In this paper, we present the results of our security analysis
and co-development of requirements, security concepts, and
modular verification thereof. We show that by, and only by,
a combination of advanced security protocols and local
process measures we obtain secure and resilient designs
for TAGA. 
In particular, the design of choice is fully resilient against
long-term key compromises and parallel escalation
of attacks.



%% file: intron.tex
\section{Introduction}

Machine to machine (M2M) communication over wireless networks 
is increasingly adopted to improve speed, efficiency and accuracy
of service and maintenance processes at airports, ports, and manufacturing
plants. This does not come without security challenges: often these 
processes are safety-critical, and often, multi-instance attacks against
them would disrupt critical infrastructures. 
One example are the ground processes at an airport.
When an airplane has landed and reached its parking slot at the
apron many processes such as refuelling and air conditioning are performed. 
These processes determine the turnaround time, i.e.\ the time an airplane
needs to remain parked at the apron. Most ground operations that require a 
control loop are currently based on human-to-human communication, 
i.e.\ an operator at the ground machine and another in the aircraft.
M2M communication between the ground unit and the aircraft 
will simplify such ground operations. Moreover, the use of wireless 
connections is a must-have
due to the harsh conditions of the apron and increased flexibility.
 
Hence, there is ongoing activity on how to integrate secure M2M communication
between ground units and an aircraft during turnaround. It is planned
to use IEEE 802.11 WLAN as wireless channel, and to secure the data
exchanged by AES-CCMP analogously to WPA2. However, the challenge 
remains of how to set up a session key securely in this setting
where one party is mobile, the other is at different
locations, and there is no notion of global trust.
   
A use case within the project XXX is currently underway that
advocates a \emph{touch and go assistant in the aerospace domain (TAGA)}
to solve this challenge. 
The idea behind TAGA is that a key can be established
by pairing a ground unit and an aircraft both present at the same parking
slot based on their proximity --- a bit similar to pairing a bluetooth
device to a vehicle infotainment system but over a greater distance and
truly secure. To this end, each aircraft and ground unit is to be 
equipped with a TAGA controller that contains a secure element for
cryptographic operations and an NFC reader (accessible from an outside 
control panel). Moreover, the operator of each ground unit is to 
be provided with a passive NFC card. Altogether, this allows them 
to transport messages for key establishment from the ground unit, 
to the aircraft, and back
by means of taps with the NFC card against the respective NFC reader.
The `TAGA walk' can conveniently be integrated into the operator's
usual path to the aircraft and back while connecting up 
the respective supply hose. 

\removed{
There are two key differences between TAGA and related scenarios:
In contrast to local pairing of infotainment devices, TAGA is  
safety-critical and the design must be amenable to formal
verification to meet current safety and security norms of the aerospace
domain.
In contrast to security-critical NFC applications such as card
payment TAGA takes place in a secure zone that is established
around the airplane during takeover, to which only authorized
personnel have access. 
}

The advantage of such \emph{local} key establishment is that there is
no need to exchange the identities of the airplane and ground units
beforehand, and that there is no need for a central trusted party that provides
for key management. Thereby the flexibility
and trust assumptions of the current paper- and human-based ground processes
can be maintained. However, in contrast to local pairing of infotainment 
devices, TAGA is  
safety-critical and the design must be amenable to formal
verification to meet current safety and security norms in the aerospace
domain.

\removed{
The local paradigm is to av flexibilty in that there is no need
to exchange identities of the plane and the GU to be paired,
and also security in that the process is not dependent
on central key management. and fast availablity of revocation.}

\paragraph*{Our Contribution}
We have conducted a security analysis and co-developed 
security and resilience requirements, security concepts, and 
a modular verification
pattern of how to establish that the requirements 
are met by a security concept. We have also provided informal
proofs of all our results. To our knowledge this is the first systematic 
analysis, design space exploration, and verification method for 
local key establishment. We hope the results 
can also serve as a blueprint for other settings of M2M communication.
In more detail, our contributions are: 

(1) We concisely define the TAGA use case, and show how it
can be integrated into current processes for air conditioning
and fuelling. We make precise which overall guarantee the 
setup of a ground service with M2M communication must 
meet. We also provide an interface to safety analysis by
identifying five categories of attacks, what potential safety
impact they can cause, and which preconditions concerning key
establishment the attacker has to reach to mount them.
Throughout it becomes clear: While, as usual, it is crucial 
that the established key remains secret from the attacker 
a new challenge critical to key establishment for M2M communication 
is that the key is indeed
shared between the parties that are (or will be) physically connected.

(2) We explore whether we can obtain secure \emph{plain} TAGA,
where the underlying key establishment (KE) protocol is unauthenticated,
and hence, does not require any PKI.
It is well-known that such protocols cannot be secure against
active adversaries but that, like the Diffie-Hellman key exchange,
they can guarantee security in the presence of passive adversaries.
Hence, the challenge lies in whether we can implement the TAGA transport
in a way so that it provides an authentic message exchange. We show that
this is indeed possible but it comes at a cost, which involves
key management local to the airport, distance
bounding the communications of the readers, and a high dependence on
ground staff. 
On the positive side, our designs for plain TAGA come with the advantage 
that they ensure resilience against parallel escalation of attacks. 

(3) We investigate \emph{authenticated} TAGA (a-TAGA), which is based
on an authenticated KE (AKE) protocol and PKIs governed locally by 
airports, and airlines respectively.  
It is well-known that AKE protocols can securely establish a key in 
the presence of active adversaries. We show that it is straightforward
to implement secure a-TAGA when we relax the definition of security
in a way that only guarantees \emph{location agreement} rather than
alignment with the operator's TAGA walk. The new definition still
enables the correct setup of ground processes.
However, we are faced with another challenge now:
Our setting requires that a-TAGA must be highly resilient against 
long-term key compromises (LTKC).
We identify three notions of LTKC resilience and 
show how each can be met by a combination of advanced properties of AKE
protocols and local process measures. 

Altogher, we obtain several designs that are secure and 
resilient. The design of choice combines 
a-TAGA with an authentic
local channel, and is fully resilient against LTKCs and parallel escalation
of attacks. The remainder of the paper is structured according 
to the three parts. Related work can be found in App.~\ref{a:related},
and the notation we use for KE protocols is explained in 
App.~\ref{app:prot:props}. The proofs are provided in the remainder
of the appendix.

%% file: usecase/usecase.tex
\section{Ground Services with M2M Communication}

\subsection{The TAGA Use Case}
\input{usecase/taga-actors}

\input{usecase/taga-process}

\input{usecase/taga-setup}

%% file: usecase/taga-actors.tex
\subsubsection{Actors and Setting}

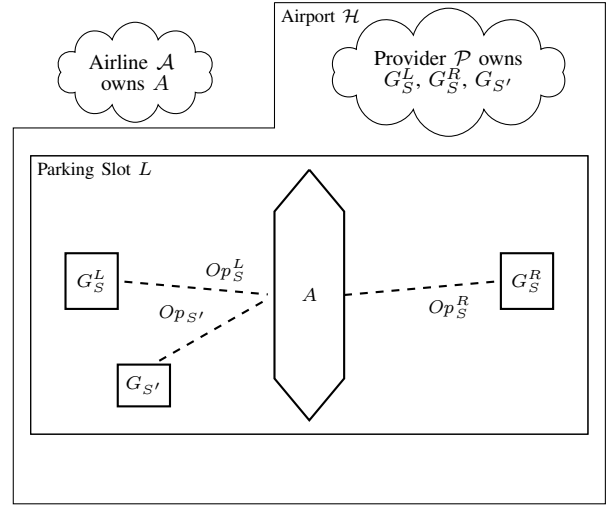
\begin{figure}
\centering{
\input{figures/fig-taga-actors}
}
\caption{\label{fig:taga:actors} TAGA Actors} 
\end{figure}

%

TAGA pairing takes place when an \emph{aircraft (AC)} $A$ has its 
turnaround at an \emph{airport $\mcH$}. 
Then $A$ is parked at a \emph{parking slot} $L$ of $\mcH$, 
and a secure zone is established around $A$, to which only authorized
personnel have access. When a \emph{ground unit (GU)} is to
provide a \emph{ground service} $S$ to $A$
then the \emph{operator} $\Op$ of $G$ carries out TAGA pairing as 
part of setting up the GU for the
service. To this end, each AC and GU is equipped with a NFC reader,
and the operator holds a passive NFC card. By
means of NFC taps with the card the operator can transfer messages for 
key establishment from the GU to the AC and back.
Once TAGA pairing is successfully concluded $A$ and $G$ share a session 
key for secure M2M communication on service $S$ over WLAN.

One AC can perform the TAGA process with several GUs 
in parallel, and one service might require several GUs.
However, each card, operator, and GU are involved 
in at most one TAGA process at a time. 
Each AC is owned by an \emph{airline $\mcA$}, and each GU is
governed by an \emph{airport $\mcH$}.  It is usual that an airport 
$\mcH$ does not operate the ground services itself, but has 
entrusted an \emph{airport service provider} $\mcP$ with 
the handling thereof.  
Each AC and GU is equipped with a controller that computes
the TAGA functionality. The controller is equipped with the NFC reader
and a secure element, which supports secure key storage and cryptographic 
operations for key management. 

%% file: figures/fig-taga-actors.tex
\begin{scaletikzpicturetowidth}{\columnwidth}
\begin{tikzpicture}[scale=\tikzscale,
yscale=0.8] 

{\scriptsize

\draw (-4.7,0) node {};
\draw (4.7,0) node {};

\draw (-0.5,5.25) -- (-0.5,3) -- (-4.25,3) -- (-4.25,-3.75) 
  -- (4.25,-3.75) -- (4.25,5.25) -- (-0.5,5.25);

\node [below right] at (-0.5,5.25) {Airport $\mcH$};

\node [cloud, draw,cloud puffs=10,cloud ignores aspect,
align=center] 
at (2,4)
{{\footnotesize Provider $\mcP$ owns} \\
  {\footnotesize $G^L_S$, $G^R_S$, $G_{S'}$}};

\node [cloud, draw,cloud puffs=10,cloud ignores aspect, align=center] 
at (-2.5,4)
{{\footnotesize Airline $\mcA$ } \\ 
  {\footnotesize owns $A$ }};

\draw[thick] (0,2.25) -- (0.5,1.5) -- (0.5,-1.5) -- (0,-2.25)
   -- (-0.5,-1.5) -- (-0.5,1.5) -- (0,2.25);
\node at (0,0) {$A$};

\draw[thick] (-3.5,-0.25) rectangle (-2.75,0.75);
\node at (-3.125,0.25) {$G^L_S$};
\draw[thick] (3.5,-0.25) rectangle (2.75,0.75);
\node at (3.125,0.25) {$G^R_S$};

\draw[thick] (-2.75,-2) rectangle (-2,-1.25);
\node at (-2.375,-1.625) {$G_{S'}$};

\node at (-0.5,0) (accontr) {};

\node at (-2.75,0.25) {} edge[thick,dashed] node[auto] 
    {$\Op^L_S$}  (accontr);
\node at (2.75,0.25) {} edge[thick,dashed] node[auto] {$\Op^R_S$}  (0.5,0);

\node at (-2.25,-1.25) {} edge[thick,dashed] node[auto] {$\Op_{S'}$} 
(accontr);

\draw (-4,-2.5) rectangle (4,2.5);
\draw (-4,-2.5) rectangle (4,2.5);

\node[below right] at (-4,2.5) {Parking Slot $L$};

}
\end{tikzpicture}
\end{scaletikzpicturetowidth}


%% file: usecase/taga-process.tex
\subsubsection{TAGA Pairing Process}

\begin{figure}
\centering{
\input{figures/fig-taga}}
\caption{\label{fig:taga:process:pn} TAGA pairing with
Diffie-Hellman key exchange}
\end{figure}
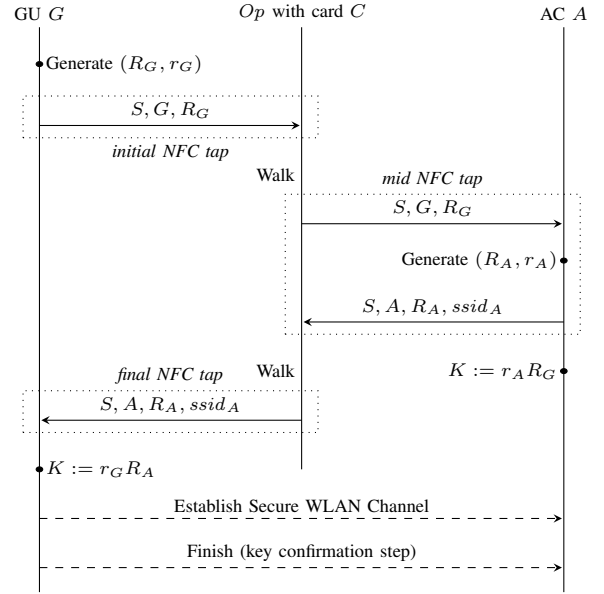

TAGA pairing is based on a three-pass key establishment 
(KE) protocol, where the third pass is a key confirmation step. 
It is illustrated in Fig.~\ref{fig:taga:process:pn} for
the case when the Diffie-Hellman (DH) key exchange is used
as the underlying protocol. Let $A$ be the AC, $G$ be the GU,
and $\Op$ be the operator of $G$.

\emph{First message pass.} $\Op$ initiates TAGA by an 
initial NFC tap with her NFC card $C$ at $G$'s controller. 
Thereby the first message $M_1$ is written to the card. $M_1$ contains 
information necessary for establishing the key together with the ID of
$G$ and the service $S$ that $G$ wishes to provide. $\Op$ carries $M_1$ 
stored on $C$ to the AC. There she carries out a mid NFC tap, during
which message $M_1$ is transferred to $A$'s controller.

\emph{Second message pass.} The AC checks whether the connection should
be granted. In the positive case, $A$ generates message $M_2$, which is 
written onto the NFC card, also during the mid NFC tap. 
$M_2$ contains information necessary for establishing the key 
together with the ID of $A$ and
access data to $A$'s WLAN such as the SSID. It also contains a ciphertext
to grant key confirmation to $G$.
The operator then carries the card with $M_2$ back to $G$. There she
transfers $M_2$ from the card to $G$'s controller by a final NFC tap.

\emph{WLAN and third message pass.} $G$ is now 
able to connect to $A$'s WLAN. A third message
is passed over the WLAN connection to achieve key confirmation
to $A$.

%% file: figures/fig-taga.tex
\begin{scaletikzpicturetowidth}{\columnwidth}
\begin{tikzpicture}[scale=\tikzscale,
yscale=0.75,
nfc/.style={->,>=stealth,shorten >=0.025cm},
wlan/.style={->,>=stealth,shorten >=0.025cm,dashed},
point/.style={radius=0.05}, 
dot/.style={circle,draw=black,fill=black, radius=0.025}]
{\scriptsize


\draw (-4,1) node[above] {GU $G$};
\draw (0,1) node[above] {$\Op$ with card $C$};
\draw (4,1) node[above] {AC $A$};

\draw (5,0) node[above] {};
\draw (-5,0) node[above] {};

\draw (0,1) -- (0,-8);
\draw (-4,1) -- (-4,-10.5);
\draw (4,1) -- (4,-10.5);

\coordinate (G1) at (-4,0.25);
\coordinate (G2) at (-4,-1);
\coordinate (G3) at (-4,-7);
\coordinate (G4) at (-4,-8);
\coordinate (G5) at (-4,-9);
\coordinate (G6) at (-4,-10);

\coordinate (O1) at (0,-1);
\coordinate (O2) at (0,-2);
\coordinate (O3) at (0,-3);
\coordinate (O4) at (0,-5);
\coordinate (O5) at (0,-6);
\coordinate (O6) at (0,-7);

\coordinate (A1) at (4,-3);
\coordinate (A2) at (4,-3.75);
\coordinate (A3) at (4,-5);
\coordinate (A4) at (4,-6);
\coordinate (A5) at (4,-9);
\coordinate (A6) at (4,-10);

\draw[nfc] (G2) -- (O1);
  \coordinate (G2O1) at (-2,-1);
\draw[nfc] (O3) -- (A1);
  \coordinate (O3A1) at (2,-3);
\draw[nfc] (A3) -- (O4);
  \coordinate (O4A3) at (2,-5);
\draw[nfc] (O6) -- (G3);
  \coordinate (G3O6) at (-2,-7);
\draw[wlan] (G5) -- (A5);
  \coordinate (G5A5) at (0,-9);
\draw[wlan] (G6) -- (A6);
  \coordinate (G6A6) at (0,-10);

\draw[dotted] (-4.25,-0.4) rectangle (0.25,-1.25);
\draw (-2,-1.25) node[below] {\emph{initial NFC tap}};
\draw[dotted] (-0.25,-2.4) rectangle (4.25,-5.25);
\draw (2,-2.4) node[above] {\emph{mid NFC tap}};
\draw[dotted] (-4.25,-6.4) rectangle (0.25,-7.25);
\draw (-2,-6.4) node[above] {\emph{final NFC tap}};

\draw (G1) node[right, align=left] {Generate $(R_G, r_G)$};
\draw[fill] (G1) circle [point];
\draw (G2O1) node[above] {$S, G, R_G$};

\draw (O2) node[left] {Walk};
\draw (O3A1) node[above] {$S, G, R_G$};
\draw (A2) node[left,align=left] {Generate $(R_A, r_A)$};
\draw[fill] (A2) circle [point];
\draw (O4A3) node[above] {$S, A, R_A, \ssid_A$};
\draw (O5) node[left] {Walk};
\draw (G3O6) node [above] {$S, A, R_A, \ssid_A$};
\draw (G4) node[right] {$K := r_G R_A$};
\draw[fill] (G4) circle [point];
\draw (A4) node[left] {$K := r_A R_G$};
\draw[fill] (A4) circle [point];
\draw (G5A5) node[above] {Establish Secure WLAN Channel};
\draw (G6A6) node[above] {Finish (key confirmation step)};

}
\end{tikzpicture}
\end{scaletikzpicturetowidth}

%% file: usecase/taga-setup.tex
\subsubsection{Ground Services with TAGA}
\label{s:taga:services}

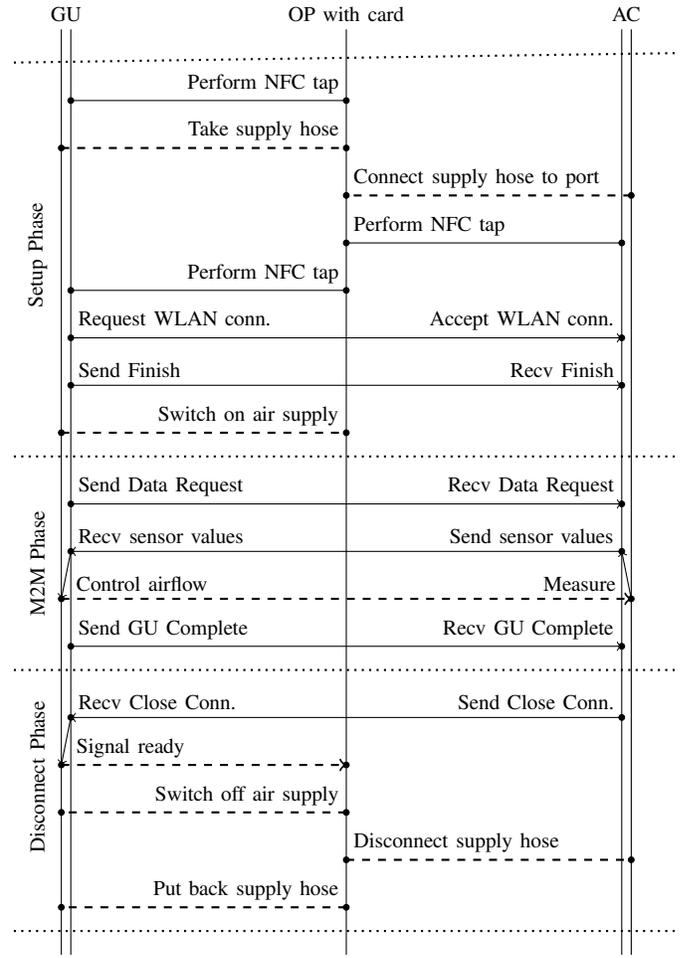
\begin{figure}
\centering{
\input{figures/fig-aircond}}
\caption{\label{fig:process:aircond} Process flow for air conditioning}
\end{figure}

We explain the process flow for air conditioning with
secure M2M communication, and TAGA pairing integrated into the
setup phase. Examples for fuelling with one and two trucks are 
provided in App.~\ref{app:fuel}. 
\removed{
We now give examples of the process flows for 
air conditioning and fuelling with secure M2M communication, and TAGA 
pairing integrated into the setup phase.
}

The process flow for air conditioning is shown in 
Fig.~\ref{fig:process:aircond}.
During the setup phase TAGA pairing is naturally integrated 
into the path that the operator needs to trace for the physical setup. 
First, the operator performs the first NFC tap at the pre-conditioning unit, 
and then walks to the AC carrying the air supply hose. After he has 
connected the hose to the AC's supply port the operator performs the 
second NFC tap at the AC, and then walks back to the pre-conditioning unit. 
There he performs the third NFC tap, and switches on the air supply. 

The pre-conditioning unit and AC are now ready to carry out the actual 
air conditioning during the M2M phase. First, the unit sends a 
request to receive data from the AC. This e.g.\ includes the desired 
temperature, the airflow parameters suitable for the AC, and 
readings of temperature sensors. The physical process is 
then controlled by the pre-conditioning unit based on temperature 
readings continuously communicated from the AC. When the desired 
temperature is reached the unit notifies the AC, and the service 
enters the disconnect phase. 
During this final phase the communication session is closed,
and the operator disconnects the supply hose.

\removed{
\begin{figure}
\input{figures/fig-fuelling}
\caption{\label{fig:process:fuelling} Process flow for fuelling}
\end{figure}

\paragraph*{Fuelling}
Fig.~\ref{fig:process:fuelling} shows a process for fuelling in a
setting where the fuel is obtained from an underground pipeline 
via a fuel truck. The fueller's first step is to connect the input 
fuel hose of the truck to a hydrant in the ground. Moreover, 
the fueller needs to connect the output fuel hose of the truck 
to the fuelling port of the AC. The fuelling ports are typically
positioned under the wings of the plane, and this step often involves 
that the fueller uses a lift integrated in the truck to take him up 
with the hose to one of the wings. In addition, safety measures are
carried out: before the AC is hooked up to the fuel hose 
a ground wire is laid from the 
fuel truck to the AC, and before the fuelling starts the fuel is
checked for contamination. 
TAGA pairing is integrated as follows: The first and second 
NFC taps are aligned with connecting up the ground wire, and the third 
tap is carried out after the fuel hose is connected up. 
Finally, the fueller activates the fuel
pump. At the AC, the pilot (or automatic control) waits until the 
secure channel is established and he has the okay from the fueller
(or the fuel truck via the secure channel) that the fuel hose is connected up. 

At the AC, the first action of the M2M phase is that the pilot 
(or automatic control) opens up the valves of the tanks, and activates the 
automatic fuel system. 
The M2M process makes use of the fact that most ACs already have an 
automated fuelling system: given a specified amount of fuel, the 
fuelling system distributes incoming fuel automatically into the various 
sections of the tank, monitors the amount of fuel already
received, and automatically shuts the valves of the tank when the 
specified amount has been reached. As usual backflow will stop 
the fuel pump of the truck. During the M2M phase the AC can 
communicate several fuel orders to the fuel truck so that the fuel
can automatically and precisely be topped up according to an increase in
the weight of the plane. When the final weight is known and the final 
fuel amount reached the AC notifies the fuel truck that the
service is complete. After an analogous reply by the truck the service
enters the disconnect phase. In the latter the communication session
is closed and the physical setup is reversed.

\paragraph*{Fuelling with Two Trucks}
Large ACs such as the A380 usually employ two fuel trucks to fuel
from the left and right wing in parallel. Then two parallel sessions
of the above process must be run. The 
following synchronization point between the two fuelling sessions
is required: the pilot (or automatic control) only opens the valves of the 
tank system after two secure channels are established and
both fuellers (or fuel trucks) have confirmed that the physical 
connection of the fuel hose on their side is ready. 

}

%% file: figures/fig-aircond.tex
\begin{scaletikzpicturetowidth}{\columnwidth}
\begin{tikzpicture}[scale=\tikzscale,
dot/.style={circle,draw=black,fill=black,inner sep=0pt,minimum size=1mm}]
{\footnotesize

\draw (-0.5,-4.25) node [rotate=90] {Setup Phase};
\draw (-0.5,-10.75) node [rotate=90] {M2M Phase};
\draw (-0.5,-15.25) node [rotate=90] {Disconnect Phase};

\node[above] at (0.1,0.5) {GU}; 
\node[above] at (6,0.5) {OP with card};
\node[above] at (11.9,0.5) {AC}; 

\coordinate (Gs) at (0,0.5);
\coordinate (Ge) at (0,-19);
\draw (Gs) -- (Ge);
\coordinate (Gsp) at (0.2,0.5);
\coordinate (Gep) at (0.2,-19);
\draw (Gsp) -- (Gep);

\coordinate (As) at (11.8,0.5);
\coordinate (Ae) at (11.8,-19);
\draw (As) -- (Ae);
\coordinate (Asp) at (12,0.5);
\coordinate (Aep) at (12,-19);
\draw (Asp) -- (Aep);

\coordinate (Os) at (6,0.5);
\coordinate (Oe) at (6,-19);
\draw (Os) -- (Oe);

\draw[dotted, thick] (-1,-0.2) -- (13,-0);

\coordinate (O2) at (6,-1);
\node[above left] at (O2) {Perform NFC tap};
\draw[fill] (O2) circle [radius=0.06];
\coordinate (G1) at (0.2,-1);
\node[above right] at (G1) {}; 
\draw[fill] (G1) circle [radius=0.06];
\draw (G1) -- (O2);

\coordinate (O4) at (6,-2);
\node[above left] at (O4) {Take supply hose};
\draw[fill] (O4) circle [radius=0.06];
\coordinate (G2) at (0,-2);
\draw[fill] (G2) circle [radius=0.06];
\draw [dashed, thick] (G2) -- (O4);

\coordinate (O6) at (6,-3);
\node[above right] at (O6) {Connect supply hose to port};
\draw[fill] (O6) circle [radius=0.06];
\coordinate (A1) at (12,-3);
\draw[fill] (A1) circle [radius=0.06];
\draw [dashed, thick] (A1) -- (O6);

\coordinate (O8) at (6,-8); \coordinate (O8) at (6,-4);
\node[above right] at (O8) {Perform NFC tap};
\draw[fill] (O8) circle [radius=0.06];
\coordinate (A2) at (11.8,-8); \coordinate (A2) at (11.8,-4);
\node[right] at (A2) {}; 
\draw[fill] (A2) circle [radius=0.06];
\draw  (O8) -- (A2);

\coordinate (O10) at (6,-5);
\node[above left] at (O10) {Perform NFC tap};
\draw[fill] (O10) circle [radius=0.06];
\coordinate (G3) at (0.2,-5);
\draw[fill] (G3) circle [radius=0.06];
\node[above right] at (G3) {}; 
\draw (G3) -- (O10);

\coordinate (G5) at (0.2,-6);
\draw[fill] (G5) circle [radius=0.06];
\node[above right] at (G5) {Request WLAN conn.};  
\coordinate (A5) at (11.8,-6);
\node[above left] at (A5) {Accept WLAN conn.};
\draw[fill] (A5) circle [radius=0.06];
\draw[->] (G5) -- (A5);

\coordinate (G5b) at (0.2,-7);
\draw[fill] (G5b) circle [radius=0.06];
\node[above right] at (G5b) {Send Finish};
\coordinate (A5b) at (11.8,-7);
\draw[fill] (A5b) circle [radius=0.06];
\node[above left] at (A5b) {Recv Finish};
\draw[->] (G5b) -- (A5b);

\coordinate (O12) at (6,-8);
\node[above left] at (O12) {Switch on air supply};
\draw[fill] (O12) circle [radius=0.06];
\coordinate (G4) at (0,-8);
\draw[fill] (G4) circle [radius=0.06];
\draw [dashed, thick] (G4) -- (O12);

\draw[dotted, thick] (-1,-8.5) -- (13,-8.5);


\coordinate (G6) at (0.2,-9.5);
\draw[fill] (G6) circle [radius=0.06];
\node[above right] at (G6) {Send Data Request};
\coordinate (A6) at (11.8,-9.5);
\draw[fill] (A6) circle [radius=0.06];
\node[above left] at (A6) {Recv Data Request};
\draw[->] (G6) -- (A6);

\coordinate (G7) at (0.2,-10.5);
\draw[fill] (G7) circle [radius=0.06];
\node[above right] at (G7) {Recv sensor values};
\coordinate (A7) at (11.8,-10.5);
\draw[fill] (A7) circle [radius=0.06];
\node[above left] at (A7) {Send sensor values};
\draw[->] (A7) -- (G7);

\coordinate (G8) at (0,-11.5);
\draw[fill] (G8) circle [radius=0.06];
\node[above right] at (G8) {\ Control airflow};
\coordinate (A8) at (12,-11.5);
\draw[fill] (A8) circle [radius=0.06];
\node[above left] at (A8) {Measure\ \ };
\draw[dashed, thick, ->] (G8) -- (A8);

\draw[->] (G7) -- (G8);
\draw[->] (A8) -- (A7);

\coordinate (G9) at (0.2,-12.5);
\draw[fill] (G9) circle [radius=0.06];
\node[above right] at (G9) {Send GU Complete};
\coordinate (A9) at (11.8,-12.5);
\draw[fill] (A9) circle [radius=0.06];
\node[above left] at (A9) {Recv GU Complete};
\draw[->] (G9) -- (A9);

\draw[dotted, thick] (-1,-13) -- (13,-13);


\coordinate (G10) at (0.2,-14);
\draw[fill] (G10) circle [radius=0.06];
\node[above right] at (G10) {Recv Close Conn.};
\coordinate (A10) at (11.8,-14);
\draw[fill] (A10) circle [radius=0.06];
\node[above left] at (A10) {Send Close Conn.} ;
\draw[<-] (G10) -- (A10);

\coordinate (G12) at (0,-15);
\draw[fill] (G12) circle [radius=0.06];
\node[above right] at (G12) {\ Signal ready};
\coordinate (O13a) at (6,-15);
\draw[fill] (O13a) circle [radius=0.06];
\node[above left] at (O13a) {}; 
\draw[->,dashed,thick] (G12) -- (O13a);

\draw[->] (G10) -- (G12);

\coordinate (O13) at (6,-16); 
\draw[fill] (O13) circle [radius=0.06];
\node[above left] at (O13) {\ Switch off air supply};

\draw[dashed, thick] (O13) -- (0,-16);
\draw[fill] (0,-16)  circle [radius=0.06];

\coordinate (O14) at (6,-17); 
\draw[fill] (O14) circle [radius=0.06];
\node[above right] at (O14) {Disconnect supply hose};

\draw[dashed, thick] (O14) -- (12,-17);
\draw[fill] (12,-17)  circle [radius=0.06];

\coordinate (O15) at (6,-18); 
\draw[fill] (O15) circle [radius=0.06];
\node[above left] at (O15) {\ Put back supply hose};

\draw[dashed, thick] (O15) -- (0,-18);
\draw[fill] (0,-18)  circle [radius=0.06];

\draw[dotted,thick] (-1,-18.5) -- (13,-18.5);

}
\end{tikzpicture}
\end{scaletikzpicturetowidth}

%% file: figures/fig-fuelling.tex
\begin{scaletikzpicturetowidth}{\columnwidth}
\begin{tikzpicture}[scale=\tikzscale,
phys/.style={dashed, thick}, nfc/.style={}, wlan/.style={},
wland/.style={->}, phase/.style={dotted, thick},
point/.style={radius=0.05}] 
{\footnotesize

\node[above] at (-6.125,0) {GU}; 
\node[above] at (0,0) {OP with MU};
\node[above] at (6.125,0) {AC}; 

\draw (0,0) -- (0,-20.5);
  \draw (0,-22) -- (0,-23);
\draw (-6,0) -- (-6,-23);
\draw (-6.25,0) -- (-6.25,-23);
\draw (6,0) -- (6,-23);
\draw (6.25,0) -- (6.25,-23);

\coordinate (F1) at (0,-1);
\coordinate (F2) at (0,-2); 
\coordinate (F3a) at (0,-3);
  \coordinate (GF3a) at (-6.25,-3);
\coordinate (F3) at (0,-4);
  \coordinate (AF3) at (6.25,-4);
\coordinate (F4) at (0,-5); 
\coordinate (F6) at (0,-6);
  \coordinate (AF6) at (6.25,-6);
\coordinate (F7) at (0,-7);
  \coordinate (GF7) at (-6.25,-7);
\coordinate (F8) at (0,-8); 
\coordinate (F9) at (0,-11.25);  

\coordinate (G1) at (-6,-2); 
\coordinate (G2) at (-6,-8); 
\coordinate (G3) at (-6,-9); 
\coordinate (G4) at (-6,-10); 

\coordinate (A1l) at (6,-5);  
\coordinate (A2l) at (6,-9);  
\coordinate (A3l) at (6,-10); 

\coordinate (A4) at (6.25,-11.25);  

\coordinate (Osynch) at (0,-10.75);
\coordinate (Asynch) at (6.25,-10.75);
\draw[phys] (Osynch) -- (Asynch);

\draw[->] (A3l) -- (A4);

\draw[nfc] (G1) -- (F2);
\draw[nfc] (G2) -- (F8);
\draw[phys] (F3a) -- (GF3a);
\draw[phys] (F3) -- (AF3); 
\draw[nfc] (F4) -- (A1l);
\draw[phys] (F6) -- (AF6);
\draw[phys] (F7) -- (GF7);
\draw[wland] (G3) -- (A2l);
\draw[wland] (G4) -- (A3l);

\draw[phase] (-6.75,-0.5) -- (6.75,-0.5);
\draw[phase] (-6.75,-11.6) -- (6.75,-11.6);
\draw[phase] (-6.75,-17) -- (6.75,-17);
\draw[phase] (-6.75,-22.5) -- (6.75,-22.5);

\draw (-6.75,-5.75) node [rotate=90] {Setup Phase};
\draw (-6.75,-14.25) node [rotate=90] {M2M Phase};
\draw (-6.75,-19.75) node [rotate=90] {Disconnect Phase};

\removed{
\draw [decorate,decoration={brace,amplitude=10pt},xshift=-4pt,yshift=0pt]
(-6.5,-11.5) -- (-6.5,-0.5)  
node [black,midway,xshift=-0.6cm,rotate=90] 
{\footnotesize Setup Phase};

\draw [decorate,decoration={brace,amplitude=10pt},xshift=-4pt,yshift=0pt]
(-6.5,-17) -- (-6.5,-11.5)  
node [black,midway,xshift=-0.6cm,rotate=90] 
{\footnotesize M2M Phase};

\draw [decorate,decoration={brace,amplitude=10pt},xshift=-4pt,yshift=0pt]
(-6.5,-22.5) -- (-6.5,-17)  
node [black,midway,xshift=-0.6cm,rotate=90] 
{\footnotesize Disconnect Phase};
}


\node[right] at (F1) {Connect input hose to hydrant};
\node[above left] at (F2) {Perform NFC tap};
\node[above left] at (F3a) {Take ground wire from GU};
\node[above right] at (F3) {Connect ground wire to AC};
\node[above right] at (F4) {Perform NFC tap};
\node[above right] at (F6) {Connect output hose to AC port};
\node[above left] at (F7) {Check fuel for contamination};
\node[above left] at (F8) {Perform NFC tap};
\node[left] at (F9) {Activate fuel pump};

\node[above right] at (G1) {}; 
\node[above right] at (G2) {}; 
\node[above right] at (G3) {Request WLAN Conn.};
\node[above right] at (G4) {Send Finish};

\node[above left] at (A1l) {}; 
\node[above left] at (A2l) {Accept WLAN Conn.};
\node[above left] at (A3l) {Recv Finish};


\node[above right] at (Osynch) {Signal hose connected};
\node[left] at (A4) {Open tank valve\ \ };


\coordinate (Gm1) at (-6,-12.5); 
\coordinate (Gm2) at (-6,-13.5); 
\coordinate (Gm3) at (-6.25,-14.5); 
\coordinate (Gm4) at (-6,-15.5); 
\coordinate (Gm5) at (-6,-16.5); 

\coordinate (Am1) at (6,-12.5); 
\coordinate (Am2) at (6,-13.5); 
\coordinate (Am3) at (6.25,-14.5); 
\coordinate (Am4) at (6,-15.5); 
\coordinate (Am5) at (6,-16.5); 

\node[above right] at (Gm1) {Send Data Request};
\node[above right] at (Gm2) {Recv fuel order};
\node[above right] at (Gm3) {\ \ Pump fuel};
\node[above right] at (Gm4) {Recv AC Complete};
\node[above right] at (Gm5) {Send GU Complete};

\node[above left] at (Am1) {Recv Data Request};
\node[above left] at (Am2) {Send fuel order};
\node[above left] at (Am3) {Measure\ \ };
\node[above left] at (Am4) {Send AC Complete};
\node[above left] at (Am5) {Recv GU Complete};

\draw[->] (Gm2) -- (Gm3);
\draw[->] (Am3) -- (Am2);


\draw[wland] (Gm1) -- (Am1);
\draw[wland] (Am2) -- (Gm2);
\draw[phys] (Gm3) -- (Am3);
\draw[wland] (Am4) -- (Gm4);
\draw[wland] (Gm5) -- (Am5);


\coordinate (Af1) at (6.25,-18); 
\draw[fill] (Af1) circle [point];
\node[left] at (Af1) {Close tank valve\ \ };
\coordinate (Af2) at (6,-19); 
\coordinate (Gf1) at (-6,-19); 
\coordinate (Gf2) at (-6.25,-20); 
\coordinate (Ff1) at (0,-20); 

\node[above left] at (Af2) {Send Close Conn.};

\node[above right] at (Gf1) {Recv Close Conn.};
\node[above right] at (Gf2) {\ \ Signal ready};

\draw[wland] (Af2) -- (Gf1);
\draw[phys] (Gf2) -- (Ff1);

\draw[->] (Am5) -- (Af1);
\draw[->] (Gf1) -- (Gf2);  

\draw (-1.5,-20.5) rectangle (1.5,-22);
\node[align=center] at (0,-21.25) {Reverse \\ physical setup}; 
\draw[phys] (-6.25, -21.25) -- (-1.5,-21.25);
\draw[phys] (6.25, -21.25) -- (1.5,-21.25);


\draw[fill] (F1) circle [point];
\draw[fill] (F2) circle [point];
\draw[fill] (F3) circle [point];
\draw[fill] (F3a) circle [point];
\draw[fill] (AF3) circle [point];
\draw[fill] (F4) circle [point];
\draw[fill] (F6) circle [point];
\draw[fill] (AF6) circle [point];
\draw[fill] (F7) circle [point];
\draw[fill] (F8) circle [point];
\draw[fill] (F9) circle [point];

\draw[fill] (G1) circle [point];
\draw[fill] (G2) circle [point];
\draw[fill] (G3) circle [point];
\draw[fill] (G4) circle [point];
\draw[fill] (GF3a) circle [point];
\draw[fill] (GF7) circle [point];

\draw[fill] (A1l) circle [point];
\draw[fill] (A2l) circle [point];
\draw[fill] (A3l) circle [point];
\draw[fill] (A4) circle [point];

\draw[fill] (Am1) circle [point];
\draw[fill] (Am2) circle [point];
\draw[fill] (Am3) circle [point];
\draw[fill] (Am4) circle [point];
\draw[fill] (Am5) circle [point];

\draw[fill] (Gm1) circle [point];
\draw[fill] (Gm2) circle [point];
\draw[fill] (Gm3) circle [point];
\draw[fill] (Gm4) circle [point];
\draw[fill] (Gm5) circle [point];

\draw[fill] (Af2) circle [point];
\draw[fill] (Gf1) circle [point];
\draw[fill] (Gf2) circle [point];
\draw[fill] (Ff1) circle [point];

}
\end{tikzpicture}
\end{scaletikzpicturetowidth}

%% file: setup-2/setup.tex
\subsection{Correct Setup of Ground Services}

The setup of a ground service must guarantee
that the AC and each of the participating GUs are 
both physically connected as well as linked by a secure cyber channel.
We now make precise this  guarantee, 
based on definitions of secure channel and process flow
for a ground service.

\input{setup-2/setup-prel}

\input{setup-2/taga-goals}

\input{setup-2/safety}

%% file: setup-2/setup-prel.tex

\subsubsection{Cyber Connection and Secure Channel}
A \emph{(cyber) connection} owned by a GU or AC for (M2M 
communication on) service $S$ is defined by a tuple 
$(\cid, \DP, K, S)$ where $\cid$ is the identifier of the
connection, $\DP$ is a data confidentiality and integrity 
protocol (such as AES-CCMP), and
$K$ is a session key (such as established by TAGA).
Hence, two parties that own matching connections can exchange data 
on $S$ over the connection $\cid$ transformed by $\DP$ 
using $K$ as the key.
We say a GU $G$ and an AC $A$ share a \emph{secure channel}
for (M2M communication on) service $S$ if $G$ and $A$ both own
a connection $(\cid, \DP, K, S)$ such that
(1)~$\DP$ is secure, and
(2)~$K$ is freshly generated and only known to
$G$ and $A$. When $\DP$ and $\cid$ are clear from the context we
also write $(K, S)$ for a secure channel.

\removed{
We say a GU $G$ and an AC $A$ share a \emph{secure channel (established
by TAGA)} iff $G$ has a session $(S, \Gini, \Gfin, K)$ and
$A$ has a session $(\Amid, S, K, c)$ for some $\Gini$, $\Gfin$,
$\Amid$ and $K$ is known to no party other than $G$ and $A$,
and $G$ and $A$ have no other sessions where $K$ is the key.}


%% file: setup-2/taga-goals.tex
\subsubsection{Process Flow for a Ground Service}
A \emph{process flow} $F$ for a ground service $S$ (with M2M
communication) is a tuple $(n_F, A_F, <_F, P_F)$, where 
$n_F$ specifies the number of required GUs, $(A_F, <_F)$ 
is a partial order of actions to be carried out by the participating
actors (GUs, AC, operators, perhaps crew), and $P_F$ is a
multiset of parking positions relative to the AC such that $P_F$ is of 
size $n_F$. 
$P_F$ specifies where the $n$ GUs are parked during the process 
(e.g.\ [right wing, left wing] or [mid, mid]).
We require $P_F$ to be a multiset rather
than a set since a service can involve several GUs at one relative position,
and in practice it is important to avoid unnecessary 
precision (such as defining sub-positions). Given GUs $G_1, \ldots, G_n$ 
such that each $G_i$ is located at the same
parking slot, we denote the multiset of their relative positions
by $\pos(\{G_1, \dots, G_n\})$. 

We assume the actions of each participant are divided into
a setup phase, a M2M phase, during which the actual service is
carried out, and a disconnect phase.
We require that $(A_F, <_F)$ guarantees: 
(1)~A ground coordinator ensures that exactly $n_F$ GUs for service $S$,
    say $G_1, \ldots, G_n$, are at the parking slot such that 
      $\pos(\{G_1, \ldots, G_n\}) = P_F$.
(2)~When, in the view of an operator $\Op$, the setup phase is 
    completed then the GU $\Op$ operates is physically connected up 
  for service $S$ to $A$.
(3)~When, in the view of a GU $G$, the setup phase is completed 
    then $G$ has a cyber connection for service $S$.
(4)~When, in the view of an AC $A$, the setup phase is completed 
    then $A$ has exactly $n_F$ cyber connections for service $S$. 

\subsubsection{Correct Setup of Ground Service}
\label{s:taga:goals}
Let $L$ be a parking slot, on which a turnaround process takes place.
We denote the one AC that is present at $L$ by $\AC(L)$.
Let $F$ be a process flow for service $S$.
We say $F$ guarantees \emph{Correctness of Setup to a GU $G$} if, whenever 
$G$ engages into the M2M phase of $F$ at a parking slot $L$ then
\begin{enumerate}
\item $\pos(\{G\}) \in P_F$ and
    $G$ is physically connected for service $S$ to $\AC(L)$, and
\item $G$ shares a secure channel for service $S$ with $\AC(L)$. 
\end{enumerate}
We say $F$ guarantees \emph{Correctness of Setup to an AC $A$} if, whenever 
$A$ engages into the M2M phase of $F$ at a parking slot 
$L$  then
there are $n_F$ GUs $G_1, \ldots, G_n$ at $L$ with 
$\pos(\{G_1, \ldots, G_n\}) = P_F$ such that 
\begin{enumerate}
\item $A$ is physically connected for service $S$ to and only
      to the GUs $G_1, \ldots, G_n$, and
\item $A$ shares a secure channel for service $S$ with and only with
      each of the GUs $G_1, \ldots, G_n$. 
\end{enumerate}
We say $F$ guarantees \emph{Correctness of Setup} if it guarantees this
to both, GUs and ACs.

\removed{
\paragraph*{no}
We expect that well-defined process flows guarantee certain
requirements regarding the physical setup and cyber connections.
It is straightforward to check that
the process flows for the services air conditioning, fuelling, and 
fuelling with two trucks 
described in Section~\ref{s:taga:services} are well-defined 
(by inspection of the process flow).

\begin{definition}
We say a process flow $F$ for ground service $S$ is well-defined 
iff for all GUs $G$ and ACs $A$ 
governed by $F$ the following holds:
\begin{enumerate}
\item Ground staff physically connect up exactly $n_F$
      GUs $G_1, \ldots, G_n$ of service $S$ to $A$ such that 
      $\pos(\{G_1, \ldots, G_n) = P_F$.
\item When the setup phase is completed in the view of $G$
  then $G$ has a cyber connection for service $S$.
\item When the setup phase is completed in the view of $A$
   then $A$ has $n_F$ cyber connections for service $S$. 
\item $A$ does not accept more than $n_F$ cyber connection 
      requests for $S$.
\end{enumerate}
\end{definition}
}
\removed{
\begin{lemma}
\label{lem:setup}
The process flows for the services air conditioning, fuelling, and 
fuelling with two trucks 
described in Section~\ref{s:taga:services} are well-defined. 
\end{lemma}

\begin{proof}[Proof of Lemma~\ref{lem:setup}]
\end{proof}}

%% file: setup-2/safety.tex
\subsection{M2M Attack Categories and Safety Impact}

\label{s:goals:safety}

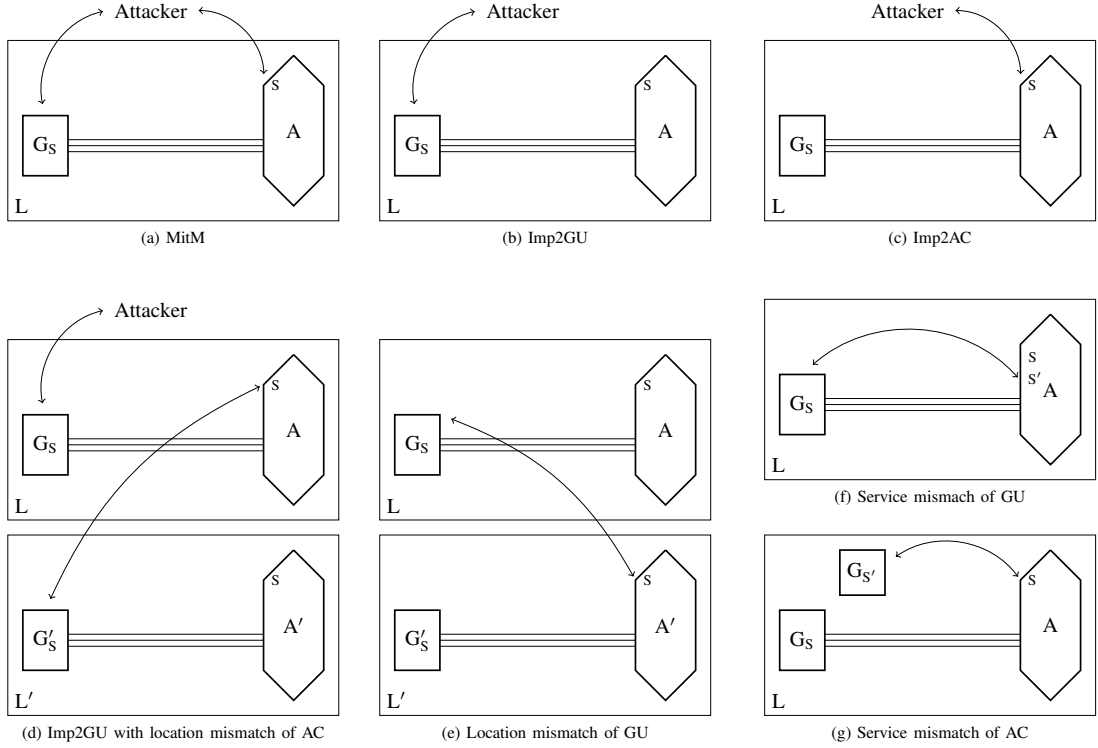
\begin{figure*}
\centering{
\scalebox{0.8}{
\begin{tabular}{ccc}
\subfloat[MitM]{
\input{figures/att-mitm}}
&
\subfloat[Imp2GU]{
\input{figures/att-imp2gu}}
&
\subfloat[Imp2AC]{
\input{figures/att-imp2ac}}
\\
\subfloat[Imp2GU with location mismatch of AC]{
\input{figures/att-imp2gu-div}}
&
\subfloat[Location mismatch of GU]{
\input{figures/att-mism-gu}}
& 
\begin{tabular}[b]{c}
\subfloat[Service mismach of GU]{
\input{figures/att-mism-ser-gu}}
\\
\subfloat[Service mismatch of AC]{
\input{figures/att-mism-ser-ac}}
\end{tabular}
\end{tabular}
}}
\caption{\label{fig:astates} M2M attack categories}
\end{figure*}

\input{setup-2/attacks}

\paragraph{Man-in-the-Middle (MitM)} The most serious attack 
is a MitM: it means the attacker can insert, delete or modify 
any messages as he pleases. 
As our examples show the safety impact is high for air conditioning 
but not so for fuelling due to inbuilt safety measures.

\begin{example}[Air Conditioning]
\label{ex:MitM:aircond}
The attacker can forge airflow parameters and sensor values
that will induce $G_S$ to apply air pressure and temperature
unsuitable to $A$. This can be highly damaging: if the cooling process 
is too fast, water in the pipes can quickly become frozen and 
clog up the pipes. This can happen very quickly: e.g. with the lowest 
inlet temperature within 30 seconds, with safety considerations 
still within 100 seconds. The resulting backflow will be detected by the GU.
However, in the worst case pipes might already have burst. In any case 
the pipes have to be checked for damage afterwards, which is a costly procedure.
In the worst case, the attacker could try to optimize the attack based
on the sensor values sent by $A$: he could try to control the airflow in
a way that maximizes the strain on the pipes without this being
detected during service time but with a high risk that pipes burst during 
flight.
\end{example}

\begin{example}[Fuelling]
\label{ex:MitM:fuelling}
The attacker can forge fuel orders, and induce the fuel truck to
load an insufficient or surplus amount of fuel. While this can be highly 
disruptive there is no safety impact. Since the AC measures the fuel 
itself $A$ will notice if the loaded fuel is not sufficient. Moreover, 
if the attacker tries to cause spillage (and hence a fire
hazard) by too large a fuel order this will not succeed either since
backflow will stop the pump of the fuel truck. 
\end{example}

\paragraph{Impersonation to GU (Imp2GU)} 
Another serious attack state is that for (pure) Imp2GU.  
This allows the attacker to feed any information he likes to $G_S$ while 
$G_S$ thinks this information stems from $A$ and adjusts the 
service correspondingly. The safety impact is potentially high.

\begin{example}[Air Conditioning]
The attacker feeds in airflow parameters and sensor values,
and $G_S$ will control the airflow based on this information.
Since the air supply leads directly into the mixer unit of the plane 
this will take immediate effect without $A$ itself having to open a valve or
the like first. Crew or ground stuff might notice that something is
wrong and switch off the air supply manually. However, as discussed in
Example~\ref{ex:MitM:aircond} damage 
can occur quickly and this might be too late. In contrast to the MitM 
attack the attacker is not able to optimize the attack based on 
sensor values from $A$.
\end{example}

\paragraph{Impersonation to AC (Imp2AC)} 
When the physical control of the service lies entirely with the GU
then an Imp2AC attack is usually harmless. However, this might 
change when several GUs are employed. 

\removed{
Although the attacker can feed any information he likes to
$A$ nothing will happen until a GU engages into the service. And
any GU will only do so after it has completed the physical setup 
and established a channel for the service itself. 
When the physical control lies at least partly with
the AC then this attack state can have safety impact or not, e.g.\ depending 
on whether there are inbuilt safety measures.
}

\begin{example}[Fuelling with two trucks]
Assume fuelling with two trucks is carried out. Assume the truck
under the left wing is already connected (both physically and in cyber) 
while the truck under the right wing is still being set up and the 
fueller currently opens the tank. Now assume that by an Imp2AC attack
the attacker pretends the right truck is also fully set up. Then (in
the case of automatic control) the AC will open the fuel valves, 
order fuel, and the left truck will start
pumping fuel. However, the right operator is at the open tank, and
a highly unsafe situation has been reached.

\removed{
The attacker can impersonate the fuel truck and in its name report to 
$A$ that a problem has occurred. This could induce $A$ to close the
valves of the tank while in reality nothing is wrong with the
fuel truck and it keeps pumping fuel. However, this will not result 
in any safety issue since the truck will simply stop pumping due to back flow.
}
\end{example}

\removed{
Analogously to Imp2GU, the attacker could try to reach the attack
state $(\mmX, \att)$ to combine  Imp2AC with a mismatch 
diversion of the GU. Alternatively, in this state the Imp2AC attack
could itself be the diversion to a safety-critical location mismatch attack 
as we will see next.}

\paragraph{Imp2GU with Mismatch Diversion}
When an attacker mounts a Imp2GU attack it might be beneficial for him
to divert the AC by matching it up with another GU. 
One reason is to thereby ensure that the impersonation attack
remains undetected for longer;
another reason is that without a successful session the
AC might not be affected at all, perhaps because it controls a valve
that it will not open otherwise.

\begin{example}[Fuelling]
When fuelling the pilot has to open the valves of the tanks before
the fuel flow can start. Hence, as long as $A$ does not have a
successful session for fuelling the Imp2GU attack will be detected
early on: the fuel pump of $G_S$ will stop almost immediately due to backflow.
However, if the attacker conducts a `location mismatch of AC attack' 
in parallel then $A$  will open the fuel valves, and the attacker can
make $G_S$ load fuel based on a forged order. Fortunately, in this
case there is no safety impact as explained in Example~\ref{ex:MitM:fuelling}.
\end{example}

\paragraph{Location Mismatch of GU}
We have just seen that mismatch attacks can be used as a diversion.
The next example 
illustrates that they can be safety-critical in their own right. 

\begin{example}[Air Conditioning]
Let $A'$ be an AC at another parking position $L'$ with ongoing
air conditioning service. 
Assume $A'$ is a large plane (e.g.\ an A380), and $A$ is a smaller plane
(e.g.\ an A319). Moreover, assume the attacker has obtained the 
state for location mismatch of $G_S$ where $G_S$ is speaking to $A'$.
Hence, $A'$ will communicate its parameters for airflow to $G_S$ while
$G_S$'s air supply hose is connected to the smaller $A$. 
Since the airpacks of $A$ are smaller than those of $A'$ 
airflow parameters suitable for $A'$ can be damaging for $A$: 
the air pressure is likely to be too high and the cooling process too 
fast. As a consequence pipes of $A$ can be damaged as explained
in Example~\ref{ex:MitM:aircond}.
\end{example}


\paragraph{and g) Service Mismatch}
A service mismatch 
will likely 
be noticed as soon as M2M communication starts. 
The messages will not follow the protocol the AC, and respectively 
GU expects, and hence errors will be raised: e.g.\ the AC 
provides temperature readings while the GU expects a fuel order. 
However, a service mismatch could still be employed as a diversion
to support another attack: when the only goal is that the party to divert
has established a channel for the service (because only then it will
open a valve), or when the service is such that there usually is
a delay between the point when the channel is established and the start
of the M2M communication. 

\removed{
\subsubsection{Identity Mismatch}
TODO to synchronize with something in TAGA auth.

billing, (fuel) unknown-key share for location-mismatch,
talk with the one for L and S, have certificate for low S'/A 
}

%% file: figures/att-mitm.tex
\begin{tikzpicture}[
pil/.style={
           <->,
           shorten <=2pt,
           shorten >=2pt,}
]


\draw (-0.25,0.25) rectangle (5.25,3.25);
\draw (-0.25,0.25) node[above right] {L};

\draw[thick] (4.5,3) -- (4,2.5) node[right] { 
  \scriptsize S} 
  -- (4,1) -- (4.5,0.5)
        -- (5,1) -- (5,2.5) -- (4.5,3);
\draw (4.5,1.75) node {A};

\draw[thick] (0,1) rectangle (0.75,2);
\draw (0.375,1.5) node {G$_{\mbox{\scriptsize S}}$}; 


\draw (0.75,1.4) -- (4,1.4);
\draw (0.75,1.5) -- (4,1.5);
\draw (0.75,1.6) -- (4,1.6);

\draw[above] (2.125,3.5) node (att) {Attacker};


\node at (0.375,2) {}
  edge[pil,bend left=45] (att.west); 

\node at (4,2.5) {}  
  edge[pil,bend right=45] (att.east); 

\end{tikzpicture}

%% file: figures/att-imp2gu.tex
\begin{tikzpicture}[
pil/.style={
           <->,
           shorten <=2pt,
           shorten >=2pt,}
]


\draw (-0.25,0.25) rectangle (5.25,3.25);
\draw (-0.25,0.25) node[above right] {L};

\draw[thick] (4.5,3) -- (4,2.5) node[right] { 
  \scriptsize S} 
  -- (4,1) -- (4.5,0.5)
        -- (5,1) -- (5,2.5) -- (4.5,3);
\draw (4.5,1.75) node {A};

\draw[thick] (0,1) rectangle (0.75,2);
\draw (0.375,1.5) node {G$_{\mbox{\scriptsize S}}$}; 


\draw (0.75,1.4) -- (4,1.4);
\draw (0.75,1.5) -- (4,1.5);
\draw (0.75,1.6) -- (4,1.6);

\draw[above] (2.125,3.5) node (att) {Attacker};


\node at (0.375,2) {}
  edge[pil,bend left=45] (att.west); 


\end{tikzpicture}

%% file: figures/att-imp2ac.tex
\begin{tikzpicture}[
pil/.style={
           <->,
           shorten <=2pt,
           shorten >=2pt,}
]


\draw (-0.25,0.25) rectangle (5.25,3.25);
\draw (-0.25,0.25) node[above right] {L};

\draw[thick] (4.5,3) -- (4,2.5) node[right] { 
  \scriptsize S} 
  -- (4,1) -- (4.5,0.5)
        -- (5,1) -- (5,2.5) -- (4.5,3);
\draw (4.5,1.75) node {A};

\draw[thick] (0,1) rectangle (0.75,2);
\draw (0.375,1.5) node {G$_{\mbox{\scriptsize S}}$}; 


\draw (0.75,1.4) -- (4,1.4);
\draw (0.75,1.5) -- (4,1.5);
\draw (0.75,1.6) -- (4,1.6);

\draw[above] (2.125,3.5) node (att) {Attacker};


\node at (4,2.5) {}  
  edge[pil,bend right=45] (att.east); 

\end{tikzpicture}

%% file: figures/att-imp2gu-div.tex
\begin{tikzpicture}[
pil/.style={
           <->,
           shorten <=2pt,
           shorten >=2pt,}
]


\draw (-0.25,3.5) rectangle (5.25,6.5);
\draw (-0.25,3.5) node[above right] {L};

\draw[thick] (4.5,6.25) -- (4,5.75) node[right] { 
  \scriptsize S} 
  -- (4,4.25) -- (4.5,3.75)
        -- (5,4.25) -- (5,5.75) -- (4.5,6.25);
\draw (4.5,5) node {A};

\draw[thick] (0,4.25) rectangle (0.75,5.25);
\draw (0.375,4.75) node {G$_{\mbox{\scriptsize S}}$}; 


\draw (0.75,4.65) -- (4,4.65);
\draw (0.75,4.75) -- (4,4.75);
\draw (0.75,4.85) -- (4,4.85);

\draw[above] (2.125,6.75) node (att) {Attacker};


\node at (0.375,5.25) {}
  edge[pil,bend left=45] (att.west); 



\draw (-0.25,0.25) rectangle (5.25,3.25);
\draw (-0.25,0.25) node[above right] {L$'$};

\draw[thick] (4.5,3) -- (4,2.5) node[right] { 
  \scriptsize S} 
  -- (4,1) -- (4.5,0.5)
        -- (5,1) -- (5,2.5) -- (4.5,3);
\draw (4.5,1.75) node {A$'$};

\draw[thick] (0,1) rectangle (0.75,2);
\draw (0.375,1.5) node {G$'_{\mbox{\scriptsize S}}$}; 


\draw (0.75,1.4) -- (4,1.4);
\draw (0.75,1.5) -- (4,1.5);
\draw (0.75,1.6) -- (4,1.6);





\node at (0.375,2) {} 
  edge[pil,bend left=20] (4,5.75) ;

\end{tikzpicture}

%% file: figures/att-mism-gu.tex
\begin{tikzpicture}[
pil/.style={
           <->,
           shorten <=2pt,
           shorten >=2pt,}
]


\draw (-0.25,3.5) rectangle (5.25,6.5);
\draw (-0.25,3.5) node[above right] {L};

\draw[thick] (4.5,6.25) -- (4,5.75) node[right] { 
  \scriptsize S} 
  -- (4,4.25) -- (4.5,3.75)
        -- (5,4.25) -- (5,5.75) -- (4.5,6.25);
\draw (4.5,5) node {A};

\draw[thick] (0,4.25) rectangle (0.75,5.25);
\draw (0.375,4.75) node {G$_{\mbox{\scriptsize S}}$}; 


\draw (0.75,4.65) -- (4,4.65);
\draw (0.75,4.75) -- (4,4.75);
\draw (0.75,4.85) -- (4,4.85);

\draw[above] (2.125,6.75) node (att) {}; 





\draw (-0.25,0.25) rectangle (5.25,3.25);
\draw (-0.25,0.25) node[above right] {L$'$};

\draw[thick] (4.5,3) -- (4,2.5) node[right] { 
  \scriptsize S} 
  -- (4,1) -- (4.5,0.5)
        -- (5,1) -- (5,2.5) -- (4.5,3);
\draw (4.5,1.75) node {A$'$};

\draw[thick] (0,1) rectangle (0.75,2);
\draw (0.375,1.5) node {G$'_{\mbox{\scriptsize S}}$}; 


\draw (0.75,1.4) -- (4,1.4);
\draw (0.75,1.5) -- (4,1.5);
\draw (0.75,1.6) -- (4,1.6);






\node at (0.75,5.25) {} 
  edge[pil,bend left=20] (4,2.5) ;

\end{tikzpicture}

%% file: figures/att-mism-ser-gu.tex
\begin{tikzpicture}[
pil/.style={
           <->,
           shorten <=2pt,
           shorten >=2pt,}
]


\draw (-0.25,0.25) rectangle (5.25,3.25);
\draw (-0.25,0.25) node[above right] {L};

\draw[thick] (4.5,3) -- (4,2.5) node[align=center,below right] { 
   {\scriptsize S } \\[1ex] {\scriptsize S$'$}} 
  -- (4,1) -- (4.5,0.5)
        -- (5,1) -- (5,2.5) -- (4.5,3);
\draw (4.5,1.75) node {A};

\draw[thick] (0,1) rectangle (0.75,2);
\draw (0.375,1.5) node {G$_{\mbox{\scriptsize S}}$}; 


\draw (0.75,1.4) -- (4,1.4);
\draw (0.75,1.5) -- (4,1.5);
\draw (0.75,1.6) -- (4,1.6);

\draw[above] (2.125,3.5) node (att) {}; 




\node at (0.375,2) {}  
  edge[pil,bend left=45] node[auto] {} (4,1.9) {};

\end{tikzpicture}

%% file: figures/att-mism-ser-ac.tex
\begin{tikzpicture}[
pil/.style={
           <->,
           shorten <=2pt,
           shorten >=2pt,}
]


\draw (-0.25,0.25) rectangle (5.25,3.25);
\draw (-0.25,0.25) node[above right] {L};

\draw[thick] (4.5,3) -- (4,2.5) node[right] { 
  \scriptsize S} 
  -- (4,1) -- (4.5,0.5)
        -- (5,1) -- (5,2.5) -- (4.5,3);
\draw (4.5,1.75) node {A};

\draw[thick] (0,1) rectangle (0.75,2);
\draw (0.375,1.5) node {G$_{\mbox{\scriptsize S}}$}; 

\draw[thick] (1,2.25) rectangle (1.75,3);
\draw (1.375,2.625) node {G$_{\mbox{\scriptsize S$'$}}$}; 


\draw (0.75,1.4) -- (4,1.4);
\draw (0.75,1.5) -- (4,1.5);
\draw (0.75,1.6) -- (4,1.6);




\node at (1.75,2.75) {}  
  edge[pil,bend left=45] node[auto] {} (4,2.5) {}; 

\end{tikzpicture}

%% file: setup-2/attacks.tex
Let $L$ be a parking position, $A$ be the AC at $L$, and $G_S$ a
GU to provide service $S$ at $L$. Assume the attacker intends
to cause harm at $L$ by undermining the M2M communication between 
$A$ and $G_S$. There are five categories of attacks the attacker can 
aim for. Each category comes with different preconditions he needs 
to achieve first.  We say the attacker has reached a state from 
which he can mount a \ldots
\begin{enumerate}
\item
\emph{Man-in-the-Middle (MitM) attack} if $G_S$ has
a connection $(K, S)$ and $A$ has a connection $(K',S)$
but the attacker knows both $K$ and $K'$.
\item
\emph{Impersonation to GU (Imp2GU) attack} if $G_S$ has a
  connection $(K, S)$ but the attacker knows $K$.
\item
\emph{Impersonation to AC (Imp2AC) attack} if
  $A$ has a connection $(K, S)$ but the attacker knows $K$.
\item \emph{Mismatch of GU attack} if $G_S$ 
 has a connection $(K, S)$, and some $A'$ at parking slot
 $L'$ has a connection $(K, S')$ but $L \not= L'$
  or $S \not = S'$ while the attacker does not know $K$. 
  In the first case we speak of a \emph{location mismatch},
  in the second of a \emph{service mismatch}. There can also be
  location \emph{and} service mismatch attacks.
\item 
\emph{Mismatch of AC attack} is analogously defined. 
\end{enumerate}

The attack categories are illustrated in Fig.~\ref{fig:astates}.
We now discuss their potential safety impact by means of our
examples air conditioning and fuelling.


%% file: plain-2/plain.tex
\section{Plain TAGA}
\label{s:plain}



\subsubsection{Preliminaries}
\paragraph*{NFC System, Sessions, and Taps}
An \emph{NFC system (for TAGA)} is a triple $N = (\RGU, C, \RAC)$, where 
$\RGU$ is the reader at the GU, $C$ is the NFC card,  
and $\RAC$ is the reader at the AC. We assume readers are equipped with
a status display. 

An \emph{NFC data exchange session (NFC ssn)} is an event $\delta$
of a reader or card $D$, during which $D$ establishes an NFC link with
another party, say $D'$, exchanges data with $D'$ (in the view of
$D$), and then closes the link. If $D$ encounters an error during
any of these steps then $\delta$ is unsuccessful otherwise
it is successful. In the latter case, we denote $D'$ by $\partner(\delta)$,
and the trace of received and sent messages by $\msg(\delta)$. 
Often the data exchange 
will consist of a pair of messages $m_r$ and $m_s$ such that
$m_r$ is received by $D$ and $m_s$ is  sent by $D$. Then we denote
$m_r$ by $\msgr(\delta)$ and $m_s$ by $\msgs(\delta)$.
If $D$ is a reader it will signal on its display whether an 
NFC ssn has been completed successfully or not (e.g.\ green or red light). 
 
An \emph{NFC tap} is an event $\tau$ of an actor $O$, during which
$O$ brings a card $C$ close to a reader $R$ and verifies whether
$R$ displays that it has completed an NFC ssn successfully.  
We say $\tau$ is successful if $O$ observes a positive confirmation,
and unsuccessful otherwise. 
We denote $C$ by $\card(\tau)$ and $R$ by $\reader(\tau)$.

Given an NFC tap $\tau$ of an actor $O$ and an NFC ssn $\gamma$
of a reader $R$, we say $\tau$ and $\delta$ \emph{coincide}
when $\reader(\tau) = R$ and $O$ perceives the feedback that 
$R$ displays upon completion of $\delta$ as the reply to 
his tap $\tau$. We then write $\tapl{\tau}{\delta}$. 
Note that when $\tau$ and $\delta$
coincide it is a priori not given that $R$ has indeed communicated 
with $\card(\tau)$, even when $R$ confirms a successful
exchange.
%
Given NFC tap or ssn events $\epsilon_1$, $\epsilon_2$, we
write $\epsilon_1 < \epsilon_2$ when $\epsilon_1$ occurs
earlier in time than $\epsilon_2$. 

\paragraph*{TAGA Process and Sessions}
A concrete \emph{TAGA process} is defined by a triple
$T = (P, N, M)$, where $P$ is a KE protocol, $N$ is a NFC system,
and $M$ is a set of process measures.

A \emph{(TAGA) session} of a GU $G$ is a tuple 
$(S, \Gini, \Gfin, K)$, where $S$ is the service of $G$,
$\Gini$ and $\Gfin$ are successful NFC ssns of $G$'s reader such that 
$\Gini < \Gfin$, and $K$ is a key. 

A \emph{(TAGA) session} of
an AC $A$ is a tuple $(\Amid, S, K, \cst)$, where $\Amid$ is a
successful NFC ssn of $A$'s reader, $S$ is a service, $K$ is a key, and 
$\cst \in \{c,u\}$. $\cst$ indicates
the status of whether $A$ has already received the finish message
(sent via WLAN for key confirmation) or not: $c$ stands for confirmed, 
and $u$ stands for unconfirmed respectively.

A \emph{(TAGA) session} of an OP $O$ is a tuple $(\Wini, \Wmid, \Wfin)$,
where $\Wini$, $\Wmid$, and $\Wfin$ are successful NFC taps of $O$ 
such that (1)~$\Wini < \Wmid < \Wfin$, 
(2)~$\card(\Wini) = \card(\Wmid) = \card(\Wfin)$,
(3)~$\reader(\Wini) = \reader(\Wfin)$ is a GU reader,
(4)~$\reader(\Wmid)$ is an AC reader, and
(5)~there is no other NFC tap $\Wx$ of $O$ such that
$\Wini < \Wx < \Wmid$. 
Moreover, we assume that the OP's
actions are well-defined as spelled out in Def.~\ref{def:op:welld}.

\removed{
TODO We also allow partial sessions, 
allow partial sessions. write with $\ldots$ any 
A complete GU or AC session induces a connection in the obvious way.}

\subsubsection{Concept and Security}

The original idea behind TAGA is this: Since TAGA takes place in a 
secure zone it seems plausible that by a good choice of 
NFC system and, perhaps, additional process measures we can ensure 
that the TAGA transport guarantees an authentic message exchange 
(regarding both origin and message authenticity). On the one hand, 
this means that we can employ
an unauthenticated KE protocol such as the basic DH exchange to
securely establish a key, and TAGA will not depend on any PKI. 
On the other hand, this allows for a precise alignment
of setting up the secure channels and the physical connections.
More precisely, we hope to achieve a notion of security 
for TAGA that not only asserts the existence of secure keys but also 
that the NFC ssns and OP taps coincide in the expected way.

\begin{definition}
Let $T$ be a TAGA  process.
\item
We say $T$ guarantees \emph{secure TAGA I to a GU $G$} if, 
whenever $G$ has a session 
$(S, \Gini, \Gfin, K)$ then 
\begin{enumerate}
\item there is an OP with a session $(\Wini, \Wmid, \Wfin)$ such that 
      $\tapl{\Wini}{\Gini}$, and $\tapl{\Wfin}{\Gfin}$, 
\item there is an AC $A$ with a session $(\Amid, S, K, \ldots)$ 
such that $\tapl{\Wmid}{\Amid}$, and
\item $K$ is a fresh key known only to $G$ and $A$.
\end{enumerate}
\item
Similarly, we say $T$ guarantees \emph{secure TAGA I to an AC $A$} if, 
whenever $A$ has a session $(\Amid, K, S, c)$ 
then
\begin{enumerate}
\item there is an OP with a session $(\Wini, \Wmid, \Wfin)$ such that 
     $\tapl{\Wmid}{\Amid}$, 
\item there is a GU $G$ with a session
      $(S, \Gini, \Gfin, K)$ such that $\tapl{\Wini}{\Gini}$ and
      $\tapl{\Wfin}{\Gfin}$, and      
\item $K$ is a fresh key known only to $G$ and $A$.
\end{enumerate}
\item
We say $T$ guarantees \emph{secure TAGA I} iff $T$ guarantees
secure TAGA I to both parties, GUs and ACs.
\end{definition}

We first confirm that secure TAGA I indeed entails correct setup
when the physical setup is aligned with the TAGA walk.

\begin{definition}
A process flow $F$ for ground service $S$ is \emph{suitable for
secure TAGA I} if $F$ 
guarantees: When a TAGA walk has been completed between a GU $G$ of 
type $S$ and an AC $A$ then $G$ is physically connected to $A$ for $S$.
\end{definition}

\begin{theorem}
\label{thm:plain:secureSetup}
Let $F$ be a process flow for some ground service, 
and $T$ be a TAGA process such that $F$ includes $T$. 
If $F$ is suitable for secure TAGA~I and 
$T$ satisfies secure TAGA~I then $F$ guarantees correctness of
setup.  
\end{theorem}

\removed{
origin of origin of nfc sessn indeed from a operator TAGA session and
matching taps between the sessions and the walk: 
channel authenticity means that 
have matching conversations with respect to a corresponding
taga walk that ties the taga sessions together.
matching taps and messages.
matching taps between the sessions and the walk: 
matching conversations between the sessions:}

To show that secure TAGA I can indeed be obtained as sketched above
we first need to make precise what it means for the TAGA transport
to guarantee an authentic message exchange. Our notion of 
\emph{channel authenticity} asserts that if there is a GU session
then there must be an AC session with matching messages as well
as an operator session whose taps coincide with the 
NFC ssns of the GU session, and that of the AC session respectively.
And we require the analogue for the other direction.

\begin{definition}
Let $N$ be an NFC system, and $M$ be a set of process measures.
\item
We say $N$ under $M$ guarantees \emph{channel authenticity to a GU $G$} 
if, whenever $G$ has a session 
$(S, \Gini, \Gfin, \ldots)$ then there is some OP with a session
$(\Wini, \Wmid, \Wfin)$, and some AC with a session 
$(\Amid, ..)$ such that 
\begin{enumerate}
\item $\tapl{\Wini}{\Gini}$, $\tapl{\Wmid}{\Amid}$, and $\tapl{\Wfin}{\Gfin}$, 
      and
\item $\msgs(\Gini) = \msgr(\Amid)$, and $\msgs(\Amid) = \msgr(\Gfin)$.
\end{enumerate}
\item
Similarly, we say $N$ under $M$ guarantees \emph{channel auhenticity to 
an AC $A$} if, whenever $A$ has a session
$(\Amid, \ldots)$ then there is some OP with a
session $(\Wini, \Wmid, \ldots)$, and some GU with a
session $(S, \Gini,\ldots)$ such that 
\begin{enumerate}
\item $\tapl{\Wini}{\Gini}$, and $\tapl{\Wmid}{\Amid}$, and
\item $\msgs(\Gini) = \msgr(\Amid)$. 
\end{enumerate}
\item
We say $N$ under $M$ guarantees \emph{channel authenticity}
iff $N$ under $M$ guarantees this to both parties, GUs and ACs.
\end{definition}

\begin{definition}
We say a KE protocol $P$ is \emph{secure for plain TAGA} when $P$ 
guarantees secrecy and key freshness in the presence of passive adversaries.
\end{definition}

\begin{theorem}
\label{thm:plain:secureTAGA}
A TAGA process $T = (P, N, M)$ guarantees secure TAGA I if
$P$ is secure for plain TAGA, and $N$ under $M$ guarantees 
channel authenticity.
\end{theorem}

\input{plain-2/attacks}

\input{plain-2/sols}

\input{plain-2/resil}


\removed{
\begin{enumerate}
\item \emph{Channel Authenticity} to a GU $G$ if, whenever
$G$ has sent a message $M_1$, and accepted a message $M_2$ 
as first and second message of the same TAGA session then 
there is TAGA unit $U$ such that $U$ has received $M_1$ 
as first message, and transmitted $M_2$ has third tap of the
TAGA same session.

\item to an AC $A$ if, whenever $A$ has read a message $M_1$ and
written a message $M_2$ as TAGA message then there is a TAGA unit
$U$ such that $U$ has writen $M_1$ and read $M_2$ as mid tap
of a TAGA session. 

\item to $U$ if, whenever $U$ reads $M_1$ accepts as first tap then
it writes it as secodn tap, and when it reads $M_2$ as second tap 
it writes it as third tap. 
\end{enumerate}
}

%% file: plain-2/attacks.tex
\subsubsection{The Challenge of Channel Authenticity}
\label{s:plain:attacks}
We now illlustrate that it is rather challenging to obtain
channel authenticity. Even though TAGA takes place in a secure
zone, where only authorized personnel have access, 
the attacker has various indirect ways of compromising the 
TAGA transport, and combining them to MitM or other attacks. 

\paragraph*{Threats against the TAGA Channel}
\hspace{\fill}

\noindent
\emph{Swapping the Card.}
During a break a GU operator might reside in a less restricted area. 
An attacker with pickpocketing capability could  
use this as a window of opportunity to swap a counterfeit card
for the TAGA card (carried by the operator throughout the day).
Since pickpocketing can include social engineering such a card swap 
cannot be ruled out entirely even when the card must be worn physically 
attached to the operator,
e.g.\ by a lanyard. (A lanyard is a cord or strap worn around 
the neck or wrist for security.)

\emph{Eavesdropping.}
Another threat is that the attacker eavesdrops on the NFC exchange 
between the card and the GU or AC reader. While the nominal range 
of NFC communication is only 5-10 cm the range of eavesdropping can 
be considerably larger. In \cite{FK:05} Finke and Kelter first 
demonstrate that the communciation between a ISO-14443 reader 
and tag can be eavesdropped from 1-2 meters using a large loop
antenna. 
If a device is sending in active mode (like the GU and
AC reader) then it can even be eavesdropped on from a distance of up to
10 m \cite{HasBr:RFID06}.  
Moreover, based on simulation Kfir and Wool predict an eavesdropping 
distance of tens of meters when an active antenna is used \cite{Kfir:2005}.
At a parking slot, a large antenna could be hidden in large luggage, sports 
equipment, prams, wheelchairs, cages for pets, or instruments. 
A large antenna could even be hidden in clothing of an operator
such as a prepped safety vest swapped for the real one in a cafeteria.
Hence, it seems impossible that eavesdropping can be ruled out 
by practically enforceable measures in our settting.

\emph{Planting a Leech and Skimming.}
An important component in relay attacks against electronic
payment systems is a \emph{leech}, which is a device that fakes
a reader to a card. The leech is smuggled into a victim's pocket 
or bag so that it is close to her payment card. The leech can then 
power up the tag of the card, and communicate with it while remaining
undetected. The communication can be relayed by WLAN or 
cellular connection via an attacker's laptop to and from
a real payment terminal. A leech can also be employed as a skimmer 
to read out a tag (without writing to it). 

The success of a leech depends on whether it is
indeed able to activate the victim's card. This can be improved
by employing a `tuned leech' with a range beyond the nominal 5-10 cm,
e.g.\ by increasing the power of transmission and/or using larger
antenna. This has been shown to be possible: 
In \cite{Kfir:2005} Kfir and Wool predict a distance of
40-50 cm by simulation, and in \cite{KirWool:UsenixSec06} Kirschenbaum and Wool
have experimentally confirmed a skimming range of about 25 cm. 
In our setting, we cannot rule out that an attacker is able to
plant a leech device to skim or masquerade to the TAGA card:  
Similarly to card swapping an attacker could smuggle 
a leech device into an operators pocket while he is on a break. 

\emph{Malware Leech.}
A particular threat in our setting is that a standard device 
with NFC capability carried by authorized personnel 
is converted into a leech by malware. For example,
during the TAGA transport the operator might carry a smartphone
with NFC capability with her. An attacker with standard hacking 
capability could compromise the smartphone, e.g.\ by planting 
malware on the device via a Trojan app or a spear phishing email 
to the operator. While this allows for purely remote attacks, in contrast  
to a hardware leech, the success of a malware leech is constrained 
by the standard device's nominal range, and the operator
being in the habit of storing it close enough to the card.

\emph{Ghost Attacks.}  
The other important component in relay attacks against
electronic payment systems is a \emph{ghost}, which is
a device that fakes a card to a reader. TAGA seems protected
from ghost attacks by two factors: First, only authorized 
personnel can get close to the TAGA readers.
Second, TAGA employs passive cards, and one would expect that
for a ghost to communicate with a GU or AC reader it would have
to be within their nominal activation range.
However, in \cite{Kfir:2005} Kfir and Wool show how to build an 
\emph{active} ghost, which considerably increases the possible
distance to the reader. The idea is that the ghost transmits 
in a non-standard way using its 
own power but so that to the reader the transmission is indistinguishable from
that of a passive tag's. Based on simulation they  
predict a distance of up to 50 m for bi-directional communication.
To our knowledge this has not yet been experimentally confirmed
but it shows that a priori we cannot exclude ghost attacks. 
A ghost could either be remote or it could be infiltrated 
similarly to a leech device, in which case it can act from a 
closer distance.

\removed{
can be clearly formulated and realistically enforced is that he 
can eavesdrop on the NFC exchange between the card and the
TAGA controller (on either side). formula passive NFC. 
While it seems not easy for a practical attacker to do this,
there are many options: camouflage the antenna in bulky luggage 
(e.g. as a movable trampolin, surfboard), sth attached to the
operator herself (by social eningeering or swap of work jacket).
And not sufficient data available to rule out by a clear argument, 
perhaps because usually in secure applicaiton this is not necessary.
difficult to rule it out .}

\paragraph*{Attack Examples}
We now give examples of how these attack vectors can be
combined to carry out MitM attacks against a simple NFC system.
%
The first attack shows that the combination of
eavesdropping and card swapping is all the attacker needs to
obtain full MitM capability. 

\begin{figure}
\centering{
\input{figures/fig-attack-swap}
}
\caption{\label{fig:att:swap} Attack~\ref{att:swap}: 
MitM by Swap \& Eavesdrop}
\end{figure}
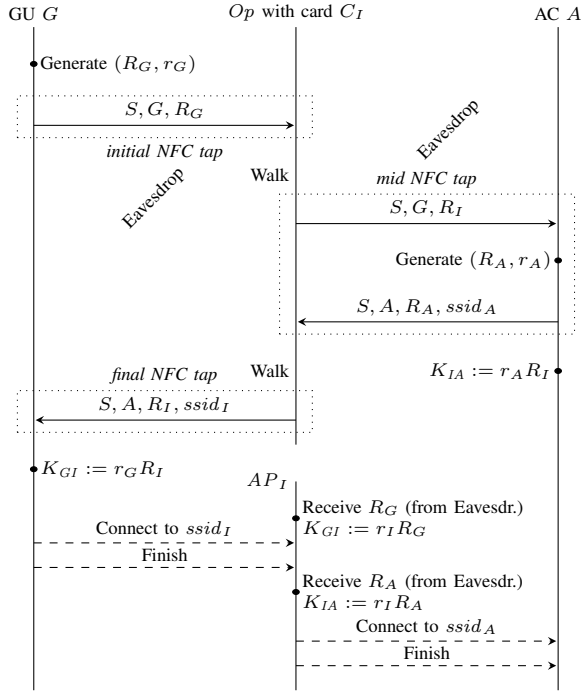

\begin{attack}[MitM by Swap \& Eavesdrop]
\label{att:swap}
Let $A$ be an AC and $G$ be a GU at parking slot $L$ so that
$G$ is to service $A$. In preparation, the attacker swaps
his own prepped card $C_I$ for the operator's card as described above. 
 Moreover, the attacker sets up NFC
eavesdropping capability, and his own WLAN access point $\AP_I$
in the range of $L$. Both $C_I$ and $\AP_I$ are prepped with
a fixed ephemeral key $(r_I, R_I)$, and an ssid $\ssid_I$ 
for the attacker's WLAN. 

The attack then proceeds as depicted in Fig.~\ref{fig:att:swap}. 
The card $C_I$ carries out the first tap as usual. However, with
the second tap the counterfeit card writes the attacker's public key 
$R_I$ to $A$ rather than $G$'s public key $R_G$. Similarly, with
the third tap the card writes $R_I$ and $\ssid_I$ to $G$ rather 
than $A$'s public key $R_A$ and ssid $\ssid_A$. Hence, $G$ 
computes session key $K_\GI$ based on $r_G$ and $R_I$, and
$A$ computes session key $K_\IA$ based on $r_A$ and $R_I$.
To be able to compute the same keys the attacker needs to
get $R_G$ and $R_A$ onto his access point $\AP_I$. Since the 
card only has a passive NFC interface he relies on eavesdropping to do so.
Once he has computed $K_\GI$ and $K_\IA$ he can establish the 
corresponding channels and mount a MitM attack.
\end{attack}

If the attacker has difficulties in setting up eavesdropping 
capability he can skim the keys $R_A$ and $R_G$ from the card
by a leech device instead. The leech could be planted on the 
operator at the same
time when the card is swapped, e.g. in a classical `pour drink over
shirt' attack.

The next attack demonstrates that a MitM attack might be
possible purely from remote, without having to break integrity or 
authenticity of the card. We assume that the operator is
in the habit of carrying her smartphone with NFC capability
in her pocket, and storing the TAGA card in the same pocket
next to the smartphone while she walks from the GU to the AC,
and similarly for the way back. This is a natural assumption: the 
operator typically needs both hands to carry the supply hose 
and might not like the dangling of a lanyard.

\removed{
\begin{figure}
\vspace{3cm}
\caption{\label{fig:att:phone} Attack~\ref{att:phone}}
\end{figure}}

\begin{attack}[MitM with Hacked Smartphone]
\label{att:phone}
Let $A$ be an AC and $G$ be a GU at parking position $L$ so
that $G$ is to service $A$. In preparation, the attacker 
hacks the smartphone of the operator of $G$, say $P_I$, and sets up
his WiFi point $\AP_I$ in the range of $L$. We assume that the
TAGA card is a simple store NFC card with a register Reg1 assigned
for the first message, and a register Reg2 for the second. 

The first NFC tap proceeds as usual, and writes $G$'s message 
into register Reg1 of the card. After the tap the
operator puts the card into her pocket next to her hacked
phone $P_I$. This triggers the NFC module of $P_I$, upon which it
reads the contents of Reg1 (with $G$'s public key) and then 
overwrites them with the attacker's own key $R_I$. 
Hence, when the operator arrives
at the AC and performs the (supposedly) second NFC tap
the AC reads the attacker's key $R_I$. The second message
pass is manipulated in the analogous fashion. The attacker
can easily relay the public keys $R_G$ and $R_A$
from the hacked smartphone to his access point $\AP_I$.
Now he can establish the corresponding channels, and act as MitM.
\end{attack}

We have implemented this attack against the current prototypical
implementation of TAGA. It works reliably as long as the 
operator observes her habit of
storing the TAGA card in the same pocket next to her
phone. If a more complex card than a simple store card is employed
this attack can be prevented by the card keeping 
state of the taps, and implementing input/output checks. However, such
measures can be overcome by a multi-session attack.

\begin{attack}[Multi-Session Variation]
When the phone $P_I$ first interacts with the card it will now 
perform three read/writes instead of one: the first impersonates
the AC and reads $G$'s public key; the second impersonates the GU 
during the final tap and concludes the current session of the card; 
the third starts a fresh session, where $P_I$ impersonates 
the GU and writes the attacker's key $R_I$.
The next tap will be carried out with the real AC in the appropriate
state. The AC has no means to notice that the attacker has skipped to
a new session and will accept $R_I$ as the GU's key.
On the way back the attacker can proceed in the analogous fashion,
and plant successfully his own key $R_I$ for the AC's key.
\end{attack}


Finally, we illustrate that if the attacker manages to install
a ghost then he can very easily stage attacks. In particular,
Imp2AC attacks only require one `fake tap'.

\begin{attack}[Imp2AC by Ghost]
Let $D_I$ be a remote or infiltrated ghost device that can 
reach the AC reader. Then $D_I$ can establish an NFC link with 
the AC, and masquerade as a TAGA card: $D_I$
writes the attacker's key $R_I$ to the AC, and obtains 
the AC's key $R_A$ in return. The resulting DH session key 
allows him to mount an Imp2AC attack against the AC.
\end{attack}

\removed{
\subsubsection{Full Dependency on GU Operator}
One assume card cannot be swapped and bound to the GU.
While the GU Operator is a trust anchor in ..
The following shows that by using two cards he can always
undermine the process.
he can do this fully covertly in his pocket. even when
under camera surveillance.
}

%% file: figures/fig-attack-swap.tex
\begin{scaletikzpicturetowidth}{\columnwidth}
\begin{tikzpicture}[scale=\tikzscale,
yscale=0.75,
nfc/.style={->,>=stealth,shorten >=0.025cm},
wlan/.style={->,>=stealth,shorten >=0.025cm,dashed},
point/.style={radius=0.05}, 
dot/.style={circle,draw=black,fill=black, radius=0.025}]
{\scriptsize


\draw (-4,1) node[above] {GU $G$};
\draw (0,1) node[above] {$\Op$ with card $C_I$};
\draw (4,1) node[above] {AC $A$};
\draw (0,-8.25) node[left] {$\AP_I$};

\draw (5,0) node[above] {};
\draw (-5,0) node[above] {};

\draw (0,1) -- (0,-7.5);
\draw (0,-8.25) -- (0,-12.5); 
\draw (-4,1) -- (-4,-12.5);
\draw (4,1) -- (4,-12.5);

\coordinate (G1) at (-4,0.25);
\coordinate (G2) at (-4,-1);
\coordinate (G3) at (-4,-7);
\coordinate (G4) at (-4,-8);
\coordinate (G5) at (-4,-9.5);
\coordinate (G6) at (-4,-10);

\coordinate (O1) at (0,-1);
\coordinate (O2) at (0,-2);
\coordinate (O3) at (0,-3);
\coordinate (O4) at (0,-5);
\coordinate (O5) at (0,-6);
\coordinate (O6) at (0,-7);

\coordinate (A1) at (4,-3);
\coordinate (A2) at (4,-3.75);
\coordinate (A3) at (4,-5);
\coordinate (A4) at (4,-6);
\coordinate (A5) at (4,-9);
\coordinate (A6) at (4,-10);

\draw[nfc] (G2) -- (O1);
  \coordinate (G2O1) at (-2,-1);
\draw[nfc] (O3) -- (A1);
  \coordinate (O3A1) at (2,-3);
\draw[nfc] (A3) -- (O4);
  \coordinate (O4A3) at (2,-5);
\draw[nfc] (O6) -- (G3);
  \coordinate (G3O6) at (-2,-7);

  \coordinate (G5A5) at (-2,-9.5);
\draw[wlan] (G5) -- (0,-9.5);
\draw[wlan] (G6) -- (0,-10);
  \coordinate (G6A6) at (-2,-10);

\draw[dotted] (-4.25,-0.4) rectangle (0.25,-1.25);
\draw (-2,-1.25) node[below] {\emph{initial NFC tap}};
\draw[dotted] (-0.25,-2.4) rectangle (4.25,-5.25);
\draw (2,-2.4) node[above] {\emph{mid NFC tap}};
\draw[dotted] (-4.25,-6.4) rectangle (0.25,-7.25);
\draw (-2,-6.4) node[above] {\emph{final NFC tap}};

\draw (-2,-2.7) node[above, rotate=45] {Eavesdrop};
\draw (2.5,-1.3) node[above, rotate=45] {Eavesdrop};

\draw (G1) node[right, align=left] {Generate $(R_G, r_G)$};
\draw[fill] (G1) circle [point];
\draw (G2O1) node[above] {$S, G, R_G$};

\draw (O2) node[left] {Walk};
\draw (O3A1) node[above] {$S, G, R_I$};
\draw (A2) node[left,align=left] {Generate $(R_A, r_A)$};
\draw[fill] (A2) circle [point];
\draw (O4A3) node[above] {$S, A, R_A, \ssid_A$};
\draw (O5) node[left] {Walk};
\draw (G3O6) node [above] {$S, A, R_I, \ssid_I$};
\draw (G4) node[right] {$K_\GI := r_G R_I$};
\draw[fill] (G4) circle [point];
\draw (A4) node[left] {$K_\IA := r_A R_I$};
\draw[fill] (A4) circle [point];
\draw (G5A5) node[above] {Connect to $\ssid_I$};
\draw (G6A6) node[above] {Finish};

\draw[wlan] (0,-11.5) -- (4,-11.5);
\draw[wlan] (0,-12) -- (4,-12);
\draw (2,-11.5) node[above] {Connect to $\ssid_A$ };
\draw (2,-12) node[above] {Finish};

\draw (0,-9) node[right,align=left] {Receive $R_G$ (from Eavesdr.)\\ 
        $K_\GI := r_I R_G$};
\draw[fill] (0,-9) circle [point];

\draw (0,-10.5) node[right,align=left]
        {Receive $R_A$ (from Eavesdr.)\\
        $K_\IA := r_I R_A$};
\draw[fill] (0,-10.5) circle [point];
}
\end{tikzpicture}
\end{scaletikzpicturetowidth}

%% file: plain-2/sols.tex
\subsubsection{Designs for the NFC System}
\label{s:plain:sols}

\paragraph*{Distribution Model}
The swap attack highlights that authenticity of the card is an important 
requirement. Hence, it is advisable that there is not just a pool of 
TAGA cards that are distributed daily to the operators but that a
distribution model is in place that allows for better accountability.
The following two models offer themselves.

In the \emph{OP-bound model} each GU operator is equipped with their 
personal card, which she can use on any GU that she is authorized
to operate. This model is particularly suitable when airport staff
are already provided with multi-functional access cards, into which
TAGA can be integrated as one of several applications. The card will
usually contain a public/private key pair that is bound to the
identity of the operator and a certificate for the public key
signed by the airport. 

In the \emph{GU-bound model} each GU is equipped with its own TAGA card. 
The card can be cryptographically bound to the GU reader so that
the GU can only be operated with its own card, and, conversely, 
the card will only work with the GU that owns it. To this end
a shared symmetric key can be pre-installed on both the card and 
the GU reader in a secure environment. 
During operative times the 
respective operator is responsible for 
handling the card while during non-operative times it is stored within 
the GU or within a secure compartment next to the GU reader.

\paragraph*{Securing the NFC Links between GU and Card}
Both of the models have the advantage that they come with 
long-term cryptographic keys that allow us to add mutual 
authentication and to secure the NFC links between the GU 
and the card. More precisely, when a new TAGA session 
is initiated with the first tap then the GU and card will run an 
AKE protocol to mutually authenticate each other and establish 
a session key. The session key will be used to secure the data 
exchange during the first tap as well as the final tap. 
In the GU-bound model authentication can be anchored in the 
symmetric key pre-installed on the card and GU reader. In
the OP-bound model this can be based on the public/private 
key pairs on the card and likewise on the GU reader and 
a PKI local to the airport.

\paragraph*{Detection of AC Masquerading}
Ensuring that the NFC exchange between the GU and the card
is authenticated has the advantage that AC masquerading (i.e.\
attacks where a leech device masquerades as the AC) can be detected.

A TAGA session of a card $C$ now consists of three steps:
First, $C$ performs an authenticated NFC exchange with a GU, say $G$. 
Second, $C$ carries out an unauthenticated NFC exchange with,
presumably, the AC. Third, $C$ carries out an authenticated NFC
exchange with, and only with, $G$ again. The implementation will
ensure that the card does not accept any other sequence of 
exchanges (apart from being reset by an authenticated GU).  
Hence, if the operator carries out 
the tap with the AC successfully then the second exchange cannot 
have resulted from a leech device because there can only be one
successful unauthenticated exchange. If a leech device masquerades
as the AC while the operator walks to the AC then her tap with
the AC will not be successful and she will abort the TAGA session
at the GU.

\paragraph*{Securing the NFC Link between Card and AC} 
Since the card and AC reader do not share any long-term keys 
(unless we break the `plain' paradigm) the NFC exchange between
the card and the AC cannot be secured with regards to origin 
authenticity. However, we can still protect the exchange from 
data modification as we now explain. 
Due to its channel characteristics
it is generally assumed that NFC is not susceptible to 
man-in-the-middle attacks, in particular when run in the 
active-passive mode. For example, compare \cite{HasBr:RFID06} 
for a precise argument.

\begin{assumption}
\label{ass:nfc}
It is pratically impossible to carry out a man-in-the-middle
attack over NFC in the active-passive mode.   
\end{assumption}

When the card and AC engage
in the mid tap then they can first run an unauthenticated KE protocol
such as the basic DH exchange and establish a session key to 
secure their data exchange.
\removed{
\footnote{
NFC security standards (NFC-SEC-01) specifies the following
mechanisms for use cases without common secret: ECDH for key agreement and
AES-CTR for data encryption and authentication.}}
While this will a priori not establish
that the AC securely communicates with the card, based on
Ass.~\ref{ass:nfc}, this will hold if, in addition, we can ascertain that 
the card is engaged in the same NFC link as the AC during the time. 
The other direction is analogous.


\paragraph*{Distance Bounding}
The features discussed so far still leave open that an attacker 
carries out remote attacks such as relay attack to undermine the 
NFC link between GU and card, or a ghost attack to masquerade
as a card to the AC. This can be prevented by distance bounding 
protocols \cite{mh13:db}, which guarantee an upper bound on the distance of
the card to the reader. To this end the reader measures the
round-trip time of cryptographic challenge/response pairs it
poses to the card. In our setting, the GU reader must bind this to the
authentication of the card to prevent distance hijacking \cite{Cremers:2012}. 
Simiarly, the AC reader must bind this to the ephemeral DH key 
of the card for the analogous reason.

\paragraph*{Secure Proximate Zone}
While distance bounding makes sure that the AC and GU reader
only communicate with a device close by process measures
we can ensure that only a TAGA card carried by an operator
comes closed to the reader. Together, this will ensure that
we obtain the coincidence between operator taps and NFC ssns
of card, AC, and GU as required for channel authenticity.

\begin{definition}
Let $d$ be the precision of the distance bounding protocol. 
$M$ guarantees \emph{secure proximate zone}, if the following holds:
no NFC device other than a card hold by an operator who
presumes this to be a TAGA card is brought within distance
$d$ of the GU reader, and AC reader respectively.
\end{definition}


\removed{
\begin{table}
\input{plain-2/tab-designs}
\caption{ \label{t:cards}}
\end{table}}

Altogether, we arrive at a specification of NFC systems
for the GU-bound model, and the OP-bound model respectively.
We denote the first by $\NGU$, and the second by $\NOP$. 

\begin{theorem}
\label{thm:plain:N}
The NFC systems $\NGU$ and $\NOP$ can be implemented so that
under process measures that include `secure proximate zone' 
they guarantee channel authenticity.
\end{theorem}

%% file: plain-2/tab-designs.tex
\begin{tabular}{lll}
 & GU-bound card & personal card \\
\hline
Distr & GU-bound & personal \\
Auth  & \parbox{3.5cm}{key $K$ shared between $G$ and $C(G)$ }
      & \parbox{3.5cm}{public/private key pair stored on $C$ and
                     PKI of airport domain}\\
GU link & \parbox{3.5cm}{$K'$ for each taga session symmetric $KE(K)$}
        & \parbox{3.5cm}{$K$ for each taga session KE based on $Kp/Kpriv$}\\
DB  & \parbox{3.5cm}{dB bound to $K$}
     & \parbox{3.5cm}{bound to $Kp$} \\
AC link & AES DH & AES DH\\
DB  & bound to $K$ & bound to $Kp$ \\
\end{tabular}

%% file: plain-2/resil.tex
\subsubsection{Multi-Instance Resilience and Discussion}
\label{s:plain:resil}

To sum up, we have obtained secure plain TAGA based on the following
insight: While the large secure zone established around an AC during 
turnaround leaves the TAGA channel open to many indirect attacks we 
can bootstrap channel authenticity from binding the card to the GU
(either one-to-one or via a PKI local to the airport), 
ensuring that the AC and GU 
readers accept requests only from a small \emph{proximate zone} and 
guaranteeing that this proximate zone is secure in a strong sense.

\begin{corollary}
The TAGA processes $T_\GU = (\bDH, \NGU, M)$ and $T_\OP = 
(\bDH, \NOP, M)$ both guarantee secure TAGA I.
\end{corollary}

Moreover, these TAGA processes come with a strong notion of
\emph{multi-instance resilience}. Let $T$ be a TAGA process, and
let $T_1$, $T_2$, ..., $T_n$ be $n$ instances of $T$.
Let for each $i \in [1,n]$, $G_i$ be the
GU, $A_i$ the AC, and $O_i$ the GU operator respectively.
We say the instances $T_1, \ldots, T_n$ are \emph{in parallel} if
the respective GUs, ACs, and GU operators are distinct:
i.e.\ for all $i, j \in [1,n]$, $G_i = G_j$ implies $i = j$,
and analogously for the ACs, and GU operators.

\begin{definition}
We say a secure TAGA process $T$ guarantees
\emph{resilience against parallel scaling of attacks} if
the effort for the attacker to successfully attack $n$ parallel 
instances of $T$ grows linearly with $n$ wrt a factor of physical
effort the attacker has to spend (e.g.\ work force, breach of a
physical security measure, compromise of personnel).
\end{definition}

\begin{theorem}
\label{thm:plain:resil}
$T_\GU$ and $T_\AC$ are resilient against
parallel scaling of attacks (assuming that system integrity of the GU, 
AC, and card is resilient in the same way). 
\end{theorem}

On the downside, plain TAGA comes with the caveat that its security 
heavily depends on the GU operator. Any security proof will rest on
the premise that the operator observes procedures
such as verifying that the AC has successfully received the tap,
and aborting TAGA at the GU when an error occurs. Moreover,
if a GU operator is compromised she can always mount a MitM
attack in the following way:  
While she carries out the GU taps with the authentic card as usual
she performs the AC tap with her own counterfeit card and
uses a leech device to masquerade to the authentic card as
AC reader. The latter can be carried out secretly in her pocket
sometime during the TAGA walk. Of course,
the GU operator is already a trusted party wrt safety of the turnaround.
Nevertheless, from a security viewpoint it seems unorthodox to make 
her a trusted third party in the key establishment. For example,
this has disadvantages
when it comes to IT forensics and liability claims.
A similar problem arises when ground units of different providers
with different security levels are to participate in TAGA: Plain TAGA 
assumes that all ground units are equally trustworthy. 

%% file: auth-2/auth.tex
\section{Authenticated TAGA}
\label{s:auth}

\input{auth-2/prel}

\input{auth-2/auth-secure}

\input{auth-2/challenge}

\removed{
\subsection{Concept}

We say a TAGA process $(P, N, M)$ is \emph{authenticated} 
when the KE protocol $P$ is authenticated. can do this fully
covertl 

while there are many AKE protocols can make sure 
securely establish a key between two parties in the presence
of adversaries there are two challenges:
 - the channel will be established/bound between cryptographic identities
    rather than physical entities. 
 - TAGA will be dependent on a PKI infrastrcutre with secure
    managmemnt and timely revocation of certificates.  what set out
    to avoid in the first place.

Revocation is the problem.

The first challenge is solved by relaxing the definition
of secure TAGA process while still making sure that Secure Setup
is guaranteed.

The second challenge is solved by ensureing ..
resilient to compromises of parties in other realms.

\subsection{Resi for ...}
assume ...

PKI: rather than makin
to send only valid certificates.
to use only private keys with good certificates.
and valid certificates.
rather than the other's responsibility to only use non-revoked
certificates of the other/any other party.

\begin{enumerate}
\item Private or secret keys are generated and only stored on the owner.
    never exported: only local parties can try to obtain them by
     side-channel/hardware/NFC attacks.
\item .. 
\end{enumerate}
}

%% file: auth-2/prel.tex
\subsubsection{Preliminaries}

\label{s:auth:prel}
In the setting of authenticated TAGA, every AC $A$ has a long-term 
key pair $(W_A, w_A)$, where
$W_A$ is the public key and $w_A$ is the private key. Moreover, $A$ holds a 
certificate for its public key $W_A$, which is issued by the airline $\mcA$
that owns $A$ (or an entity commissioned by $\mcA$). We denote the
certificate by $\cert_\mcA(A, W_A, T_A, V_A)$, where $T_A$ is the aircraft
type of $A$, and $V_A$ specifies the validity period of the certificate. 
 
Analogously, every GU $G$ has a long-term key pair $(W_G, w_G)$,
and a certificate for its public key $W_G$, which is issued by
the airport $\mcH$ that harbours $G$ (or an entity commissioned by
$\mcH$). We denote the certificate by $\cert_\mcH(G, W_G, S_G, V_G)$,
where $S_G$ is the service type of $G$ and $V_G$ is the validity
period of the certificate.

Moreover, we assume that every AC has installed the root certificates
of those airports it intends to land, and each GU has installed
the root certificates of those airlines it is authorized to handle.
Since a GU is specific to its airport, and an AC knows at which
airport it is currently parked (available from its electronic flight
system) we require that when an AC verifies a GU certificate
received during TAGA it will check that the airport within the
certificate agrees with its own location.

\begin{notation}
For short notation of certificates we often leave away the issueing party,
type of aircraft or service, and/or validity period 
when this is implicitly clear from the context.
\end{notation}

\removed{
\begin{assumption}
\label{ass:publicKeys}
Since certificates only contain public information we assume that
an attacker can obtain them. Usually he can also simply obtain
certificates by eavesdropping. The same applies to ssids of
WLANs.


\end{assumption}
}

%% file: auth-2/auth-secure.tex
\subsubsection{Concept and Security}
We say a TAGA process $(P, N, M)$ is \emph{authenticated} when $P$ 
is an authenticated KE (AKE) protocol. While authenticated TAGA 
allows us to establish secure keys without having to rely on channel 
authenticity without the latter we cannot expect to reach Secure 
TAGA I: We lose the guarantee that the OP taps and NFC ssns of the
GU and AC coincide in the expected way. However, we can reach another
notion of security for TAGA that, together with a straightforward
requirement on the setup process, allows for a match between 
cyber channels and physical
connections as required for secure setup.

\emph{Secure TAGA II} guarantees that when a GU $G$ and an AC $A$ 
successfully complete
a TAGA session then they are located at 
the same parking slot. This can be established by the underlying
AKE protocol in terms of data agreement when the GU and AC  
provide their location as input to the protocol. The requirement on the 
setup process is that the GU confirms that its physical setup is 
complete over the 
established cyber channel. The AC can then conclude that it is also
physically connected to its communication partner: $G$ has confirmed 
it is ready, and since $G$ is located at the same parking slot as
$A$, it must indeed be connected to $A$ and not to some other AC.

\begin{definition}
Let $T$ be a TAGA process.
\item
We say $T$ guarantees \emph{Secure TAGA II} to a GU $G$ if,
whenever $G$ has a session $(S, \Gini, \Gfin, A, K)$ then 
\begin{enumerate}
\item $A$ has a session $(\Amid, S, G, K, \ldots)$,
\item $K$ is a fresh key known only to $G$ and $A$,  and
\item $A$ is located at the same parking slot as $G$. 
\end{enumerate}
\item
Similarly, we say $T$ guarantees \emph{Secure TAGA II} to an AC $A$ if,
whenever $A$ has a session $(\Amid, S, G, K, c)$ then
\begin{enumerate}
\item $G$ has a session $(S, \Gini, \Gfin, A, K)$, 
\item $K$ is a fresh key known only to $A$ and $G$, and
\item $G$ is located at the same parking slot as $A$.
\end{enumerate}
\item
We say $T$ guarantees \emph{Secure TAGA II} iff $T$ guarantees
Secure TAGA II to both parties, GUs and ACs.
\end{definition}

\begin{definition}
\label{def:auth:suit}
A process flow $F$ for ground service $S$ is \emph{suitable for
secure TAGA II} if it guarantees: 
(1)~When a GU of type $S$ has completed its physical setup it 
confirms to the end point of its cyber channel for $S$
that it is ready.
(2)~Before an AC engages into service $S$ it waits until each
end point of its $n_F$ cyber channels for $S$ has confirmed that it is ready. 
\end{definition}

\begin{theorem}
\label{thm:setupAuth}
Let $F$ be a process flow for some service, and $T$ be a TAGA process such
that $F$ includes $T$. If $F$ is suitable for secure TAGA II
and $T$ satisfies secure TAGA II then $F$ guarantees correctness of setup.
\end{theorem}

To implement the requirement of Def.~\ref{def:auth:suit} the GU control needs  
to ``know'' when the physical setup is complete. This can be realized 
by means of a `ready button' at the GU to be pressed by the operator 
on completion of the setup. Alternatively,  a `ready signal' can be 
triggered when the GU's machinery (such as air supply or 
fuel pump) is activated by the operator, which is typically 
the last action of a GU setup. 

To reach Secure TAGA II we only need to choose an AKE protocol
that satisfies standard secrecy and authentication properties,
and extend it so that it also guarantees agreement on service and location. 
GUs and ACs can communicate their location explicitly as the 
number of the parking slot they are located at. In this case the 
parking slot needs to be provided by the GU operator, and pilot 
respectively, via a user interface. Alternatively, GUs and ACs can communicate 
their location in terms of GPS coordinates. This has the advantage that the data
can be provided automatically by their positioning system.
Due to safety distances that
need to be kept between ACs on the tarmac we expect that this will
be sufficiently precise to verify whether a GU and AC are on the
same parking slot. 

\removed{
We require the underlying protocol to satisfy standard
properties of KE protocols such as secrecy, key freshness, and
mutual authentication, which we capture as injective agreement between the
runs of the two parties \cite{lowe:auth97}. For TAGA we also
require that this includes agreement on location and service.
Moreover, the property
`Opposite Type' reflects that entities have a fixed role in our setting.
}

\begin{definition}
A KE protocol $P$ is \emph{secure for authenticated TAGA}
when $P$ satisfies secrecy, key freshness, opposite type, 
and agreement on peer, service and location in the presence of active
adversaries. (C.f.\ App.~\ref{app:prot:props})
\end{definition}

\begin{theorem}
\label{thm:tagaAuth}
Given any $N$ and $M$, $(P, N, M)$ guarantees secure TAGA II if 
$P$ is secure for authenticated TAGA.
\end{theorem} 

\begin{figure}
\centering{
\input{figures/prot-fhmqv}}
\caption{\label{prot:fmqv} TAGA pairing based on the FHMQV\CLS\ protocol}
\end{figure}
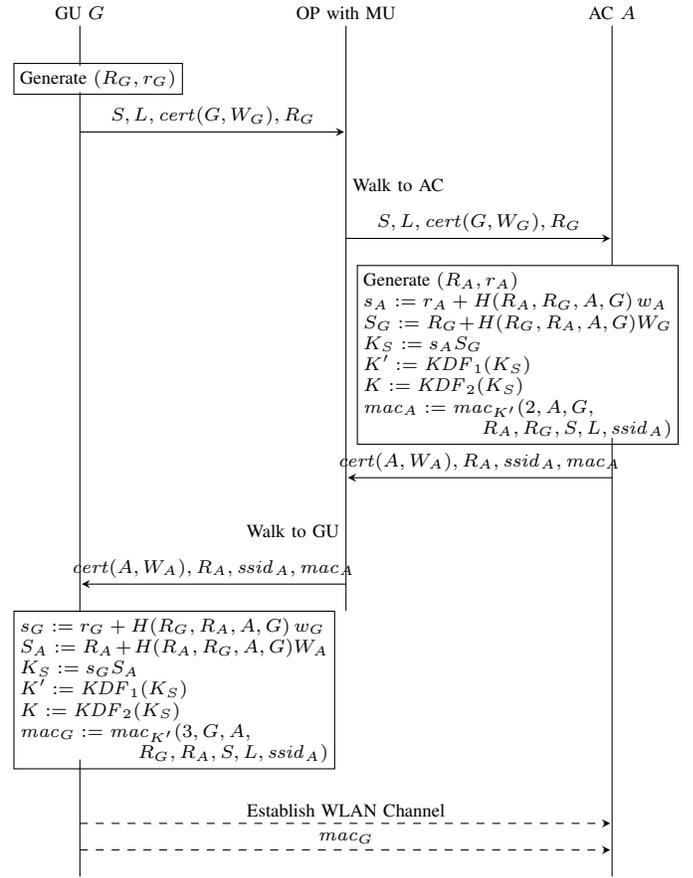

Fig.~\ref{prot:fmqv} shows TAGA based on the \emph{Fully Hashed 
Menezes-Qu-Vanstone protocol (FHMQV)} 
\cite{sarrEtAl:EuroPKI09,sarrEtAl:africacrypt16}. For TAGA we 
include service and location into the key confirmation step, yielding
\emph{FHMQV\CLS}. FHMQV is one of the strongest protocols 
regarding security, resilience and efficiency, 
and comes with a security proof. 

\begin{corollary}
Given any $N$ and $M$, $($FHMQV\CLS$, N, M)$ guarantees secure TAGA II. 
\end{corollary}

%% file: figures/prot-fhmqv.tex
\begin{scaletikzpicturetowidth}{\columnwidth}
\begin{tikzpicture}[scale=\tikzscale,
yscale=0.8,
nfc/.style={->,>=stealth,shorten >=0.025cm},
wlan/.style={->,>=stealth,shorten >=0.025cm,dashed},
dot/.style={circle,draw=black,fill=black, radius=0.025}]
{\scriptsize


\draw (-4,1) node[above] {GU $G$};
\draw (0,1) node[above] {OP with MU};
\draw (4,1) node[above] {AC $A$};

\draw (0,1) -- (0,-10);
\draw (4,1) -- (4,-3.5);
\draw (4,-6.75) -- (4,-15);

\draw (-4,1) -- (-4,0.25); 
\draw (-4,-0.25) -- (-4,-10);
\draw (-4,-12.8) -- (-4,-15);

\coordinate (G1) at (-5,0);
\coordinate (G2) at (-4,-1);
\coordinate (G3) at (-4,-9.5);
\coordinate (G4) at (-5,-10);
\coordinate (G5) at (-4,-14);
\coordinate (G6) at (-4,-14.5);

\coordinate (O1) at (0,-1);
\coordinate (O2) at (0,-2);
\coordinate (O3) at (0,-3);
\coordinate (O4) at (0,-7.5);
\coordinate (O5) at (0,-8.5);
\coordinate (O6) at (0,-9.5);

\coordinate (A1) at (4,-3);
\coordinate (A2) at (5,-3.5);
\coordinate (A3) at (4,-7.5);
\coordinate (A4) at (4,-6);
\coordinate (A5) at (4,-14);
\coordinate (A6) at (4,-14.5);

\draw[nfc] (G2) -- (O1);
  \coordinate (G2O1) at (-2,-1);
\draw[nfc] (O3) -- (A1);
  \coordinate (O3A1) at (2,-3);
\draw[nfc] (A3) -- (O4);
  \coordinate (O4A3) at (2,-7.5);
\draw[nfc] (O6) -- (G3);
  \coordinate (G3O6) at (-2,-9.5);
\draw[wlan] (G5) -- (A5);
  \coordinate (G5A5) at (0,-14);
\draw[wlan] (G6) -- (A6);
  \coordinate (G6A6) at (0,-14.5);


\draw (G1) node[draw,right] {Generate $(R_G, r_G)$};
\draw (G2O1) node[above] {$S, L, \cert(G, W_G), R_G$};

\draw (O2) node[right] {Walk to AC};
\draw (O3A1) node[above] {$S, L, \cert(G, W_G), R_G$};
\draw (A2) node[draw,text width=4.1cm,align=left,below left] 
 { Generate $(R_A, r_A)$\\
  $s_A := r_A + H(R_A, R_G, A, G)\, w_A$\\
  $S_G := R_G + H(R_G, R_A, A, G) W_G$\\
  $K_S := s_A S_G$\\
  $K' := \KDF_1(K_S)$ \\
  $K := \KDF_2(K_S)$\\
  $\mac_A := \mac_{K'}(2, A, G,$\\ \hfill $R_A, R_G,
     S, L, \ssid_A)$};

\draw (O4A3) node[above] {$\cert(A, W_A), R_A, \ssid_A, \mac_A$};
\draw (O5) node[left] {Walk to GU};
\draw (G3O6) node [above] {$\cert(A, W_A), R_A, \ssid_A, \mac_A$};
\draw (G4) node [draw, text width=4.1cm,align=left,below right]
 { $s_G := r_G + H(R_G, R_A, A, G)\, w_G$\\
  $S_A := R_A + H(R_A, R_G, A, G) W_A$\\
  $K_S := s_G S_A$\\
  $K' := \KDF_1(K_S)$ \\
  $K := \KDF_2(K_S)$\\
  $\mac_G := \mac_{K'}(3, G, A,$\\ \hfill $R_G, R_A,
     S, L, \ssid_A)$};

\draw (G5A5) node[above] {Establish WLAN Channel};
\draw (G6A6) node[above] {$\mac_G$};

}
\end{tikzpicture}
\end{scaletikzpicturetowidth}

%% file: auth-2/challenge.tex
\subsubsection{The Challenge of LTKC}

Let $A$ be an AC. We say the attacker has obtained a \emph{long-term 
key compromise (LTKC)} of $A$ if he has managed to get hold of 
credentials that authenticate $A$: A public/private key pair 
$(W_A, w_A)$ and a valid certificate $\cert(A, W_A)$, which asserts 
that $W_A$ belongs to $A$. The definition for a GU $G$ is analogous.
 
There are many different ways of how an attacker might obtain
a LTKC of a party $X$. One way is to obtain the private key of $X$
and use the existing certificate belonging to $X$. Another way is that the 
attacker tricks the certification authority (CA) into issueing a certificate 
for a key pair he has generated himself, perhaps for an entity that
does not even exist. This is more indirect but often easier to obtain; e.g.\  
when the private keys are generated within hardware security modules
and never exported
from there. Finally, if the attacker manages to get hold of the private key 
of the respective CA then he can generate as many valid certificates
for self-generated key pairs as he likes. All these cases have
been shown to be possible for real key management applications, e.g.\ by 
exploiting vulnerabilities in the APIs of the employed hardware
security modules \cite{AndersonBCS:survey,bfs:acm}.

Given a LTKC of an entity $X$ it is clear that the attacker can 
impersonate $X$ to other entities. In classical AKE settings this will only
impact on $X$'s resources or on peers that communicate specifically
with $X$ (based on $X$'s identity). However, in our setting of local
key establishment, 
a LTKC of an AC, say $A_I$, has wider consequences: 
Given any GU $G$ that is to
service the AC at some parking slot $L$, the attacker can impersonate
$\AC(L)$ to $G$ using the credentials of $A_I$ 
(c.f.\ App.~\ref{app:auth}, Fig.~\ref{fig:att:ltkc}).
%
Given a LTKC of a GU, say $G_I$, the impact is locally contained:
The attacker can exploit the compromised certificate of $G_I$ only within 
the realm of $G_I$'s airport.
This is because the AC will check that the GU certificate is issued for
the airport that it is currently parked at 
(c.f.\ Section~\ref{s:auth:prel}).
 
\removed{
\begin{figure}
\input{figures/att-ltkc}
\caption{\label{fig:att:ltkc} 
Impersonation to GU with LTKC of any AC}
\end{figure}

\begin{attack}[Impersonation to GU with LTKC of any AC]
\label{att:ltkc}
Let $A_I$ be a real or non-existent AC of airline $\mcA_I$, and assume
the attacker has achieved a LTKC of $A_I$. Further, let $A$ be an AC of 
airline $\mcA$, and $G$ be a GU at airport $\mcH$ such that 
$G$ provides service $S$ to $A$ during turnaround at parking slot $L$. 
In preparation, the attacker swaps his own counterfeit card
$U_I$ for the card of $G$'s operator.
 Moreover, the attacker sets up NFC eavesdropping capability,
and his own WLAN access point $\AP_I$ within range of $L$. Both 
$\AP_I$ and $U_I$ are prepped with $A_I$'s long-term 
credentials $w_I$ and $\cert(A_I, W_I)$, 
a fixed ephemeral key pair $(r_I, R_I)$, and an ssid $\ssid_I$ for
the attacker's WLAN.  

Then the attacker can proceed as shown in Fig.~\ref{fig:att:ltkc}:
he simply establishes a key with $G$ using $A_I$'s credentials rather
than those of $A$. Since $A_I$'s  ephemeral key pair can be fixed beforehand,
the resulting session key can be computed independently on the card $U_I$,
and the attacker's WLAN point $\AP_I$ respectively. The latter only 
needs to receive 
$G$'s public keys by relay from the eavesdropping device. Altogether, 
the attacker can now mount a potentially safety-critical M2M Imp2GU
attack.
\end{attack}
}

Say entity $X$ has a LTKC.
While, on the level of KE protocols, it cannot be prevented that 
the attacker can impersonate $X$ to any other participant, it is
important to realize that, in general, impersonation is not the only 
attack he can mount based on the LTKC of $X$. 
He might also be able to exploit the credentials of $X$ to impersonate 
any other party to $X$, or to mount an attack that violates authentication 
rather than secrecy. In the full version we show an example where,
based on a LTKC of a GU $G_I$, the attacker can stage a \emph{MitM} attack
against \emph{any} AC serviced by $G_I$. Another example shows
how, given a compromise of \emph{any} airline's root key pair, the attacker 
can reach a `Imp2GU with mismatch diversion attack' against \emph{any} two
TAGA instances that happen about the same time at some airport.
Fortunately, advanced AKE protocols such as the FHMQV are resilient
against LTKCs in a way that restores the expected correspondence
between LTKC and impersonation. 

\begin{definition}
Let $P$ be a protocol that is secure for a-TAGA. 
$P$ guarantees to an AC $A\,$ \emph{LTKC resilience for a-TAGA}
if $P$ guarantees secrecy, key freshness, opposite type and
agreement on peer, service and location to $A$ even when $A$ has a LTKC.  
The definition for a GU $G$ is analogous. $P$ is \emph{resilient
for a-TAGA} when $P$ guarantees this to both parties, ACs and GUs.
\end{definition}

\begin{lemma}
\label{lem:prot:resil}
Let $P$ be a protocol that is secure and LTKC resilient for a-TAGA.
Assume we lift the correctness assumption that long-term key
pairs of ACs and GUs are secure. 
\begin{enumerate}
\item If a GU has a session $(S, A, K)$ and the attacker knows $K$
      then $A$ must have a LTKC. 
\item If an AC has a session $(S, G, K, \ldots)$ and the attacker knows
 $K$ then $G$ must have a LTKC. 
\end{enumerate}
\end{lemma}

\removed{
In contrast, a LTKC of a GU $G$ does not have an equally wide impact:
Since a GU is specific to its airport, and an AC knows at which
airport it is currently parked (available from the electronic flight
system of the AC) an AC can easily validate that the airport within 
a presented GU certificate
agrees with its own location. This means a compromised GU certificate can only
be abused in the realm of its airport. 

While it is not possible to c. by. }

\removed{
As shown a single remotely 
obtained LTKC at airline $\mcA$ can impact on the safety of any AC of
$\mcA$.  

Note that this violates Principle~\ref{p:localAL} `Trust in Airline of
Local AC and no Other': a LTKC at airline $\mcA_I$ can impact on the
safety of an AC of airline $\mcA$. And even when $\mcA = \mcA_I$
this violates Principle~\ref{p:localACResil} 
`Resilience against Remote AC Security Incidence': 

Hence, it is essential
to harden authenticated TAGA against LTKC. 
We now
discuss other threats posed by LTKCs in this setting, and how 
we can still achieve resilient authenticated TAGA by a combination of advanced 
properties of AKE protocols and additional process measures that
exploit the local setting.  
}

\removed{
In a classical AKE setting the credentials of a user $A$ are typically
used as follows: either they govern access to resources bound to $A$ 
such as $A$'s money, data, or account on a computing system; or 
they establish $A$'s identity to peers that wish to communicate with
$A$ (and $A$ specifically). Given a LTKC of $A$ it is clear that the 
attacker can impersonate $A$. However, this will only impact on $A$'s 
resources or on peers that communicate with $A$ based on its identity.

In our setting the credentials of an AC $A$ are used in a different
way: they give access to establishing a secure channel with any GU 
at any airport that takes part in TAGA. This is so because
a priori the GU, say $G$, does not know the identity of $A/AC(L)$ but 
establishes the channel with any AC that has valid credentials 
and says they are at the same location as $G$. Hence, a LTKC of $A$
has wider consequences here: given any GU $G$ that is dispatched to service
an AC at location $L$, the intruder can impersonate $\AC(L)$ to $G$. 
Moreover, the attacker has the power to affect the safety of
this AC. as discusse din Section~\ref{} can be safety-critical. 
}

%% file: figures/att-ltkc.tex
\begin{center}
\begin{scaletikzpicturetowidth}{\columnwidth}
\begin{tikzpicture}[scale=\tikzscale,
yscale=0.75, 
nfc/.style={->,>=stealth,shorten >=0.025cm},
wlan/.style={->,>=stealth,shorten >=0.025cm,dashed},
dot/.style={circle,draw=black,fill=black, radius=0.025}]
{\scriptsize


\draw (-4,1) node[above] {GU $G$};
\draw (0,1) node[above] {OP with $U_I$};
\draw (4,1) node[above] {AC $A$};

\draw (0,1) -- (0,-7.5);

\draw (0,-10) -- (0,-12.5);

\draw (-0.1,-10) node[above right] {$\AP_I$};

\draw (4,1) -- (4,-3.75);
\draw (4,-4.25) -- (4,-10);

\draw (-4,1) -- (-4,0.25); 
\draw (-4,-0.25) -- (-4,-7.5);
\draw (-4,-10.25) -- (-4,-13);

\coordinate (G1) at (-4.6,0);
\coordinate (G2) at (-4,-1);
\coordinate (G3) at (-4,-7);
\coordinate (G4) at (-4.6,-7.5);
\coordinate (G5) at (-4,-11.5);
\coordinate (G6) at (-4,-12);

\coordinate (O1) at (0,-1);
\coordinate (O2) at (0,-2);
\coordinate (O3) at (0,-3);
\coordinate (O4) at (0,-5); 
\coordinate (O5) at (0,-6);
\coordinate (O6) at (0,-7);

\coordinate (A1) at (4,-3);
\coordinate (A2) at (4.5,-4);
\coordinate (Ospecial) at (-0.1,-4);
\coordinate (A3) at (4,-5);
\coordinate (A4) at (4,-6);
\coordinate (A5) at (4,-14);
\coordinate (A6) at (4,-14.5);

\draw[nfc] (G2) -- (O1);
  \coordinate (G2O1) at (-2,-1);
\draw[nfc] (O3) -- (A1);
  \coordinate (O3A1) at (2,-3);
\draw[nfc] (A3) -- (O4);
  \coordinate (O4A3) at (2,-5);
\draw[nfc] (O6) -- (G3);
  \coordinate (G3O6) at (-2,-7);
\draw[wlan] (G5) -- (0,-11.5);
  \coordinate (G5A5) at (-2,-11.5);
\draw[wlan] (G6) -- (0,-12);
  \coordinate (G6A6) at (-2,-12);


\draw (G1) node[draw,right] {Generate $(R_G, r_G)$};
\draw (G2O1) node[above] {$S, L, \cert(G, W_G), R_G$};

\draw (O2) node[right] {Walk to AC};
\draw (O3A1) node[above] {$S, L, \cert(G, W_G), R_G$};
\draw (A2) node[draw,text width=2cm,align=left,left]
 {Compute as usual}; 
\draw (4.5,-13)  node[draw,text width=6cm,align=left,below left] 
 {Algorithm $\ICrypto(G, W_G, R_G)$: \\
  \quad $s_I := r_I + H(R_I, R_G, A_I, G)\, w_I$\\
  \quad $S_G := R_G + H(R_G, R_I, A_I, G) W_G$\\
  \quad $K_S := s_I S_G$\\
  \quad $K' := \KDF_1(K_S)$; $K := \KDF_2(K_S)$\\
  \quad $\mac_I := \mac_{K'}(2, A_I, G, R_I, R_G,
      S, L, \ssid_I)$};

\draw (Ospecial) node[draw, minimum height = 2.5cm, text width=2.75cm, 
align=left,left]
{On card $U_I$: \\[3ex]
 Ignore $A$'s response \\
 Compute $A_I$'s ``response''\\  
 To obtain $\mac_I$ do\\ 
 \ \ $\ICrypto(G, W_G, R_G)$};

\draw (0.1,-11) node[draw,text width=3.5cm, right, align=left]
   {Receive $W_G$, $R_G$\\
      \quad (from Eavesdropper)\\
  To obtain $K$ do:\\
   \ \ $\ICrypto(G, W_G, R_G)$};

\draw (-1.5,0.3) node[above, rotate=45] {Eavesdrop};

\draw (G3O6) node[above] {$\cert(A_I, W_I), R_I, \ssid_I, \mac_I$};
\draw (O5) node[right] {Walk to GU};
\draw (O4A3) node [above] {$\cert(A, W_A), R_A, \ssid_A, \mac_A$};
\draw (G4) node [draw, text width=4.1cm,align=left,below right]
 {$s_G := r_G + H(R_G, R_I, A_I, G)\, w_G$\\
  $S_I := R_I + H(R_I, R_G, A_I, G) W_I$\\
  $K_S := s_G S_I$\\
  $K' := \KDF_1(K_S)$ \\
  $K := \KDF_2(K_S)$\\
  $\mac_G := \mac_{K'}(3, G, A_I,$ \\ \hfill $R_G, R_I,
     S, L, \ssid_I)$};

\draw (G5A5) node[above] {Establish WLAN Channel};
\draw (G6A6) node[above] {$\mac_G$};

}
\end{tikzpicture}
\end{scaletikzpicturetowidth}
\end{center}

%% file: resil/resil.tex
\subsubsection{Resilience against LTKCs}
\label{s:resil}

\input{resil/localResil}

\removed{
\subsection{Advanced Security Properties of AKE Protocols}
\removed{
We motive that resilient .. should guarantee several
advanced properties : Resilience against ..,
..., and The latter is but closely related to .
Definitions and proofs in the appendix.}

\input{resil/kci}

\input{resil/agree}

}
\input{resil/designs}

%% file: resil/localResil.tex

While advanced AKE protocols tame the effects of LTKCs, on the
protocol level, it is not possible to prevent an attacker 
from impersonating the entity with the LTKC to others. 
Neither is our setting amenable to fast and trustworthy 
certificate revocation and certificate status verification.
The reason is twofold:
First, this would require that all airlines and airports mutually trust
each other and apply equivalently strong security management to detect and
report LTKCs. Second, fast verification of the certificate status 
would require online connectivity of at least the airport ---
a disadvantage compared to the current paper- and human-based processes. 
Instead, in our setting of \emph{local} key establishment
we can harden TAGA against LTKCs by adding physical and 
local process measures. 

We require that TAGA at least guarantees 
\emph{basic LTKC resilience}: while when an AC $A$ and GU $G$
carry out TAGA together they must mutually trust each other 
and the key management processes of $A$'s airline and $G$'s airport
they will never be affected by a security incidence 
of another airline or airport.
\removed{
This ensures that
an airline $\mcA$ must only trust the key management processes of
airports where its ACs land (rather than all airports and airlines
participating in TAGA), and analogously an airport $\mcH$
must only trust the key management processes of those airlines whose
aircraft it harbours (rather than all airlines participating in 
This is important since another airport might participate in TAGA
but $A$'s airline might never have ACs land there, and similarly, another 
airline might participate in TAGA but $G$'s airport might never
harbour any of its planes. }
\emph{Airport-reliant resilience} is stronger in that an airport
and the local AC $A$ do not require to trust the key management processes
of $A$'s airline at all. This also answers to the situation when 
the internal certificate revocation process of an airline might not
reach its aircraft in time. 
\emph{Full resilience} ensures that TAGA has a second line of
defense and does not rely on long-term credentials as long as the local
measures are not compromised.  We now make precise the three
notions of resilience and provide examples of how they can be
met. 


\begin{definition}
Let $T$ be a a-TAGA process that guarantees secure TAGA II.
We say $T$ guarantees \emph{$x$ LTKC resilience}, where 
\emph{$x \in \{$basic, airport-reliant, full$\}$} 
if, whenever an instance of $T$ is carried out between
a GU $G$ and an AC $A$ at a parking slot $L$ then
$T$ guarantees secure TAGA II to $G$ and $A$ even when, 
\begin{enumerate}
\item for \emph{x = basic,} all ACs and GUs other than
those in the domain of $A$ and $G$ have a LTKC.
\item for \emph{x = airport-reliant,} all ACs and all GUs other
 than those in the domain of $G$ have a LTKC.
 \item for \emph{x = full,} all ACs and GUs have a LTKC. 
\end{enumerate}
\end{definition}


\input{resil/resil-basic}
\input{resil/resil-ac}

\input{resil/resil-full}

\removed{
We now show how such resiliance can be achieved by a combination
of advanced properties of the underlying AKE protocol and
local process measures. 

the airport this
 only "locally touching
parties" need mutually trust each other. The premise is that
then there is already trust.
And similarly the other way around.
But it will
be unaffected by any other airport's security even when they
also participate in TAGA. 
}


%% file: resil/resil-basic.tex
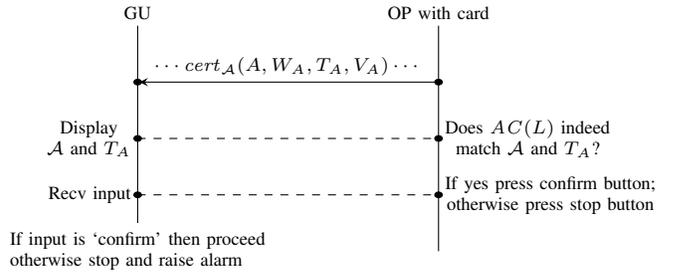
\begin{figure}
\centering{
\input{figures/fig-measure-human}}
\caption{\label{fig:meas:human} Last NFC tap extended 
by `Two Eyes' verification of AC}
\end{figure}

\paragraph{Basic LTKC Resilience}
To reach basic LTKC resilience we only need to ensure that the attacker cannot
make the GU accept a certificate that does not agree with the domain
of the AC actually on site.

\begin{definition}
A measure $M$ guarantees \emph{agreement of AC domain} if, 
when a GU has a session $(S, A, K)$ at a parking slot $L$ then $A$ 
is of the same airline and type as $\AC(L)$.
\end{definition}

\begin{theorem}
\label{thm:resil:basic}
Let $T = (P, N, M)$ be any TAGA process such that $P$ is secure and
LTKC resilient for a-TAGA and $M$ guarantees agreement of AC domain. 
Then $T$ guarantees basic LTKC resilience.
\end{theorem}

Agreement of AC domain can be implemented by a simple measure that
makes use of the awareness of the GU operator: 
while the GU has no means to verify that the received AC certificate 
(and information therein) belongs to the AC present at the parking 
slot, clearly, the operator has sight of the AC. Hence, 
she is able to verify that visually observable features of the AC 
such as its type and airline agree with the information 
received by the GU. 

\begin{measure}[`Two Eyes' Verification of AC (2EV)]
\label{m:human:two}
Assume the TAGA controller of the GU is equipped with a display and
two input buttons: one to confirm, and the second to stop the 
process and raise an alarm. Then the last NFC tap can be extended by
human verification as illustrated in Fig.~\ref{fig:meas:human}. First, 
the operator transfers the second message by NFC tap to the 
GU's controller as usual. Recall that this message contains an AC certificate
$\cert_\mcA(A, W_A, T_A, V_A)$, where $\mcA$ is the airline of $A$ and
$T_A$ the type of $A$, supposedly matching $\AC(L)$. 
Second, the GU shows $\mcA$ and $T_A$ on its display, and
the GU operator verifies whether the aircraft $\AC(L)$ she sees on 
the parking slot is indeed of airline $\mcA$ and type $T_A$. 
If yes, then she will confirm the process; otherwise she will
stop the process and raise an alarm. 
\end{measure}

Unintended errors of the GU operator can be kept small: 
they can be trained to keep awareness by injection of false alarms 
(similary to security screening at airports). It is also possible
to implement this with dual-control.

\begin{measure}[`Four Eyes' Verification of AC (4EV)]
\label{m:human:four}
For increased security a member of the AC crew can accompany the GU 
operator and perform the visual verification as well.  
\end{measure}


\begin{fact}
2EV and 4EV guarantee agreement of AC domain.
\end{fact}

%% file: figures/fig-measure-human.tex
\begin{scaletikzpicturetowidth}{\columnwidth}
\begin{tikzpicture}[
yscale=0.75,
nfc/.style={->,>=stealth,shorten >=0.025cm},
phy/.style={>=stealth,shorten >=0.025cm,dashed},
wlan/.style={->,>=stealth,shorten >=0.025cm,dashed},
point/.style={radius=0.05},
dot/.style={circle,draw=black,fill=black, radius=0.025}]
{\scriptsize

\draw (-2,0) node[above] {GU};
\draw (2,0) node[above] {OP with card};

\draw (-2,0) -- (-2,-3.5);
\draw (2,0) -- (2,-4);

\coordinate (G1) at (-2,-1); 
\coordinate (O1) at (2,-1);
\coordinate (O1G1) at (0,-1);

\coordinate (G2) at (-2,-2); 
\coordinate (O2) at (2,-2);

\coordinate (G3) at (-2,-3); 
\coordinate (O3) at (2,-3);

\coordinate (G4) at (-2,-4);

\draw[nfc] (O1) -- (G1);
\draw (O1G1) node[above] {$\cdots \cert_\mcA(A, W_A, T_A, V_A) \cdots$};

\draw[phy] (G2) -- (O2);
\draw (G2) node[align=center,left] {Display \\$\mcA$ and $T_A$};
\draw (O2) node[align=center,right] {Does $\AC(L)$ indeed \\match
         $\mcA$ and $T_A$?};

\draw[phy] (G3) --  (O3);
\draw (O3) node[align=center,right] {If yes press confirm button;\\
  otherwise press stop button};
\draw (G3) node[left] {Recv input};
\draw (G4) node[align=left] {If input is `confirm' then proceed\\
    otherwise stop and raise alarm};

\draw[fill] (G1) circle [point];
\draw[fill] (G2) circle [point];
\draw[fill] (G3) circle [point];
\draw[fill] (O1) circle [point];
\draw[fill] (O2) circle [point];
\draw[fill] (O3) circle [point];
}
\end{tikzpicture}
\end{scaletikzpicturetowidth}

%% file: resil/resil-ac.tex
\begin{figure}
\centering{
\input{figures/fig-measure-phy}
}
\caption{\label{fig:meas:phyCR} Physical Challenge/Cyber Response (PC/CR)}
\end{figure}
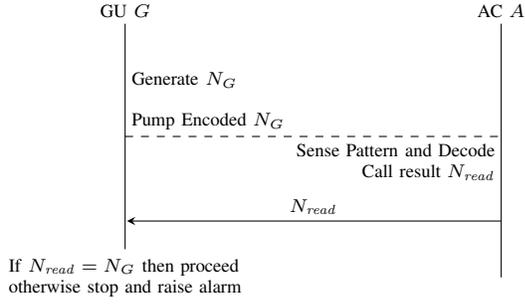

\paragraph{Airport-reliant LTKC Resilience}
Airport-reliant LTKC resilience will be reached if we make sure
that the GU can detect pure Imp2GU attacks. This is so because
by Lemma~\ref{lem:prot:resil} another
attack will always require a LTKC of the GU domain.

\begin{definition}
A measure $M$ guarantees \emph{detection against pure Imp2GU} if, whenever 
a GU $G$ has a session $(S, A, K)$ at a parking slot $L$ and the 
attacker knows $K$ then $G$ detects his and aborts the session (before
any safety-critical process settings are started)
as long as $\AC(L)$ does not have any session 
$(S, G', K', c)$ for some $G'$ and $K'$ such that the attacker 
knows $K'$. 
\end{definition}

\begin{theorem}
\label{thm:resil:detect}
Let $T = (P, N, M)$ be any TAGA process such that $P$ is secure and
LTKC resilient for a-TAGA and $M$ guarantees detection against 
pure Imp2GU. Then $T$ guarantees airport-reliant LTKC resilience.
\end{theorem}

There are several ways to implement detection against pure Imp2GU.
The following measure translates the standard scheme of challenge/response
authentication into the concept of \emph{physical challenge/cyber response:}
The GU sends a challenge via the physical connection, e.g.\ encoded 
in a pattern of pulsating flow, which the AC must answer via the 
cyber channel. Thereby the physical
connection is directly bound into the KE process. 

\begin{measure}[Physical Challenge/Cyber Response (PC/CR)]
\label{m:phyCR}
Assume that the airpacks of the AC are 
equipped with mass airflow sensors that can detect a pattern of
airflow changes and report it to the AC controller.
Then a phase of physical challenge response can be included 
between the setup phase and the M2M phase as illustrated in 
Fig.~\ref{fig:meas:phyCR}.
$G$ generates a random number of a fixed size, say $N_G$, and encodes 
this into a pattern of pulsating airflow.
$A$ reads the physical signal by the airflow sensors and decodes it
back into a number, say $N_\nread$. $A$ then responds by sending 
$N_\nread$
back to $G$ via the cyber channel. $G$ checks whether $N_\nread = N_G$.
If this is true then $G$ concludes that it speaks to 
$\AC(L)$: only the AC that is physically connected
to the GU could have known $N_G$. If the numbers don't agree $G$ stops
and raises an alarm.
\end{measure}

\begin{lemma}
\label{lem:resil:pccr}
PC/CR guarantees detection against pure Imp2GU.
\end{lemma}

The space of nonces must be sufficiently large to reduce the risk of
guessing attacks: Even when the attacker cannot receive the physical
signal he can always guess the nonce $N_G$ and send it back via a cyber 
channel he has established with the GU by an impersonation attack.
This brings about a trade-off between security and efficiency. For
example:
Say the physical channel allows a binary encoding of numbers in terms of
high and low airflow (e.g.\ using stuffing to synchronize). Say an  
encoded bit requires 2 seconds to be
transmitted, and a challenge shall maximally take 10 (or 20) seconds to be
transmitted. Then one can use a space of 32 (or 1024) nonces, and the
attacker has a 1/32 (or 1/1024) chance to guess correctly.

%% file: figures/fig-measure-phy.tex
\begin{scaletikzpicturetowidth}{\columnwidth}
\begin{tikzpicture}[
yscale=0.75,
nfc/.style={->,>=stealth,shorten >=0.025cm},
phy/.style={>=stealth,shorten >=0.025cm,dashed},
wlan/.style={->,>=stealth,shorten >=0.025cm},
point/.style={radius=0.05},
dot/.style={circle,draw=black,fill=black, radius=0.025}]
{\scriptsize

\draw (-2.5,0) node[above] {GU $G$};
\draw (2.5,0) node[above] {AC $A$};

\draw (-2.5,0) -- (-2.5,-4);
\draw (2.5,0) -- (2.5,-4.5);

\draw (-2.5,-1) node[right] {Generate $N_G$};

\draw[phy] (-2.5,-2) -- (2.5,-2);
\draw (-2.5,-2) node[above right] {Pump Encoded $N_G$}; 
\draw (2.5,-2) node[below left,align=right] 
  {Sense Pattern and Decode\\
   Call result $N_\nread$}; 

\draw (2.5,-3) node {};

\draw[wlan] (2.5,-3.5) -- (-2.5,-3.5);
\draw (0,-3.5) node[above] {$N_\nread$};

\draw (-2.5,-4.5) node[align=left] 
  {If $N_\nread = N_G$ then proceed\\
   otherwise stop and raise alarm};

\removed{
\coordinate (G1) at (-2,-1); 
\coordinate (O1) at (2,-1);
\coordinate (O1G1) at (0,-1);

\coordinate (G2) at (-2,-2); 
\coordinate (O2) at (2,-2);

\coordinate (G3) at (-2,-3); 
\coordinate (O3) at (2,-3);

\coordinate (G4) at (-2,-4);

\draw[nfc] (O1) -- (G1);
\draw (O1G1) node[above] {$\cdots \cert_\mcA(A, K_A, T_A) \cdots$};

\draw[phy] (G2) -- (O2);
\draw (G2) node[align=center,left] {Display \\$\mcA$ and $T_A$};
\draw (O2) node[align=center,right] {Does $\AC(L)$ indeed \\match
         $\mcA$ and $T_A$?};

\draw[phy] (G3) --  (O3);
\draw (O3) node[align=center,right] {If yes press confirm button;\\
  otherwise press stop button};
\draw (G3) node[left] {Recv input};
\draw (G4) node[align=left] {If input is confirm then proceed\\
    otherwise stop and raise alarm};

\draw[fill] (G1) circle [point];
\draw[fill] (G2) circle [point];
\draw[fill] (G3) circle [point];
\draw[fill] (O1) circle [point];
\draw[fill] (O2) circle [point];
\draw[fill] (O3) circle [point];
}
}
\end{tikzpicture}
\end{scaletikzpicturetowidth}

%% file: resil/resil-full.tex
\paragraph{Full LTKC Resilience}
Full LTKC resilience can be reached if the TAGA process guarantees 
channel authenticity as a second line of defense without relying on
any central key management. 
We only neeed to make sure that the underlying AKE protocol is suitable 
for this in the following sense:


\begin{definition}
Let $P$ be a protocol that is secure for a-TAGA.
$P$ guarantees \emph{LTKC resilience against local eavesdropping
(eav-LTKC resilience) for a-TAGA} if, whenever an instance of $P$
is run between a GU $G$ and an AC $A$ then, even when both
$G$ and $A$ have a LTKC, $P$ guarantees secrecy, 
key freshness, opposite type, and agreement on peer, service 
and location to $G$ and $A$ 
against an attacker that can only
eavesdrop on their current session of $P$ (but actively intervene 
in all other sessions as usual).
\end{definition}

\begin{corollary}
Let $T = (P, N, M)$ be a TAGA process such that $P$ is secure 
and eav-LTKC resilient for a-TAGA, and $N$ under $M$ 
guarantees channel authenticity without relying on any central
key management.  Then $T$ guarantees full LTKC resilience.
\end{corollary}

\begin{fact}
$N_\GU$ under $M_T$ guarantees channel authenticity without relying
on any central key management.
\end{fact}

%% file: resil/kci.tex
\subsubsection{The Significance of KCI Resilience}

Given that a participant $X$ has a LTKC, it is clear that 
the attacker can impersonate $X$ to any other participant. However,
another question to ask is whether this enables the attacker to 
impersonate any other participant to $X$. We then say the attacker 
can carry out a \emph{Key-Compromise Impersonation (KCI)} attack 
\cite{first:kci}. In our setting this has the following
consequence: if the TAGA protocol 
$P$ is not resilient against KCI attacks then, 
given a LTKC of GU $G_I$, the attacker will be able to stage an Imp2GU 
attack against any
AC serviced by $G_I$. Moreover, the attacker can combine each such KCI
attack with a standard impersonation attack to obtain M2M MitM capability. 
We illustrate this by a concrete example based 
on the \emph{Unified Model (UM)} protocol \cite{prot:um}. The UM,
shown in Table~\ref{prot:um}, is another DH protocol 
with implicit authentication. Moreover, it is well-known to be vulnerable 
to KCI attacks \cite{prot:umcki}. 

\begin{table}
\begin{enumerate}
\item $G$ generates $(R_G, r_G)$\\
      $G$ sends $S, L, \cert(G, W_G), R_G$

\item $A$ receives and validates the message\\ 
      $A$ generates $(R_A, r_A)$\\
      $A$ computes $K := H(w_A W_G\, ||\, r_A R_G)$\\
      $A$ computes $\mac_A := \mac_K(2, A, G, R_A, R_G, S, L, \ssid_A)$\\
      $A$ sends $\cert(A, W_A), R_A, \ssid_A, \mac_A$

\item $G$ receives and validates the message \\
      $G$ computes $K := H(w_G W_A\, ||\, r_G R_A)$\\
      $G$ validates $\mac_A$\\
      $G$ computes $\mac_G := \mac_K(3, G, A, R_G, R_A, S, L, \ssid_A)$\\
      $G$ establishes the WLAN connection and sends $\mac_G$.
\end{enumerate}
\caption{\label{prot:um} The UM Protocol}
\end{table}

\begin{figure}
\input{figures/fig-attack-kci}
\caption{\label{fig:att:kci} Attack~\ref{att:kci}:
KCI attack against the UM protocol}
\end{figure}

\begin{attack}[KCI attack against UM]
\label{att:kci}
Let $G_I$ be a GU for service $S$ at airport $\mcH$, for which
the attacker has achieved a LTKC. Further, let $A$ be any AC that is serviced
by $G_I$, say at parking slot $L$. 
In preparation, the attacker swaps the NFC card of 
$G_I$'s operator with his own prepared card $U_I$. Moreover, he sets
up NFC eavesdropping capability, and his own WLAN access point $\AP_I$ 
within range of $L$. Both $\AP_I$
and the card $U_I$ are prepped with a fixed ephemeral
key pair $(r_I, R_I)$ and an ssid $\ssid_I$ for the WLAN with $G_I$. 
In addition,
$\AP_I$ is prepped with $G_I$'s credentials $w_I$ and $\cert(G_I, W_I)$,
and $A$'s long-term public key $W_A$ (c.f.\ Ass.~\ref{ass:publicKeys}).
Alternatively, $W_A$ can be obtained by eavesdropping.
  
Then the attacker can proceed as shown in Fig.~\ref{fig:att:kci}.
The interaction with $G_I$ consitutes a KCI attack, where the attacker 
impersonates $A$: he can compute the same key $K_\GI$ as $G_I$ by using
his knowledge of $w_I$ rather than $w_A$ (and his own ephemeral key pair).
The interaction with $A$ constitutes a standard impersonation attack, 
where the attacker can impersonate $G_I$ due to his knowledge of $w_I$,
and establishes a key $K_\IA$ with $A$.  
Altogether, the attacker can now fully
control the M2M communication of $A$ and $G$ as the MitM.
\end{attack}

This illustrates that as long as KCI attacks are a threat it is
not possible to protect against one-sided LTKCs by detection measures
against impersonation attacks. For this we require that the
protocol guarantees \emph{KCI Resilience}. Provided that the ephemeral
private keys are secure, this will guarantee that the
natural correspondence between LTKCs and secrecy attacks one
would expect is indeed in place: 
(1)~To impersonate as the local AC to $G$ the attacker requires to 
obtain an LTKC of some AC.
(2)~To impersonate as the local  GU for service $S$ to $A$ the attacker requires
to obtain a LTKC of some GU of that domain.
And
(3)~to act as a MitM between $G$ and $A$ the attacker requires to
obtain a LTKC of some AC and a LTKC of some GU of the corresponding domain.



%% file: figures/fig-attack-kci.tex
\begin{center}
\begin{scaletikzpicturetowidth}{\columnwidth}
\begin{tikzpicture}[scale=\tikzscale,
nfc/.style={->,>=stealth,shorten >=0.025cm},
wlan/.style={->,>=stealth,shorten >=0.025cm,dashed},
dot/.style={circle,draw=black,fill=black, radius=0.025}]
{\scriptsize

\draw (0,-1.25) -- (0,-2.75);
\draw (0,-3.25) -- (0,-6.25);
\draw (0,-7.75) -- (0,-8.5);

\draw (0,-9) -- (0,-13.25);

\draw (-4,-1.25) -- (-4,-13.25);
\draw (4,-1.25) -- (4,-13.25);

\draw (0,-1.25) node[right] {OP with $U_I$};
\draw (-4,-1.25) node[right] {$G_I$};
\draw (4,-1.25) node[left] {$A$};
\draw (0,-9) node[right] {$\AP_I$};

\coordinate (G1) at (-4,-1.75);
\coordinate (G2) at (-4,-2.5);  
\coordinate (G3) at (-4,-8.25);  
\coordinate (G4) at (-4,-9.25);  

\coordinate (I1) at (0,-2.5); 
\coordinate (I2) at (0,-3);
\coordinate (I3) at (0,-3.75); 
\coordinate (I4) at (0,-6); 
\coordinate (I5) at (0,-7); 
\coordinate (I6) at (0,-8.25); 

\coordinate (A1) at (4,-3.75);
\coordinate (A2) at (4,-4.75);
\coordinate (A3) at (4,-6); 

\coordinate (G2I1) at (-2,-2.5);
\coordinate (I3A1) at (2,-3.75);
\coordinate (I4A3) at (2,-6);
\coordinate (I6G3) at (-2,-8.25);

\draw[nfc] (G2) -- (I1);
\draw[nfc] (I3) -- (A1);
\draw[nfc] (A3) -- (I4);
\draw[nfc] (I6) -- (G3); 

\draw (G2I1) node[above] {$S, L, \cert(G_I,W_I), R_G$};
\draw (I3A1) node[above] {$S, L, \cert(G_I,W_I), R_I$};
\draw (I4A3) node[above] {$\cert(A,W_A), R_A, \ssid_A, \mac_A$};
\draw (I6G3) node[above] {$\cert(A,W_A), R_I, \ssid_I, \mac_\GI$};

\draw (G1) node[right,draw] {Generate $(R_G, r_G)$};

\draw (I2) node[draw, minimum height=0.5cm] {Swap $R_G$ for $R_I$};

\draw (A2) node[left,align=left,draw] { 
 Generate $(R_A, r_A)$\\
 $K_\IA := H(w_A W_I\, ||\, r_A R_I)$\\
 $\mac_A := \mac_{K_\IA}(2,A,G_I,$\\
       \hfill $R_A,R_I,S,L,\ssid_A)$};

\draw (I5) node[draw, align=left, minimum height=1.5cm] {
 $K_\GI := H(w_I W_A\, ||\, r_I R_G)$\\
 $\mac_\GI := \mac_{K_\GI}(2,A,G_I$\\
      \hfill $R_I,R_G,S,L,\ssid_I)$\\
 Swap $R_A$ for $R_I$, $\mac_A$ for $\mac_\GI$\\
  \hfill $\ssid_A$ for $\ssid_I$} ;

\draw (G4) node[right,align=left, draw] {
  $K_\GI := H(w_I W_A\, ||\, r_G R_I)$\\
  Verify $\mac_\GI$\\
  $\mac_G := \mac_{K_\GI}(3,G_I,A,$\\
     \hfill $R_G,R_I,S,L,\ssid_I)$
  };

\draw (0,-9.5) node[below right,align=left,draw] {
  $K_\GI := H(w_I W_A\, ||\, r_I R_G)$};

\draw[wlan] (-4,-10.35) -- (0,-10.35);
\draw[wlan] (-4,-10.75) -- (0,-10.75);
\draw (-2,-10.35) node[above] {Connect to WLAN with $\ssid_I$};
\draw (-2,-10.75) node[above] {$\mac_G$};

\draw (0,-11) node[below right,align=left,draw] {
   Receive $R_A$ (from Eavesdropper)\\
   $K_\IA := H(w_I W_A\, ||\, r_I R_A)$\\
   $\mac_\IA := \mac_{K_\IA}(3,G_I,A,$\\
   \hfil $R_I,R_A,S,L,\ssid_A)$
   };

\draw[wlan] (0,-12.6) -- (4,-12.6);
\draw (2,-12.6) node[above] {Connect to WLAN with $\ssid_A$};
\draw[wlan] (0,-13) -- (4,-13);
\draw (2,-13) node[above] {$\mac_\IA$};

\draw (2.5,-6) node[below right, rotate=-45] {Eavesdrop};


\removed{
\draw[fill] (G1) circle[dot];
\draw[fill] (G2) circle[dot];
\draw[fill] (I1) circle[dot];
\draw[fill] (I2) circle[dot];
\draw[fill] (I3) circle[dot];
\draw[fill] (I4) circle[dot];
\draw[fill] (A1) circle[dot];
\draw[fill] (A2) circle[dot];}
}
\end{tikzpicture}
\end{scaletikzpicturetowidth}
\end{center}

%% file: resil/agree.tex
\subsubsection{The Significance of Key Confirmation}

Given a LTKC of a party $X$, another threat to consider is whether 
this enables the attacker to launch an attack that violates
the authentication goals (rather than Secrecy). In our setting, this 
will enable the attacker to mount one of the M2M mismatch attacks.

\begin{table}
\input{resil/prot-mqvv}
\caption{\label{t:prot:mqvv} MQV-S: MQV with Signature Confirmation}
\end{table}

We give a concrete example based on the \emph{Menzes-Qu-Vanstone (MQV)}
protocol \cite{mqv}, a predecessor of the FHMQV. 
Table~\ref{t:prot:mqvv} 
shows the MQV extended by a signature to authenticate service and
location rather than by a key confirmation step. We call it MQV-S.
MQV-S guarantees Secrecy, KCI Resilience, and the authentication goals, 
but the latter no longer hold under the assumption of an AC-side 
LTKC. Our example shows that the following is possible: given a root
certificate compromise of some airline $\mcA_I$,
the attacker can realize the `Imp2GU with Mismatch 
Diversion' (c.f.~Fig.~\ref{fig:astates}(d)) against 
\emph{any} two TAGA instances that happen about the same time at any 
airport $\mcH$ that handles ACs of $\mcA_I$. 
The fact that the attack scales to \emph{any} two parallel TAGA instances
is made possible by integrating an identity mismatch. This in turn is
made possible by the \emph{unknown key share attacks}, which the MQV 
(without key confirmation) is known to be vulnerable against \cite{kaliski}. 
For this attack we assume that TAGA uses smartphones rather than
smartcards as mobile units.

\begin{figure*}
\input{figures/fig-attack-divert}
\caption{\label{fig:attack:mqv} Attack~\ref{att:agreem} Tripartite attack
 against the MQV-S protocol}
\end{figure*}

\begin{attack}[Tripartite attack against MQV-S]
\label{att:agreem}
Let $\mcA_I$ be an airline for which the attacker has gained the 
ability to obtain AC certificates for his own key pairs, 
e.g.\ by a root certificate
compromise. Let $\mcH$ be an airport that handles ACs of $\mcA$,
and hence, the GUs of $\mcH$ have the $\mcA_I$ root certificate 
installed. Let $A$ and $A'$ be ACs of airlines \emph{other than} $\mcA_I$. 
Let $G$ and $G'$ be  GUs at airport $\mcH$ for service $S$, and 
$S'$ respectively. Further, let these entitites be such that 
TAGA pairing takes place between $G$ and $A$ at parking position $L$,
and, at about the same time, between $G'$ and $A'$ at parking position 
$L'$. In preparation,
the attacker compromises the smartphone of the operator of $G$,
and that of the operator of $G'$ respectively.

The attacker can now proceed as shown in Fig.~\ref{fig:attack:mqv}.
The attack consists of three parts. 
(1)~On the left hand side, the interaction between $U_I$ and $G$ 
constitutes a standard impersonation attack to $G$ analogously to 
Attack~\ref{att:ltkc}. The attacker simply generates his own key pair 
$(W_I, w_I)$, for which he can obtain a valid AC certificate associated
with a fictive AC $I$, $\cert(I, W_I)$, signed by $\mcA_I$.
(2)~On the right hand side, the interaction between $U_I'$ and
$A'$ constitutes a `suppress attack': $U_I'$ does not
forward the TAGA request of $G'$ to $A'$ but substitutes another command
for it, e.g.\ a request to read the status of the AC controller.
Hence, $A'$ will never know that a TAGA session for $S$ is underway. 
(3)~In the middle, across the parking positions $L$ and $L'$
and mainly involving $U_I$, $A$, $U'_I$, and $G'$ the attacker mounts a
location and service mismatch attack. As a result, $A$ shares the 
session key $K_\GpA$ with $G'$, and assumes $G'$ agrees that $K_\GpA$
is shared with the AC of identity $A$ for service $S$ at $L$. It is
true that $A$ only shares the key with $G'$. However, $G'$'s view is that it
shares $K_\GpA$ with the AC of identity $I'$ for service $S'$ at $L'$ 
(while $I'$ does not even exist).

We now explain the middle part in more detail.
At $L'$, $U_I'$ waits for the first NFC tap with $G'$, and then relays the 
received certificate and ephemeral public key of $G'$ to $U_I$. 
At about the same time at $L$, $U_I$ waits for the first NFC tap with 
$G$, and sets aside the received location $L$ and service $S$ of $G$. 
$U_I$ then prepares a forged TAGA request to $A$ based on the 
relayed certificate and ephemeral public key of $G'$ but using the 
location and service of $G$. With the second NFC tap $A$ computes 
the session key $K_\GpA$, confirmation signature $\sig_A$, and the 
response message according to the usual procedure. 
The challenge for the
attacker is that he cannot simply relay $A$'s message back
to $G'$: $\sig_A$ confirms $S$ and $L$ rather than $S'$ and $L'$, which is
necessary to make $G'$ accept the message. 

Instead, the attacker relays $A$'s public keys $R_A$ and $W_A$ to 
$U'_{I}$ and there
integrates a step analogous to Kaliski's unknown key share attack 
\cite{kaliski:mqv}:
he computes an ephemeral public key $R_I'$ (without knowing $r_I'$), 
and a long-term private key $w_I'$ based on the values of $R_A$ and $W_A$.
He also computes the matching long-term public key $W_I'$ from
$w_I'$ as usual.
$R_I'$ and $W_I'$ are chosen such that when $G'$ computes the session 
key the result is equal (by the algebraic laws valid in the underlying 
field and group) to the session key that $G'$ would have computed when
presented with $R_A$ and $W_A$ instead: i.e.\ $K_\GpA$. Moreover, 
the attacker can obtain an AC certificate for $W_I'$ associated with
a fictive AC $I'$, $\cert(I', W_I')$,
signed by $\mcA_I$. This allows him to fabricate his own confirmation
signature $\sig_I'$, where he includes $S'$ and $L'$ as required for
the message to be accepted by $G'$.  
The attacker can now relay all messages secured by $K_\GpA$ 
between $G'$ and $A$, who will think they originate from an
authenticated peer of matching service and location. 
\end{attack}

While the practical realization of mismatch attacks is made difficult by
requiring two parallel TAGA processes, due to the birthday paradox it is 
not as unlikely to identify \emph{any two} parallel TAGA processes at 
an airport as one would intuitively expect. Moreover, the TAGA smartphones
of an airport could be widely compromised by means of a targeted malware, e.g.\ 
by a Trojan app aimed at airport staff.
When, as usual, a smartcard is employed as mobile unit the attack
is mainly theoretical. However, it highlights how subtle attacks
against authentication in this setting can be. 
Note that this attack is possible even when a detection measure against
Imp2GU attacks is in place based on the AC noticing a delay 
(c.f.\ Section~\ref{s:auth:measure:time}).

Altogether, we require the TAGA protocol to guarantee \emph{Resilience of
Authentication against LTKC}. 
Attack~\ref{att:agreem} demonstrates that this does not, in general, 
follow from KCI Resilience. However, when the protocol binds all 
authentication data to the session key in a key confirmation step 
then KCI Resilience
does imply Resilience of Authentication without necessitating further
proof \cite{basin:kcir}.

\removed{
the included in a key confirmation step
implication does hold when the protocol
includes a key confirmation step (including all relevant parameters) \cite{}.
This also brings the advantage wrt provable resilience: 
that the space of attack states to
explore gets smaller/merge and proving resilience narrows down to
considering fewer attack states.}

%% file: resil/prot-mqvv.tex
\begin{enumerate}
\item $G$ generates $(R_G, r_G)$\\
      $G$ sends $S, L, \cert(G, W_G), R_G$
\item $A$ receives and validates the message\\
      $A$ generates $(R_A, r_A)$\\
      $A$ computes $s_A := r_A + \mu(R_A) w_A$\\
      $A$ computes $S_G := R_G + \mu(R_G) W_G$\\
      $A$ computes $K := h s_A S_G$\\
      $A$ computes $\sig_A := \sign_{w_A}(R_A, S, L, \ssid_A)$ \\
      $A$ sends $\cert(A, W_A), R_A, \ssid_A, \sig_A$  
\item $G$ receives and validates the message\\
      $G$ computes $s_G := r_G + \mu(R_G) w_G$\\
      $G$ computes $S_A := R_A + \mu(R_A) W_A$\\
      $G$ computes $K := h s_G S_A$\\
      $G$ verifies $A$'s signature with $W_A$\\
      $G$ establishes the WLAN connection with $\ssid_A$ and $K$
\end{enumerate}

%% file: figures/fig-attack-divert.tex
\begin{center}
\begin{tikzpicture}[
nfc/.style={->,>=stealth,shorten >=0.025cm},
wlan/.style={->,>=stealth,shorten >=0.025cm,dashed},
dot/.style={circle,draw=black,fill=black, radius=0.025}]
{\scriptsize


\draw (-2,0) -- (-2,-5.075);
\draw (-2,-6.925) -- (-2,-14.5);
\draw (-2,0) node[below right] {$A$};

\draw (-6,0) -- (-6,-3.175); 
\draw (-6,-3.825) -- (-6,-8.5); \draw (-6,-9.5) -- (-6,-11);
\draw (-6,0) node[below right] {$U_I$};

\draw (-6,-11.5) -- (-6,-14.5);
\draw (-6,-14.5) node[above right] {$\AP_I$};

\draw (-9.5,0) -- (-9.5,-14.5);
\draw (-9.5,0) node[below right] {$G$};
\draw (2,0) -- (2,-10.725);
\draw (2,-12.28) -- (2,-14.5);
\draw (2,0) node[below left] {$G'$};

\draw (6,0) node[below left] {$U_I'$};
\draw (6,0) -- (6,-8.3); 
\draw (6,-9.7) -- (6,-11);
\draw (6,-11.5) -- (6,-14.5);
\draw (6,-14.5) node[above left] {$\AP_I'$};

\draw (8.5,0) -- (8.5,-14.5);
\draw (8.5,0) node[below left] {$A'$};

\draw[dotted,thick] (0,0) -- (0,-14.5);
\draw (0,0) node[below left] {$L$};
\draw (0,0) node[below right] {$L'$};

\coordinate (G1) at (-9,-1);
\coordinate (G2) at (-9.5,-2);

\coordinate (I1) at (-6,-2);
\coordinate (I2) at (-6,-2.5);
\coordinate (I3) at (-6,-3.5);
\coordinate (I4) at (-6,-4.5);
\coordinate (I5) at (-6,-7.5);
\coordinate (I6) at (-7,-9);
\coordinate (I7) at (-6,-8);

\coordinate (A1) at (-2,-4.5);
\coordinate (A2) at (-1,-6);
\coordinate (A3) at (-2,-7.5);

\coordinate (G'1) at (2,-1);
\coordinate (G'2) at (2,-2);
\coordinate (G'3) at (2,-10.5);

\coordinate (I'1) at (6,-2);
\coordinate (I'2) at (6,-2.5);
\coordinate (I'3) at (6,-8);
\coordinate (I'4) at (6,-10.5);

\draw[nfc] (G2) -- (I1);
\coordinate (G2I1) at (-7.75,-2);
\draw[nfc] (G'2) -- (I'1);
\coordinate (G'2I'1) at (4,-2);
\draw[wlan] (I'2) -- (I2);
\coordinate (I2I'2) at (0,-2.5);
\draw[nfc] (I4) -- (A1);
\coordinate (I4A1) at (-4,-4.5);
\draw[nfc] (A3) -- (I5);
\coordinate (I5A3) at (-4,-7.5);
\draw[wlan] (I7) -- (I'3);
\coordinate (I7I'3) at (0,-8);
\draw[nfc] (I'4) -- (G'3);
\coordinate (G'3I'4) at (4,-10.5);

\draw (-9.5,-1) node[right,draw] {Generate $(R_G, r_G)$};

\draw (G2I1) node[above] {$S, L, \cert(G, W_G), R_G$};

\draw (G'1) node[right,draw] {Generate $(R_G', r_G')$};
\draw (G'2I'1) node[above] {$S', L', \cert(G', W_G'), R_G'$};
\draw (I2I'2) node[above] {$\cert(G', W_G'), R_G'$}; 

\draw (I3) node[draw, align=left, text width=3.5cm] 
  {For NFC tap with $A$:\\
   Swap $G'$ keys for $G$ keys};

\draw (I4A1) node[above]  {$S, L, \cert(G', W_G'), R'_G$};

\draw (A2) node[left,align=left,draw,text width=3.5cm] {
  Generate $(R_A, r_A)$\\
  $s_A := r_A + \mu(R_A)w_A$\\
  $S_G' := R_G' + \mu(R_G')W_G'$\\
  $K_\GpA := h s_A S_G'$\\ 
  $\sig_A := \sign_{w_A}(R_A, R_G',$\\ \hfill $S, L, \ssid_A)$};

\draw (I5A3) node[above] {$\cert(A,W_A), R_A, \ssid_A, \sig_A$};

\draw (8,-9) node[draw,left,align=left] {
$R_I' := R_A + \mu(R_A)W_A - uP$, for some $u$\\
$w_I' := \mu(R_I')^{-1} u$; $W_I' := w_I' P$ \\
Generate $\cert(I', W_I')$ with private key of $\mcA_I$\\
$\sig_I' := \sign_{w_I'}(R_I',R_G',S',L',\ssid_I')$};

\draw (I7I'3) node[above] 
{$W_A$, $R_A$};

\draw (G'3I'4) node[above] 
{$\cert(I', W_I'), R_I', \ssid_I', \sig_I'$};

\draw (0.25,-11.5) node[draw,right,align=left] 
{$s_G' := r'_G + \mu(R_G')w_G'$ \\
 $S_I' := 
       R_A+\mu(R_A)W_A - uP$\\
  \hfill $+ \mu(R_I')(\mu(R_I')^{-1}uP$\\
\hfill $= R_A + \mu(R_A)W_A = S_A$\\
 $K_\GpA := hs_G'S_A$};


\draw (-6,-9) node[draw,align=left,text width = 3.5cm] {
  Relay $W_G$, $R_G$, $\ssid_A$ to $\AP_I$\\
  For NFC tap with $G$: $\sig_I :=$\\
    $\sign_{w_I}(R_I,R_G,S,L,\ssid_I)$
};

\draw[nfc] (-6,-10.5) -- (-9.5,-10.5);
\draw (-7.75,-10.5) node[above] {$\cert(I, W_I), R_I, \ssid_I, \sig_I$};

\draw(-9.5,-11.5) node[draw,right,align=left] {
$s_G := r_G + \mu(R_G)w_G$\\
$S_I := R_I + \mu(R_I)W_I$\\
$K_{GI} := h s_G S_I$
};

\draw(-6,-12.5) node[draw,right,align=left] {
Receive $W_G$, $R_G$, $\ssid_A$\\
$s_I := r_I + \mu(R_I)w_I$\\
$S_G := R_G + \mu(R_G)W_G$\\
$K_{GI} := h s_I S_G$
};

\draw[wlan] (-9.5,-12.5) -- (-6,-12.5);
\draw (-7.75,-12.5) node[above] {$\{\msg\}_{K_\GI}$};

\draw[wlan] (2,-13) -- (6,-13);
\draw (4,-13) node[above] {$\{\msg\}_{K_\GpA}$};

\draw[wlan] (6,-13.5) -- (-6,-13.5);
\draw (0,-13.5) node[above] {$\{\msg\}_{K_\GpA}$};

\draw[wlan] (-6,-14) -- (-2,-14);
\draw (-4,-14) node[above] {$\{\msg\}_{K_\GpA}$};


\draw[nfc] (6,-5) -- (8.5,-5);
\draw (7.25,-5) node[above] {Read Status};

\draw[nfc] (8.5,-5.5) -- (6,-5.5);
\draw (7.25,-5.5) node[above] {Return Status};

}
\end{tikzpicture}
\end{center}

%% file: resil/designs.tex
\subsubsection{Resilient Designs}

\begin{table}
\centering{
\begin{tabular}{|c|c|c||c|c|}
\hline
& N & M &  LTKC R. & Multi-Instance R.  \\
\hline \hline 
1 & any & 4/2EV & basic   &  ? \\
\hline
2 & any & PC/CR & airport-reliant  & ?  \\
\hline
3 & $\NGU$ & $M_T$ & full  &  yes \\ 
\hline
4 & $\NGU$ & $M_T$, 4EV & full &  yes\\
\hline 
\end{tabular} }
\vspace*{2ex}
\caption{\label{tab:authDesigns} Resilient Authenticated TAGA}
\end{table}

\removed{
Table~\ref{tab:authDesigns} shows how, when using the 
FHMQV\CLS\ as underlying protocol, we obtain secure and resilient
designs for a-TAGA. The table also records the notion of multi-instance 
resiliance inherited 
by the choice of $N$, and whether the crew of the aircraft have
means of checking whether it is indeed a certificate of their
aircraft's domain that will be used by the GU to establish the channel.
We call this `AL Control'.}


\begin{theorem}
\label{thm:resil}
For $i \in [1,4]$ let $T_i = ($FHMQV\CLS, $N_i, M_i)$ where $N_i$ and $M_i$
are as in Table~\ref{tab:authDesigns}.
For each $i \in [1,5]$, $T_i$ is secure and guarantees the
properties as shown in row $i$.
\end{theorem}

For TAGA designs (3) or (4) are most suitable. Our modular
approach makes them amenable to be proved formally within a symbolic
or cryptographic proof framework \cite{Abrial10,maurer:comp}.



\removed{
However, for other M2M settings 
other combinations might be better. For example, it might be
preferable to establish the key `over the air' and 
secure some safety-critical services by PC/CR/KCM. Based on 
the attack categories and their pre-conditions it is straightforward
to check whether all safety-critical attack states are sufficiently covered.}


%% file: related.tex
\section{Related Work}
\label{a:related}

\paragraph*{AKE Protocols and Resilience}
Resilience against LTKC has mainly been studied in the context
of foundational research on authenticated DH protocols. 
The threat of a key compromise 
impersonation (KCI) attack where an attacker impersonates another
party to the actor with the LTKC has  been considered early 
on in \cite{menzes:kcir}, and KCI resilience (KCIR) has  been 
identified as a desirable attribute for AKE protocols to guarantee
\cite{BWM:aDH}. It has also become standard to consider KCIR and
other advanced properties such as forward secrecy and unknown
key-share by state-of-the-art protocol checkers such as the
TAMARIN prover \cite{cas:aDHTamarin}.  

However, automatic symbolic analysis is still not amenable to
protocols that use multiplication in the DH group and addition
of exponents such as the  AKE protocols with implicit
authentication.  The MQV \cite{MQV:protocol}, and its
development to FHMQV  
\cite{sarrEtAl:EuroPKI09,sarrEtAl:africacrypt16} 
have arguably the best combination of resilience properties
while in particular the FHMQV also developed for efficiency.
The FHMQV comes with a security proof in a 
cryptographic security model \cite{sarrEtAl:EuroPKI09,
sarrEtAl:africacrypt16}. 

The theory of \emph{post-compromise} security has been advanced in
\cite{cas:postCom} and \cite{basin:kcir}. 
Following \cite{basin:kcir} our definition of LTKC resilience
considers all security goals of the protocol. Moreover, we
build on the result therein that LTKC resilience of agreement
properties follows from KCIR when the protocol includes
a key confirmation step. 
Our property eav-LTKC resilience seems to be new.
It is close to the well-known property of \emph{forward secrecy}
(e.g.\ \cite{BWM:aDH}. Although for practical protocols the two properties will
typically coincide, in general, they seem incomparable: it seems possible to
construct an artificial protocol for which the two notions will differ.

\paragraph*{NFC Related Attacks and Distance Bounding}
We have already discussed in Section~\ref{s:plain:attacks} how research 
has shown that the nominal guarantees of NFC can be 
overcome by an attacker. Distance bounding protocols have been
introduced to verify the physical proximity or location of
a device \cite{Brands:1994,mh13:db,Hancke:2005}. While there are 
now many designs out there and the EMV Contactless Specifications
for Payment Systems (Version 2.6, 2016) has included relay 
resistance protocols new attacks against them have also been
devised (e.g.\ \cite{Cremers:2012,slides}). 
To counter this the formal verification community has
also provided techniques to verify such protocols \cite{Cremers:2012}.
While this work has focused on the protocols themselves our work
considers how their guarantees can be used in a overall security
concept that employs them. 

\removed{
care has to it remains challenging to design a
protocol that  
. have been included standard  
 
While so far
they have to be designed with care to meet ensure that the
proximity cannot be forged after all. (e.g.\   

}

\removed{
\paragraph*{NFC}
\cite{HasBr:RFID06}

\cite{Kfir:2005} 
predict

skimming.
leech device can comm. with an ISO-14443 RFIG tag from a distance 
40-50 cm. low power

\cite{KirWool:UsenixSec06}
confirm the claim about the practicality of the leech device.

\cite{han05} fast digital communication between the leech and
the ghost. 50 m

HK05 distsnce bounding 

Eavesdropping: 
FK05 Finke, Kelter
ISo14443 reader and tag
read from 1-2m by large loop antenna on the same plane as
the reader ant the tag.

nominal range << skimming range << eavesdropping range

\paragraph{Distance bounding}

distance bounding formal: 
}

%% file: app-protocols.tex
\section{Key Establisment Protocols}
\label{app:prot:props}

\input{notation}

\paragraph*{Security Properties}
(captured as injective agreement between
the runs of the two parties \cite{lowe:auth97}),

We say a KE protocol guarantees to an AC $A$  

\noindent
\ldots \emph{Secrecy} if, whenever $A$ completes a run of the protocol,
apparently with GU $G$, then no party other than $A$ and $G$ can 
compute the session key. 

\noindent
\ldots \emph{Key Freshness} if, whenever $A$ completes a run of the
protocol, and computes $K$ as the session key, then there is at most 
one other run of the protocol in which $K$ is the resulting session key.

\noindent
\ldots \emph{Opposite Type} if, whenever $A$ completes a run of
the protocol, apparently with GU $G$, then $G$ is indeed a GU.

\noindent
\ldots \emph{Agreement on Peer} if, whenever $A$ 
completes a run of the protocol, apparently with GU $G$,
and computes $K$ as the session key, then $G$ has a unique run in which
$K$ is the resulting session key, and in this run $A$ is the apparent peer.

\noindent
\ldots \emph{Agreement on Service}
if, whenever $A$ completes a run of the protocol, apparently with GU $G$
and for carrying out service $S$,
and computes $K$ as the session key, then $G$ has a unique run in which
$K$ is the resulting session key, and in this run $S$ is the apparent 
service to be carried out.

\noindent
\emph{Agreement on Location}
if, whenever $A$ completes a run of the protocol, apparently with GU
$G$, then $G$ is at the same location as $A$.

\noindent
The guarantees to GU $G$ are defined analogously. We say a protocol
guarantees security property $X$ when it guarantees $X$ to both
parties, ACs and GUs.

%% file: notation.tex
\paragraph*{Diffie-Hellman Protocols}

(Authenticated) Diffie-Hellman (DH) protocols assume a cyclic 
group $G$ of prime order $n$, and a 
generator $P$ of $G$ such that the decisional Diffie-Hellman 
problem is hard in $G$. The domain parameters $G, n,$ and $P$ can be fixed
or sent as part of the first message.
We use small letters to denote elements of the field
$\mathbb{Z}_n^\ast$, and capital letters for elements of $G$.
A key pair in the protocols consists of a public key $T$, which is a group
element, and a private key $t$, which is an element of
the field $\mathbb{Z}_n^\ast$ such that $T = tP$. Given an entity $X$,
we denote their ephemeral key pair by $(R_X, r_X)$, and 
their long-term key pair by $(W_X, w_X)$ respectively. 
Group operations are written additively ($A + B$, or $cA$) consistent
with notation for elliptic curve cryptography. 
Let $H$ be a $l$-bit hash function, where $l = (\log_2 n)/2$.

Moreover, let $\KDF_1$ and $\KDF_2$ be key derivation functions, 
$\mac$ be a message authentication code, and $\sign$ 
a signature scheme. We write $m_1 || m_2$ for the concatenation
of two messages $m_1$ and $m_2$.




%% file: app-fuel.tex
\section{Ground Processes: Fuelling}
\label{app:fuel}
\removed{
We now give examples of the process flows for 
air conditioning and fuelling with secure M2M communication, and TAGA 
pairing integrated into the setup phase.
}
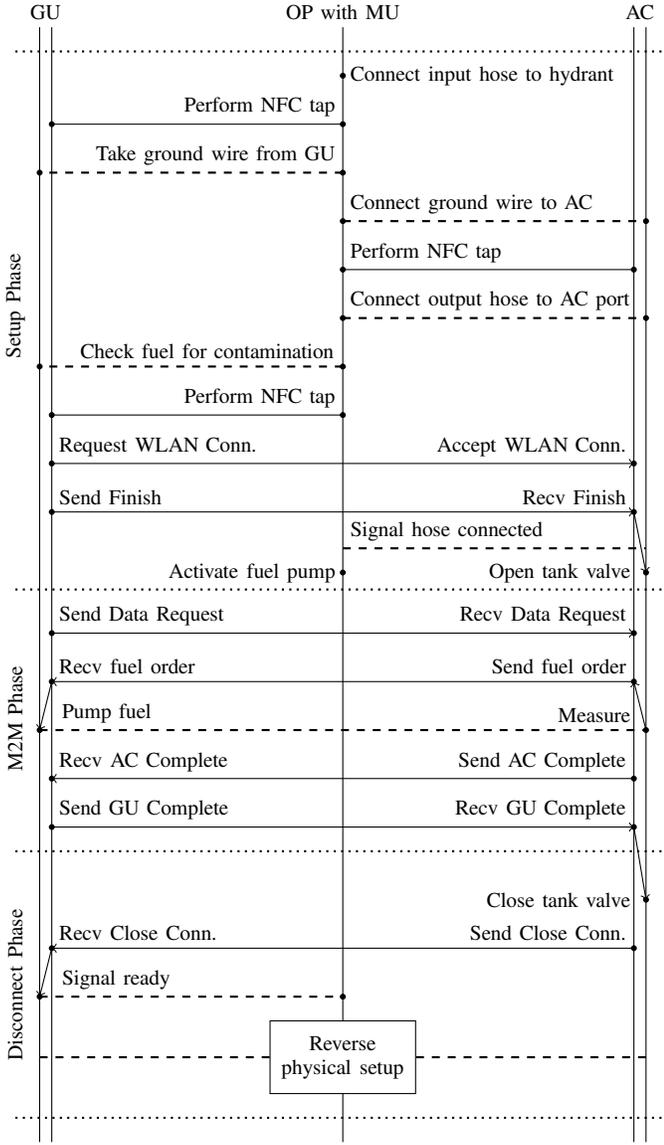
\begin{figure}
\centering{
\input{figures/fig-fuelling}}
\caption{\label{fig:process:fuelling} Process flow for fuelling}
\end{figure}

\paragraph*{Fuelling}
Fig.~\ref{fig:process:fuelling} shows a process for fuelling in a
setting where the fuel is obtained from an underground pipeline 
via a fuel truck. The fueller's first step is to connect the input 
fuel hose of the truck to a hydrant in the ground. Moreover, 
the fueller needs to connect the output fuel hose of the truck 
to the fuelling port of the AC. The fuelling ports are typically
positioned under the wings of the plane, and this step often involves 
that the fueller uses a lift integrated in the truck to take him up 
with the hose to one of the wings. In addition, safety measures are
carried out: before the AC is hooked up to the fuel hose 
a ground wire is laid from the 
fuel truck to the AC, and before the fuelling starts the fuel is
checked for contamination. 
TAGA pairing is integrated as follows: The first and second 
NFC taps are aligned with connecting up the ground wire, and the third 
tap is carried out after the fuel hose is connected up. 
Finally, the fueller activates the fuel
pump. At the AC, the pilot (or automatic control) waits until the 
secure channel is established and he has the okay from the fueller
(or the fuel truck via the secure channel) that the fuel hose is connected up. 

At the AC, the first action of the M2M phase is that the pilot 
(or automatic control) opens up the valves of the tanks, and activates the 
automatic fuel system. 
The M2M process makes use of the fact that most ACs already have an 
automated fuelling system: given a specified amount of fuel, the 
fuelling system distributes incoming fuel automatically into the various 
sections of the tank, monitors the amount of fuel already
received, and automatically shuts the valves of the tank when the 
specified amount has been reached. As usual backflow will stop 
the fuel pump of the truck. During the M2M phase the AC can 
communicate several fuel orders to the fuel truck so that the fuel
can automatically and precisely be topped up according to an increase in
the weight of the plane. When the final weight is known and the final 
fuel amount reached the AC notifies the fuel truck that the
service is complete. After an analogous reply by the truck the service
enters the disconnect phase. In the latter the communication session
is closed and the physical setup is reversed.

\paragraph*{Fuelling with Two Trucks}
Large ACs such as the A380 usually employ two fuel trucks to fuel
from the left and right wing in parallel. Then two parallel sessions
of the above process must be run. The 
following synchronization point between the two fuelling sessions
is required: the pilot (or automatic control) only opens the valves of the 
tank system after two secure channels are established and
both fuellers (or fuel trucks) have confirmed that the physical 
connection of the fuel hose on their side is ready. 


%% file: plain-2/app-plain.tex
\section{Relating to Section~\ref{s:plain}}

\begin{definition}
\label{def:op:welld}
We assume that each reader displays the status of the tap it
received (i.e.\ initial, mid, final), and that
the operator's actions are well-defined in the following sense:
\begin{enumerate}
\item The operator taps an AC only as part of the mid tap of a TAGA session.
\item If during a TAGA walk the operator observes that the display
     of the GU or AC reader does not confirm the action she expected
     to carry out (e.g.\ the tap is (repeatedly) unsuccessful
     or the GU reader
     returns `first tap' while the operator expected
     this to be the final tap)
     then she will reset the TAGA controller
     at the GU before carrying out any further taps.
\end{enumerate}
\end{definition}

\begin{proof}[Proof of Theorem~\ref{thm:plain:secureSetup}]
This is straightforward from the definitions.
\removed{
Guarantee for GU: When setup phase is complete in the view of $A$
then by (G-C) $G$ has
a connection $(Gini, Gfin, K, S)$. Then by Security by TAGA Channel
there has been a corresponding TAGA walk to $\AC(L)$, and $G$ shares
a secure channel with $\AC(L)$. By (PT1) $G$ is also physically
connected to $A$. By (Ph) $G$ must be at one of the positions.

Guarantee for AC: When setup phase is complete in the view of $A$
then by (A-C) $A$ has $n$ (confirmed) connections for service $S$. Then
by Secure TAGA there has been a corresponding TAGA walk to $A$ from some $G$
and $G$ has corresponding connection and match with walk,
$A$ shares a secure channel with that $G$. Hence, $G$ must be of
type $S$. Only that $G$ could have sent confirmation. Hence, by
(PT1) $A$ is physically connected to $G$. By (Ph) $G$ must be at one of
the positions. By (F3) and (F4) these are the only channels and connections.
}
\end{proof}

\begin{proof}[Proof of Theorem~\ref{thm:plain:secureTAGA}]
By Channel Authenticity to $G$, there is some OP with a session
$(\Wini, \Wmid, \Wfin)$ and some AC with a 
partial session $(\Amid, S', K', u)$ such that
$\Wini = \Gini$, $\Wmid = \Amid$, and $\Wfin = \Gfin$, and
$\msg(\Gini) = \msgr(\Amid)$ and $\msgs(\Amid) = \msg(\Gfin)$.
Hence, $S = S'$, and $G$ and $A$ will compute $K$, and $K'$ respectively,
based on the same DH public keys. Hence,
we also have $K = K'$, $K$ is fresh and
only known to $A$ and $G$.

By Channel Authenticity to $A$, there is some OP with a partial
session $(\Wini, \Wmid, \ldots)$ and some GU $G$ has
a partial session $(\Gini, \ldots)$ such that
$\Wini = \Gini$, $\Wmid = \Amid$, and $\msg(\Gini) = \msgr(\Amid)$.
Since $A$ has obtained the Finish message someone other than $A$
knows the key $K$. By secrecy under passive adversaries only
$G$ can know the key. 
Hence, $G$ has sent the Finish message and must
have a complete session $(x, y, z, K)$. By freshness $G$ does not use the
same key twice, and  
hence $x = S$, and $y = \Wini$.
By Authenticity to $G$ there must be a TAGA walk $(\Wini, \Wmid',
\Wfin')$ such that $z = \Wfin$. Then clearly also $\Wmid\ = \Wmid'$.
Hence, matching OP session also exists.
\end{proof}

\input{plain-2/conds}

\subsubsection*{Proofs of Section~\ref{s:plain:resil}}

\begin{proof}[Proof of Theorem~\ref{thm:plain:resil}]
To undermine a TAGA instance an attacker has to compromise channel 
authenticity. To this end the attacker
either needs to undermine authenticity of the card, compromise      
the local GU operator, or infiltrate in person the secure zone
during turnaround. 
For N-OP: Even if the attacker obtains a public/private key pair 
and can prepare a counterfeit card, they will still need to swap 
it with the operators card. For N-GU: Even if they obtain a master 
key (and update ability) they still need to swap the card within 
the GU. 
\end{proof}

%% file: plain-2/conds.tex
\subsubsection*{Proof of Theorem~\ref{thm:plain:N}}

Theorem~\ref{thm:plain:N} follows from Lemma~\ref{lem:plain:auth}
and Lemma~\ref{lem:plain:N} below. The first proves that when
a NFC system guarantees certain properties to be defined below
then under the measure `secure proximate zone' it will guarantee
channel authenticity. The second lemma then argues that 
$\NGU$ and $\NOP$ satisfy these properties. Before proving the
lemmas we need more definitions and intermediary facts. 


\paragraph*{More definitions}
If a reader $R$ has a NFC ssn $\Rx$, and a card $C$ has a NFC ssn $\Cy$,
and $R$ and $C$ communicate with each other during these NFC ssns 
then we write this by $\nfcl{\Rx}{\Cy}$.

A \emph{(TAGA) session of a card $C$} is a tuple $(\Cini, \Cmid, \Cfin)$
where $\Cini$, $\Cmid$, $\Cfin$ are successful NFC ssns of $C$
such that (1)~$\Cini < \Cmid < \Cfin$, and 
(2)~$\msgr(\Cini) = \msgs(\Cmid)$, and
$\msgr(\Cmid) = \msgs(\Cfin)$. 

\paragraph*{Between OP taps and reader NFC ssns and sessions}
Due to `secure proximate zone' and distance bounding we
can infer operator taps from NFC ssns of the GU and AC reader, 
and vice versa:

\begin{proposition}
\label{prop:nfc}
Let $N$ be a NFC system with distance bounding, and assume
`secure proximate zone' holds.
\begin{enumerate}
\item
If a GU or AC reader $R$ has a NFC ssn $\Rx$ and $\partner(\Rx)$ is
proximate during $\Rx$ then there is an operator tap $\Wx$ such that
\begin{enumerate}
\item $\card(\Wx) = \partner(\Rx)$, and
\item $\tapl{\Wx}{\Rx}$.
\end{enumerate}
\item
If there is an operator tap $\Wx$ and $\reader(\Wx)$ is a GU or
AC reader $R$ then $R$ has a NFC ssn $\Rx$ such that
\begin{enumerate}
\item $\card(\Wx) = \partner(\Rx)$, and
\item $\tapl{\Wx}{\Rx}$.
\item If $R$ is a GU reader then $C$ is authentic.
\end{enumerate}
\end{enumerate}
\end{proposition}

\begin{proof}
(1)~This follows from `secure proximate zone'.

(2)~Since the operator visually verifies that the reader has
successfully completed an NFC exchange (e.g.~green light) $\Rx$ exists.
By distance bounding and 
`secure proximate zone' the exchange must have indeed been carried
out with the operator's card. If the reader is a GU the GU will
check authenticity of the card during the exchange, and the
NFC exchange will fail if it is not (e.g.~red light).
\end{proof}

\begin{proposition}[Taps and reader sessions]
\label{fact:tapsWalk}
\begin{enumerate}
\item If there are operator taps $\Wx$ and $\Wy$ and a GU
      session $(S, \Gini, \Gfin, \ldots)$ such that
      $\tapl{\Wx}{\Gini}$ and $\tapl{\Wy}{\Gfin}$ then
      there must be an operator session $(\Wx, \Wmid, \Wy)$ for
      some $\Wmid$.

\item If there is an operator tap $\Wac$ and an AC session
      $(\Amid, \ldots)$ such that $\tapl{\Wac}{\Amid}$ then
      there must be an operator session
      $(\Wini, \Wac, \ldots)$ for some $\Wini$.
\end{enumerate}
\end{proposition}

\begin{proof}
(1) By Def.~\ref{def:op:welld}(2) the operator will immediately
 abort the GU session if $\Wx$
wasn't the first tap of a TAGA session $(\Wx, \ldots)$. Since $G$
receives the last NFC ssn the latter must be true.
By definition of TAGA session of OP, clause (5), the next
tap the operator will carry out is with  the AC. By
 Def.~\ref{def:op:welld}(2) the tap is
either successful or the operator will immediately abort
$G$'s session. Since $G$ receives the last NFC ssn the former must
be true. But then $\Wmid$ exists as required.

(2) This follows from Def.~\ref{def:op:welld}(1) .
\end{proof}

\paragraph*{Guarantees of the NFC System}
The following are guarantees that the NFC system must provide to
the GU, card, and AC.

\begin{definition}[Guarantees to GU]
\label{def:N:GU}
Let $N$ be a NFC system.
\item
We say $N$ guarantees \emph{`authenticity of initial tap to a GU $G$'} if,
whenever $G$ has a session $(S, \Gini, \ldots)$ 
then there is an \emph{authentic} card $C$ with a session 
$(\Cini, \ldots)$ such that 
\begin{enumerate}
\item $\nfcl{\Gini}{\Cini}$ and $\msg(\Gini) = \msg(\Cini)$, and
\item $C$ is proximate to $G$ during $\Cini$.
\end{enumerate}
\item
We say $N$ guarantees \emph{`authenticity of final tap to 
a GU $G$'} if, whenever $G$ has a session $(S, \Gini, \Gfin, \ldots)$ 
then there is an authentic card 
$C$ with a session $(\Cini, \Cmid, \Cfin)$ such that 
\begin{enumerate}
\item $\nfcl{\Gini}{\Cini}$ and $\msg(\Gini) = \msg(\Cini)$,
\item $\nfcl{\Gfin}{\Cfin}$ and $\msg(\Gfin) = \msg(\Cfin)$, and
\item $C$ is proximate to $G$ during $\Cini$ and $\Cfin$.   
\end{enumerate}
\end{definition}

\begin{definition}[Guarantees to Card]
\label{def:N:Card}
Let $N$ be a NFC system.
\removed{
\item
We say $N$ guarantees
\emph{`integrity of a $N$ card $C$'} if, whenever $C$ has a session
$(\Cini, \Cmid, \Cfin)$ then 
\begin{enumerate}
\item $\msgr(\Cini) = \msgs(\Cmid)$, and
\item $\msgr(\Cmid) = \msgs(\Cfin)$. 
\end{enumerate}
And similarly for partial sessions.}
\item
We say $N$ guarantees
\emph{`authenticity of initial tap to a card $C$'} if, whenever $C$
has a session $(\Cini, \ldots)$ then there is some GU $G$ with
a session $(S, \Gini, \ldots)$ such hat 
\begin{enumerate}
\item $\nfcl{\Gini}{\Cini}$, and $\msg(\Gini) = \msg(\Cini)$, and
\item $C$ is proximate to $G$ during $\Cini$.
\end{enumerate}
\item
We say $N$ guarantees
\emph{`authenticity of mid tap with AC to a card $C$} if, whenever
$C$ has a session $(\Cini, \Cmid, \ldots)$ and $\partner(\Cmid)$ is
an AC $A$ then $A$ has a session $(\Amid, \ldots)$ such that
$\nfcl{\Amid}{\Cmid}$ and $\msg(\Amid) = \msg(\Cmid)$. 
\removed{
\item
We say $N$ guarantees
\emph{`authenticity of final tap to a $N$ card $C$'} if, whenever
$C$ has a session $(\Cini, \Cmid, \Cfin)$ and `authenticity
of initial tap' is guaranteed to $C$ then there is some GU $G$ 
with a session $(S,\Gini,\Gfin, \ldots)$ such that 
\begin{enumerate}
\item $\nfcl{\Cini}{\Gini}$, 
\item $\nfcl{\Cfin}{\Gfin}$ and $\msgs(\Cfin) = \msgr(Gfin)$
\end{enumerate}
}
\end{definition}

\begin{definition}[Guarantees to AC]
\label{def:N:AC}
Let $N$ be an NFC system.
\item
We say $N$ guarantees \emph{`proximity of NFC peer to an AC $A$'} if,
whenever an AC $A$ has a session $(\Amid, \ldots)$ 
then $\partner(\Amid)$ is proximate to $A$ during $\Amid$.
\item
We say $N$ guarantees \emph{`authenticity of mid tap with an authentic 
card to an AC $A$'}, if
whenever $A$ has a session $(\Amid, \ldots)$ and $\partner(\Amid)$
is an authentic card $C$ then $C$ has a session 
$(\Cini, \Cmid, \ldots)$ such that 
$\nfcl{\Amid}{\Cmid}$ and $\msg(\Amid) = \msg(\Cmid)$.
\end{definition}

\removed{
  i.e.\ presses abort button after display information.
\begin{enumerate}
\item When an operator performs a tap $\Wx$ with the reader of a
GU $G$ and $G$'s reader displays that a new TAGA session is initialized then 
either the operator is on a TAGA walk/has a TAGA session $(\Wx, \ldots)$
or he will abort $G$'s session before carrying out any further taps with $G$.

\item When an operator is on a TAGA walk/session $(\Wini, \ldots)$ 
and carries out the mid tap with the reader of an AC $A$ then either
the tap is successful (e.g.\ green light) or the operator aborts 
the session of $\reader(\Wini)$ before carrying out any further
taps with that GU.

\end{enumerate} 
}

\removed{
\paragraph*{Operator Taps and Reader Sessions}

\begin{definition}
\label{def:op:welld}
We assume that each reader displays the status of the tap it
received (i.e.\ initial, mid, final), and that
the operator's actions are well-defined in the following sense:
\begin{enumerate}
\item The operator taps an AC only as part of the mid tap of a TAGA session.
\item If during a TAGA walk the operator observes that the display
     of the GU or AC reader does not confirm the action she expected
     to carry out (e.g.\ the tap is (repeatedly) unsuccessful  
     or the GU reader
     returns `first tap' while the operator expected
     this to be the final tap)  
     then she will reset the TAGA controller
     at the GU before carrying out any further taps.
\end{enumerate}
\end{definition}

\begin{proposition}[Taps and reader sessions]
\label{fact:tapsWalk}
\begin{enumerate}
\item If there are operator taps $\Wx$ and $\Wy$ and a GU
      session $(S, \Gini, \Gfin, \ldots)$ such that
      $\tapl{\Wx}{\Gini}$ and $\tapl{\Wy}{\Gfin}$ then
      there must be an operator session $(\Wx, \Wmid, \Wy)$ for
      some $\Wmid$.

\item If there is an operator tap $\Wac$ and an AC session 
      $(\Amid, \ldots)$ such that $\tapl{\Wac}{\Amid}$ then 
      there must be an operator session 
      $(\Wini, \Wac, \ldots)$ for some $\Wini$.
\end{enumerate}
\end{propositon}

\begin{proof}
(1) By Def.~\ref{def:op:welld}(2) the operator will immediately
 abort the GU session if $\Wx$
wasn't the first tap of a TAGA session $(\Wx, \ldots)$. Since $G$
receives the last NFC ssn the latter must be true. 
By definition of TAGA session of OP, clause (5), the next
tap the operator will carry out is with  the AC. By
 Def.~\ref{def:op:welld}(2) the tap is
either successful or the operator will immediately abort 
$G$'s session. Since $G$ receives the last NFC ssn the former must
be true. But then $\Wmid$ exists as required. 

(2) This follows from Def.~\ref{def:op:welld}(1) .
\end{proof}
}

\begin{lemma}
\label{lem:plain:auth}
Let $N$ be an NFC system, and $M$ a set of process measures
such that $N$ implements distance bounding and the 
guarantees to GU, card, and AC of 
Def.~\ref{def:N:GU}, \ref{def:N:Card}, \ref{def:N:AC}, 
and $M$ guarantees `secure proximate zone'. 
Then $N$ under $M$ guarantees Channel Authenticity.
\end{lemma}

\begin{proof}
\emph{Guarantee for GU:}
Assume $G$ has a session $(S, \Gini, \Gfin, \ldots)$. Then
by `authenticity of final tap to $G$' there is an
\emph{authentic} card $C$ with a session $(\Cini, \Cmid, \Cfin)$
such that the NFC ssns and exchanged messages are matching, i.e.: 
$\nfcl{\Gini}{\Cini}$, $\nfcl{\Gfin}{\Cfin}$, $\msg(\Gini) = \msg(\Cini)$, 
and $\msg(\Gfin) = \msg(\Cfin)$. Moreover, $C$ is proximate to $G$ 
during $\Cini$ and $\Cfin$.

By Prop.~\ref{prop:nfc}(1) there are successful operator taps $\Wini$, $\Wfin$
such that $\card(\Wini) = C$ and $\tapl{\Wini}{\Gini}$, and
$\card(\Wfin) = C$ and $\tapl{\Wfin}{\Gfin}$ respectively. Then by
Prop.~\ref{fact:tapsWalk}(1) there is an operator with a session 
$(\Wini, \Wmid, \Wfin)$ for some $\Wmid$.

Further, by definition of OP session we obtain $\reader(\Wmid) = A$ for
some AC $A$ and $\card(\Wmid) = C$. Hence, by Prop~\ref{prop:nfc}(2) 
$A$ has a successful NFC ssn $\Aac$ such that $\partner(\Aac) = C$ and 
$\tapl{\Wmid}{\Aac}$.
Let $\Cac$ be defined by $\nfcl{\Aac}{\Cac}$. 
$\Cac$ takes place after $\Gini$ (and hence $\Cini$) and before $\Gfin$ 
(and hence $\Cfin$). Hence, we conclude that $\Cac = \Cmid$.
Then by `authenticity of mid tap with AC to $C$' $A$ must have
a session $(\Amid, \ldots)$ and  $\msg(\Amid) = \msg(\Cmid)$. 
Finally, by definition of session of a card it is straightforward to check that
the messages of the $A$ session match those of the $G$ session.

\emph{Guarantee for AC:}
Assume $A$ has a session $(\Amid, \ldots)$. Then by
`proximity of NFC peer to $A$' $\partner(\Amid)$, say $C$, is proximate.
Then by Prop.~\ref{prop:nfc}(1) there is a successful operator tap $\Wac$
such that $\tapl{\Wac}{\Amid}$ and $\card(\Wac) = C$.
Then by Prop.~\ref{fact:tapsWalk}(2) there is an operator with
session $(\Winiac, \Wac, \ldots)$ for some $\Winiac$. By definition of
OP session
$\reader(\Winiac)$ is a GU $\Gac$, and $\card(\Winiac) = C$.

By Prop.~\ref{prop:nfc}(2) $\Gac$ has a successful NFC ssn $\Giniac$ such that
$\partner(\Giniac) = C$, $\tapl{\Winiac}{\Giniac}$, and $C$ is authentic.
Let $\Ciniac$ be defined by $\nfcl{\Giniac}{\Ciniac}$.

Since $C$ is authentic by `authenticity of mid tap with authentic
card to $A$' $C$ has a session $(\Cini, \Cmid, \ldots)$ such 
that $\nfcl{\Amid}{\Cmid}$ and $\msg(\Amid) = \msg(\Cmid)$.
By `authenticity of first tap to $C$' there is some GU $G$ with a session 
$(S, \Gini, \ldots)$ such that $\nfcl{\Gini}{\Cini}$,
$\msg(\Gini) = \msg(\Cini)$, and $C$ is proximate to $G$ during
$\Cini$. By definition of session of a card we also obtain that 
the messages of $\Amid$ and $\Gini$ are matching. 

Moreover, by Prop.~\ref{prop:nfc}(1) there is a successful operator
tap $\Wini$ such that $\card(\Wini) = C$ and $\tapl{\Wini}{\Gini}$.
Hence, it only remains to show that $\Wini = \Winiac$. 
To the contrary assume this not to be the case.
Then by definition of OP session, clause~(5) $\Wini < \Winiac$, and hence 
$\Cini < \Ciniac$. But this leads to a contradication:
since $\tapl{\Wac}{\Cmid}$ and $\Cini$ immediately precedes
$\Cmid$ on $C$, we must have $\Ciniac < \Cini$.
\end{proof}

\begin{lemma}
\label{lem:plain:N}
$\NGU$ and $\NOP$ implement distance bounding and meet the guarantees of 
Def.~\ref{def:N:GU}, \ref{def:N:Card}, \ref{def:N:AC}.
\end{lemma}

\begin{proof}
We only provide the argument for $\NGU$. The proof for $\NOP$
is analogous just based on an argument based in public key 
authentication instead. Distance bounding is given as explained
in Section~\ref{s:plain:sols}.

Guarantees to GU:
At the first tap the GU reader runs an AKE protocol based on $K_\GC$
shared between the GU and the authentic card. Moreover, it will run a distance
bounding protocol to ensure that the co-owner of $K_\GC$ is close by.
The communication is secured by the established session key, say $K$.
At the third tap the GU reader runs an AKE protocol based on $K$, 
including distance bounding to ensure the co-owner of $K$ is close by.
This binds the third tap to the first as required.

Guarantees to card: At the initial tap the card runs an AKE protocol
based on $K_\GC$ shared only with the GU. Since the GU will
only communicate with a card when it can successfully prove its
proximity, the card knows that if it has a successful first
NFC ssn then it communicates with a proximate GU. Again the
communication itself is secured by the established session key.

The mid tap will be protected by a key established via the basic DH exchange. 
Hence, if the mid tap is indeed carried out with an authentic AC
then by Ass.~\ref{ass:nfc} the communication is secured and the
messages are matching.

Guarantees to AC: 
`Proximity': This is guaranteed by the distance bounding protocol,
and because the communication is secured by DH such that the owner of
the DH public key is proximate.
Authenticity when the tap is carried out with an authentic card
follows because the communication is secured, and the 
card only accepts unauthenticated communication as part of its mid tap.
\end{proof}

%% file: auth-2/app-auth.tex
\section{Relating to Section~\ref{s:auth}}
\label{app:auth}

\begin{proof}[Proof of Theorem~\ref{thm:setupAuth}]
This is straightforward to check.
\end{proof}

\begin{proof}[Proof of Theorem~\ref{thm:tagaAuth}]
This is straightforward from the definitions.
\end{proof}

\begin{figure}
\input{figures/att-ltkc}
\caption{\label{fig:att:ltkc}
Impersonation to GU with LTKC of any AC}
\end{figure}
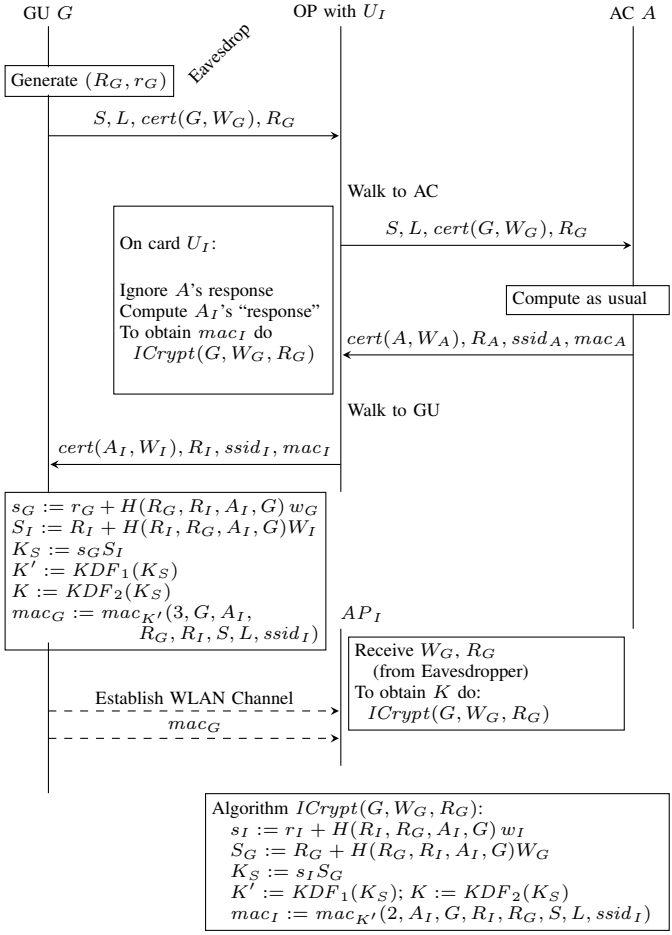

%% file: resil/app-resil.tex

\begin{proof}[Proof of Theorem~\ref{thm:resil:basic}]
Assume a TAGA process between $G_L$ and $A_L$ at $L$. 

If $G_L$
has a session $(S, A, K)$ then by key confirmation someone knows the
key. This must be an honest party or the attacker. 
If the attacker does not know $K$ then agreement is guaranteed, and
hence, $A$ has a matching session at the same location as required. 
If the attacker knows $K$ then by Lemma~\ref{lem:prot:resil} 
$A$ must have a LTKC. 
By `agreement of AC domain' $A$ must be of the same airline and 
type as $A_L$. 

If $A_L$ has a session $(S, G, K, c)$ then by key confirmation
someone knows the key. The first case is as above. If the attacker 
knows $K$ then by Lemma~\ref{lem:prot:resil} $G$ must have a LTKC. By
certificate verification $G$ is a GU of the local airport.
\end{proof}

\begin{proof}[Proof of Theorem~\ref{thm:resil:detect}]
Assume a TAGA process between $G_L$ and $A_L$ at $L$.

If $G_L$ has a session $(S, A, K)$ and the attacker knows
$K$ then the attack will be detected and the session 
aborted, or $A_L$ has a TAGA session $(S, G', K', ?)$ for some 
$G'$, $K'$ such that the attacker knows $K'$. In the latter
case by Lemma~\ref{lem:prot:resil} $G'$ must have a LTKC. 
By certificate verification $G'$ is a 
GU of the local airport. If the attacker does not know $K$ then the 
usual argument applies.

If $A_L$ has a TAGA session $(S, G, K, c)$ and the attacker knows
$K$ then by Lemma~\ref{lem:prot:resil} $G$ must have a LTKC. 
By certificate verification $G$ is a GU of the local airport.
Otherwise the usual argument applies.
\end{proof}

\begin{proof}[Proof of Lemma~\ref{lem:resil:pccr}]
Let $G$ be a GU with a session $(S, A, K)$ at parking slot $L$, and assume
$G$ has sent the challenge. Say $G$ obtains the correct answer.
We must show that the attacker does not know $K$ unless
$\AC(L)$ has a compromised session on $S$. 

Let $A_L = \AC(L)$.
Only $A_L$ could have known the nonce, and sent it on. Hence, $A_L$ 
has a session $(S, G', K', c)$ for some $G'$, $K'$ during which it
has responded to the challenge. 
If the attacker knows $K'$ then nothing has to be proved.
If the attacker does not know $K'$ then $A_L$ has the session with 
an honest party, who would not forward $A_L$'s reply to $G$. 
Hence $K = K'$. And hence, the attacker does not know $K$ as required.
\end{proof}